\documentclass{article}
\usepackage{amssymb}
\usepackage{amsmath}
\usepackage{mathrsfs}
\usepackage{amsthm}
\usepackage[latin1]{inputenc}
\usepackage{authblk}
\usepackage{verbatim} 
\usepackage[pdftex]{graphicx}
\usepackage{float}
\usepackage{capt-of} 
\usepackage{multirow}
\usepackage{color}
\usepackage[comma, authoryear, numbers]{natbib}
\usepackage{hyperref}
\usepackage{color}
\usepackage{bbm}
\usepackage{colortbl}
\usepackage[margin=1in]{geometry}

\newcommand{\indicator}[1]{\mathbbm{1}_{\left\{ {#1} \right\} }}

\newtheorem{thm}{Theorem}[section]
\newtheorem{prop}[thm]{Proposition}
\newtheorem{lemma}[thm]{Lemma}
\newtheorem{cor}[thm]{Corollary}

\theoremstyle{remark}
\newtheorem{rem}[thm]{Remark}

\theoremstyle{definition}

\newcommand{\I}{\ensuremath{\mathbb{I}}}

\newcommand{\less}{<}
\newcommand{\more}{>}

\newcommand{\Hx}{\mathcal{H}}
\newcommand{\Qx}{\mathcal{Q}}
\newcommand{\N}{\mathbb{N}}
\newcommand{\Px}{\mathbb{P}}
\newcommand{\Ex}{\mathbb{E}}
\newcommand{\R}{\mathbb{R}}

\newcommand{\sbt}{\,\begin{picture}(-1,1)(-1,-2)\circle*{2}\end{picture}\ }  

  \definecolor{Red}{rgb}{1,0,0}
    
    \definecolor{DRed}{rgb}{0.8,0,0.4}
    
    \definecolor{Green}{rgb}{0.2,0.5,0.2}
    
    \definecolor{Blue}{rgb}{0,0,1}
    
    \definecolor{PaleGrey}{rgb}{0,0,0}
    
    \definecolor{Orangee}{rgb}{1,0.647,0}

    \definecolor{Gold}{rgb}{0,0,0}

		\definecolor{Brilliantrose}{rgb}{0,0,0}
		
		\definecolor{Lgreen}{rgb}{0.596,0.984,0.596}

\begin{document}

\title{A time-dependent Markovian model of a limit order book}
\author[1]{Jonathan A. Chávez-Casillas\thanks{jchavezc@uri.edu.\qquad\qquad  ORCiD:\ \ 0000-0002-8494-7538}}
\affil[1]{Department of Mathematics and Applied Mathematical Sciences,\newline University of Rhode Island}
\date{February 18th, 2022}

\maketitle

\abstract{This paper considers a Markovian model of a limit order book where time-dependent rates are allowed. With the objective of understanding the mechanisms through which a microscopic model of an orderbook can converge to more general diffusion than a Brownian motion with constant coefficient, a simple time-dependent model is proposed. The model considered here starts by describing the processes that govern the arrival of the different orders such as limit orders, market orders and cancellations. In this sense, this is a microscopic model rather than a ``mesoscopic'' model where the starting point is usually the point processes describing the times at which the price changes occur and aggregate in these all the information pertaining to the arrival of individual orders. Furthermore, several empirical studies are performed to shed some light into the validity of the modeling assumptions and to verify whether certain stocks satisfy the conditions for their price process to converge to a more complex diffusion.}

\bigskip
\noindent\textbf{keywords:}\quad Limit order book,\quad Price process,\quad Diffusion approximation,\quad Time-dependency

\section{Introduction}

The evolution of trading markets has progressed considerably in the last few decades. One of the main drivers of such evolution has been the creation and usage of fast-paced technological developments. In the past, liquidity was provided by the so-called market makers, which collected buy and sell orders from all market participants to set bid and ask quotes. While this traditional method was generally accepted, mainly, because of the lack of alternatives, it was also criticized due to the bias and questionable conflict of interest from the market maker. Nowadays, most exchanges use completely automated platforms called Electronic Communication Networks (ECN). These ECN enables a continuous double auction trading mechanism, which eliminates the need of a market maker or an intermediary that matches the opposite parties in a trade. The auction that the ECN manages is called ``continuous double'' since traders can submit orders in form of bids (i.e., buy orders) as well as asks (i.e., sell orders) at any point in time. ECNs increase significantly the speed of trading, taking only a few milliseconds from sending an order to its execution. 

A bid limit order (resp. ask limit order) specifies the quantity and the price at which a trader wants to buy (resp. sell) certain asset. The limit order book consists of all the collection of limit orders from every trader. Outstanding limit orders are stored in different queues inside the order book. These queues are ordered by the price and type (bid or ask). The difference in price between the lowest ask price and the highest bid price is called the spread.

The counterpart of limit orders are market orders, which allow traders to buy and sell at the best available price. While limit orders will not trigger an immediate transaction, market orders are immediately executed. In this sense, limit orders accumulate, create, and extend the size of queues at both sides of the LOB, while market orders remove limit orders from the best available price. Sometimes informed traders are associated with traders that place market orders, while uninformed traders are associated to the ones that place limit orders, but this goes against the fact that many of the most successful hedge funds make extensive use of limit orders (see \cite{Bouchaud2009}).
	
	In addition to limit orders and market orders, cancellation of limit orders is another common operation. The basic idea of a cancellation is that a trader is no longer willing to buy or sell at the specified price.  Cancellations account for a large fraction of the operations on an order book, partly due to the introduction and evolution of high frequency trading (see \cite{Harris2003}), in which the inter-arrival times of limit orders and cancellations, occur at a millisecond time scale (see \cite{ContLarrard2012}). 
	
	An important feature of a LOB is that traders can choose between submitting limit and market orders. The biggest advantage of limit orders is the possibility of matching better prices than the ones they can obtain with market orders, but as drawback, there is a risk of never being executed. Conversely, market orders never match at prices better than the best bid or the best ask, but the execution is certain and immediate. Usually, the bid-ask spread can be considered as a measure of how expensive is the certainty and immediacy of buying or selling the underlying asset.

From a modeling perspective, sometimes it is important to identify the different types of traders that are able to participate in the market. As stated in \cite{Foucault2005}, LOBs allowed traders to immediately obtain liquidity, but at the same time, they also allow other traders to supply liquidity to those who require it later. On the exchanges, most traders combine limit orders and market orders to create a trading strategy according to their needs and the current state of the order book. However, broadly, as detailed in \cite{Foucault2005}, traders with short-horizon strategies, as arbitrageurs, technical traders, and indexers, prefer to post market orders, while, traders with long-horizon strategies, as portfolio managers, place limit orders.

As pointed out in \cite{Gould2013}, there are many practical advantages in understanding LOB dynamics. Examples of these are: gaining clearer insight into how best to act in a given market situation (cf. \cite{Harris1996}), devising optimal order execution strategies (\cite{Obizhaeva2013}, \cite{Law2015}, \cite{Cartea2015}), minimizing the market impact (\cite{Eisler2012}); designing better electronic trading algorithms (\cite{Engle2006}), and assessing market stability (\cite{Kirilenko2011}).

Due to the complexity of a LOB, when interpreted as a dynamical system, any attempt to model a LOB requires considerable assumptions. One of the most important assumptions is how the order flow evolves in time. One of the seminal works in modeling a complete LOB in continuous time is the one proposed by \cite{Cont2010}. This is a zero-intelligence model, in which the order flows follow independent Poisson processes whose rate parameters depend on the type of order, distance to the bid-ask spread, and state of the order book. With this model, the authors try to understand how the frequency at which certain events occur depends on the state of the order book. Another important trait of this model is the use of a power-law distribution for the arrival of limit orders depending on their relative price, which fits well with empirical observations.

In an attempt to simplify the dynamics of the order book and provide estimates for the volatility in terms of the parameters governing the order flow, \cite{ContLarrard2012} propose a model that keeps track of the first level of the order book instead of the whole LOB. When there was a depletion in either side of this simplified I-level order book, the amounts of orders available at the next best prices were assumed to be drawn from a distribution $f$ on $\N^2$, and the spread was always kept constant at 1 tick. The authors' justification for this simplified framework was that many traders can only view the depths available at the best prices and, also, that the percentage of time, in liquid markets, that the spread is 1 is typically larger than $97\%$. In this model, arrival of limit, market and cancellation orders are modeled as independent Poisson processes. One of the advantages of this approach is the ability to estimate analytically, depending just on the Poisson processes' parameters and the depth distribution $f$, quantities of interest such as the volatility, the distribution of time until the next price change, the distribution and auto-correlation of price changes, and the conditional probability that the price moved in a specified direction given the current state of the order book.

As noted in \cite{chavez2019level}, one of main drawbacks of the model presented in \cite{ContLarrard2012} is the assumption that the limit orders, market orders and cancellations arrive as homogeneous Poisson processes. In an attempt to improve this assumption, in \cite{chavez2019level}, the authors assume that all the orders arrive as Non-homogeneous Poisson processes but in contrast to the model presented in \cite{ContLarrard2012}, the random clock driving the arrivals of these orders is not reset after a price change. In fact, this random clock signaling the arrival of new orders is never reset and instead of a reset, an assumption of periodicity is used. That is, it is assumed that every week, every day, or every given period of time, the intensities of the processes signaling the arrival of orders is the same. While this approach is suitable to analyze the behaviour of the orderbook at a mesoscopic time scale, the distributional properties of the inter arrival times between price changes and their consequences on the heavy-traffic approximation of the price process are not discussed, which is the the aim of the present study with some corresponding empirical data analysis.

By analyzing the distributional properties of the inter-arrival times and allowing the intensity of the point process dictating their occurrence to be time-dependent, it can be shown that the volatility in the long-run dynamics of the price process possess, on certain cases, a time dependent component. This is a new improvement in an attempt to reconcile the macro-price dynamics, as a more general process than merely a Brownian motion with constant volatility, of the orderbook when doing the modeling from the microscopic price formation time scale. Indeed, in the existing literature, most of the models that start from the modeling of how orders arrive to the macro-price formation conclude that the price process is a Brownian motion with constant volatility. The aim of the present paper is to shed some light into how it can be possible to extend such models in such a way that the limiting price process converges to a more general diffusion.

One of the aims of this paper is to analyze the different cases under which the price process when modeled from a microscopic scale, referring to a model that starts from considering the arrival of individual orders, rather than from a mesoscopic scale, referring to a model that considers as a baseline the arrival of changes in the price and aggregating all the information provided in the arrival of individual orders all the way up to the arrival of a price change. This entails more complex dynamics and there are few models in the literature that consider creating bottom-up estimates and models; and while one of the aims of the present work is to expand in the existing literature and to study empirically whether the assumptions of the model might be validated, it also wants to serve as a first step in creating more complex models that start at a microscopic scale and can provide conditions under which the long-run dynamics of the price process possess a more complex structure than a simple Brownian motion with constant volatility.

This paper is organized as follows. In section \ref{sec:description} a precise description of the Limit Order Book model is presented with its corresponding assumptions. In Section \ref{Sec:Properties} the distributional properties of the model and the long-run dynamics of the price process are analyzed depending on the tail behaviour of the time span between price changes. In section \ref{sec:empirical} the six study cases from this paper are introduced and analyzed. Finally, in section \ref{sec:conclusions} some conclusions are presented along with further research directions in an attempt to generalize the price process when modeling from the micro-price dynamics.

\section{Description of the model} \label{sec:description}

In this paper, a Level-1 Limit Order Book model is discussed using as a framework the model proposed in \cite{ContLarrard2012}. However, in contrast to such model, the point processes describing the arrivals of Limit orders are described by a time-dependent periodic rate proportional to the rate describing the arrival of Market orders plus Cancellations.

The next paragraphs will describe the modeling assumptions of the orderbook model as well as some of its basic dynamics. First, as it was mentioned before, only the best level at each side will be considered, together with their corresponding sizes. . This assumption is justified by the fact that most of the operations inside an orderbook are performed at the level-I, as pointed out in \cite{ContLarrard2012}. Also, in this sense, the orderbook is assumed to have a unitary tick size since for the stocks considered in section \ref{sec:empirical}, the spread spend more than $90\%$ within 1 tick.

A more controversial assumption which is meant to be relaxed in further research is a classical assumption in these models that the order volume is also unitary. Since the model presented is inspired by the original model presented in \cite{ContLarrard2012}, the same assumption is kept, but it is the author's interest to relax this assumption. However, contrasting to the original model and creating a much more realistic assumption, limit orders at the bid and ask sides of the book arrive independently according to  inhomogeneous Poisson  processes $L_t^b$ and $L_t^a$, with intensities  $\lambda^b_t$ and $\lambda_t^a$ respectively. Similarly, market orders plus cancellations at the bid and ask sides of the book arrive independently according to  inhomogeneous Poisson  processes  $M_t^b$ and $M_t^a$, with intensities $\mu^b_t$ and $\mu^a_t$ respectively. Notice that, since in this model limit orders increase the queue size and both market orders and cancellations decrease the queue size the latter are merged int one process, which is natural by our assumption of a unitary volume in all orders. Further, it is assumed that the processes $L_t^a, L_t^b, M_t^a$ and $M_t^b$ are all pairwise independent.

Finally, it will be described what happens every time there is a price change. Indeed, whenever the bid queue or the ask queue gets depleted, a price change occurs. If the bid side is depleted, the price decreases and if the ask queue is depleted, the price increases. After a price change occurs, the depth of the orderbook will be considered via the distribution of two new queue sizes. That is, after a price change, both the bid and the ask prices increase by one tick, and the size of both queues gets redrawn from some distribution $f\in\N^2$.

\section{Distributional Properties of the Inter-arrival times and the Price Process}\label{Sec:Properties}

For this model it will be assumed that after each price change, the clock signaling the arrivals of new orders or cancellation in the orderbook will be reset between the arrivals of these processes, so that if $\tau_n$, $n\geq1$ represents inter-arrival time between the $(n-1)-$th and $n-$th price change, then, the sequence $\{\tau_n\}_{n\geq1}$ becomes an independent sample.
 

In order to accurately describe the distribution of the inter-arrival time in between price changes, some general notation that will be used throughout will be introduced. Let $\{(a_n,b_n)\}_{n\geq1}$ be a sequence of independent random variables generated from the distribution $f\in\N^2$. These will represent the initial amount of orders at the bid and ask side of the order book, respectively, after the $n-$th price change. For $t\geq0$, let also, $\{X_{x_{n-1}}^{(n)}(t)\}_{n\geq1}$ and $\{Y_{y_{n-1}}^{(n)}(t)\}_{n\geq1}$ be two families of mutually independent one-dimensional time continuous random walks in the positive axis parametrized by their starting point, i.e. for any $n\in\N$, $X_{x_{n-1}}^{(n)}(0)=x_{n-1}$ and $Y_{y_{n-1}}^{(n)}(0)=y_{n-1}$, whose generators $\mathscr{L}_{t}^{(X,n)}$ and $\mathscr{L}_{t}^{(Y,n)}$ are given by

\begin{equation}\label{eqn:generatorXY}
\mathscr{L}^{(X,n)}_{t}(i,j)=\mathscr{L}_{t}^{(Y,n)}(i,j):=\left\{\begin{array}{ccl} 0&\text{if}& i=0,\ j\geq0,\\ \mu_t&\text{if}& 1\leq i,\ j=i-1,\\ \lambda_t&\text{if}& 1\leq i,\ j=i+1,\\-(\lambda_t+\mu_t)&\text{if}& 1\leq i,\ j=i,\\ 0&\text{if}& |i-j|\more1\end{array}\right.
\end{equation}

Note that $0$ is an absorbing state for any Markov chain with generator $\mathscr{L}^{(X,\sbt)}_{t}$ or $\mathscr{L}^{(Y,\sbt)}_{t}$. When a chain reaches the absorbing point $0$, one calls it extinction. 

Let $\sigma_{x_{n-1}}^{(X,n)}$ and  $\sigma_{y_{n-1}}^{(Y,n)}$ be the extinction times the independent Markov chains $X_{x_{n-1}}^{(n)}(t)$ and $Y_{{y_{n-1}}}^{(n)}(t)$. Further, set $\tau_0:=0$ and $\tau_n := \min\left(\sigma_{x_{n-1}}^{(n)},\sigma_{y_{n-1}}^{(n)}\right)$. 

The dynamics of the orderbook can be described as follows. Let $q_t=(q_t^b,q_t^a)$ be the amount of bid and ask orders at time $t$, with $(x_0,y_0)$ the initial amount of orders at each side of the book, that is, $q_0=(x_0,y_0)$. For $0\leq t\less\tau_1$, define $q_t:=(X_{x_0}^{(1)}(t),Y_{y_0}^{(1)}(t))$, so that, $\tau_1$ is the time at which the first price change occurs. At time $\tau_1$, both queues move one tick (to the right if the ask queue is depleted or to the left if the bid queue is depleted). The size of the ask and bid queue are then set to be $(x_2,y_2)$ (which are drawn from the distribution $f\in\N^2$), and $q_{\tau_1}$ is defined as $(x_1,y_1)$. Furthermore, for $\tau_1\leq t<\tau_2$, $q_t$ is defined as $(X_{x_1}^{(2)}(t),Y_{y_1}^{(2)}(t))$ and the process continues. To simplify notation, we denote by $\sigma_a^n$ and $\sigma_b^n$ the random variables $\sigma_{x_{n-1}}^{(X,n)}$ and $\sigma_{y_{n-1}}^{(Y,n)}$, respectively.

\subsection{Distribution of the inter-arrival time between price changes}

Because of the independence between the ask and the bid side of the book before the first price change, to analyze the distribution of $\tau_1$, it is enough to study one side of the orderbook, say the ask. In this case, an explicit formula for $\Px[\sigma_a^1\more t]$ is given by the following result whose proof is defered to Appendix \ref{sec:AppendixA}.

\begin{prop}\label{prop:distrsigma}
Let $q_t^a$ be the process defined above starting at $x$, that is, the Markov process with generator given by Equation \ref{eqn:generatorXY} such that $q_0^a=x$. Fix any $T\more0$ and let $\bar{u}(t,z):\R_+\times(\N\cup\{0\})\to\R$ be a bounded function such that $t\mapsto \bar{u}(t,x)$ is $C_1(\R)$ for all $x$ and $(t,x)\mapsto \partial\bar{u}(t,z)/\partial t$ is bounded and satisfies the conditions:

\begin{equation}\label{eqn:PDE:HNPP}
\left\{\begin{array}{rcl}
\left(\frac{\partial}{\partial r} + \mathscr{L}_r\right) \bar{u} (T-r,z)=0 & \text{for} & 0\less r \less T,\;
z\in\N.\\
\bar{u}(T-r,0)=0& \text{for} & 0\leq r\less T.\\
\bar{u}(0,z)=1 & \text{for} & z\in\N.\end{array}\right.
\end{equation}
Then,
\[
\bar{u}(T,x)=\Px[\sigma_a^1(x)> T],
\]
where $\sigma_a^1=\inf\{t\more0\;|\; q_t^a =0\}$.
\end{prop}

As a first approach some assumptions on the behaviour of the rates $\lambda_t$ and $\mu_t$ will be made.

\vspace{.4cm}

\textbf{Assumption 1:}
There exists a function $\alpha_t:\R_+\rightarrow\R_+$ such that for some positive constants $\lambda$ and $\mu$,
\begin{equation}\label{assumption1}
\lambda_t=\lambda\alpha_t\qquad\qquad\text{and}\qquad\qquad\mu_t=\lambda\alpha_t
\end{equation}

\vspace{.4cm}

Let $\Hx_t$ denote the generator of the Markov process describing the dynamics of the queues' size under Assumption 1, that is,
\begin{equation}\label{eqn:generatorA}
\Hx_t(i,j)=\left\{\begin{array}{ccl} 0&\text{if}& i=0,\ j\geq0,\\
\mu\alpha_t&\text{if}& 1\leq i,\ j=i-1,\\
\lambda\alpha_t&\text{if}& 1\leq i,\ j=i+1,\\
-(\lambda\alpha_t+\mu\alpha_t)&\text{if}& 1\leq i,\ j=i,\\ 0&\text{if}& |i-j|\more1.\end{array}\right.
\end{equation}

In the case when the rates are constant, that is when in Assumption 1 above, $\alpha_t\equiv1$, the generator $\mathscr{A}_t$ given in Equation \ref{eqn:generatorXY} reduces to
\begin{equation}\label{eqn:generatorC}
\Qx(i,j)=\left\{\begin{array}{ccl} 0&\text{if}& i=0,\ j\geq0,\\
\mu&\text{if}& 1\leq i,\ j=i-1,\\
\lambda&\text{if}& 1\leq i,\ j=i+1,\\
-(\lambda+\mu)&\text{if}& 1\leq i,\ j=i,\\ 0&\text{if}& |i-j|\more1.\end{array}\right.
\end{equation}

Before proceeding, we will define the inter-arrival times between price changes as follows. For $\sbt\in\{a,b\}$, let $\sigma_{\sbt,\Qx}^1$ be the extinction times of the Markov processes as described in Section \ref{Sec:Properties} where $\mathscr{L}^{(X,n)}_{t}(i,j)=\mathscr{L}_{t}^{(Y,n)}(i,j)=\Qx(i,j)$. Define in the same way the sequence $\{\tau_{\Qx}^n\}_{n\geq0}$ of inter-arrival times and the queue process $q_{t,\Qx}:=(q_{0,\Qx}^{a},q_{0,\Qx}^b)$.

Analogously, for $\sbt\in\{a,b\}$, let $\sigma_{\sbt,\Hx}^1$ be the extinction times of the Markov processes as described in Section  when  $\mathscr{L}^{(X,n)}_{t}(i,j)=\mathscr{L}_{t}^{(Y,n)}(i,j)=\Hx_t(i,j)$. Define in the same way the sequence $\{\tau_{\Hx}^n\}_{n\geq0}$ of inter-arrival times and the queue process $q_{t,\Hx}:=(q_{0,\Hx}^{a},q_{0,\Hx}^b)$.

The following lemma gives the distribution of $\sigma_{a,\Qx}^1$.

\begin{lemma}\label{lemma:Sol:u:Cont}
For the difference operator $\Qx$ given by (\ref{eqn:generatorC}), a solution to the initial value problem
\begin{equation}\label{eqn:PDE:HNPP:C}
\left\{\begin{array}{rcl}
\left(\frac{\partial}{\partial r} + \Qx\right) \bar{u} (T-r,z)=0 & \text{for} & 0\less r \less T,\;
z\in\N.\\
\bar{u}(T-r,0)=0& \text{for} & 0\leq r\less T.\\
\bar{u}(0,z)=1 & \text{for} & z\in\N.\end{array}\right.
\end{equation}
is given by the function
\begin{equation}\label{eqn:soln:IVP:C}
\bar{u}_{\Qx}(T,x)=\left(\frac{\mu}{\lambda}\right)^{x/2}\int_T^{\infty}\frac{x}{s}I_x\left(2s\sqrt{\lambda\mu}\right)e^{-s(\lambda+\mu)}ds,
\end{equation}
where $I_\nu(\cdot)$ is the modified Bessel function of the first kind. Consequently, $\Px[\sigma_{a,\Qx}^1\more T]=\bar{u}_{\Qx}(T,x)$.
\end{lemma}

\begin{proof} When $\mathscr{L}_t\equiv \Qx$, the model reduces to the case consider in \cite{ContLarrard2012}. Therefore, the result follows from Proposition 1 therein.
\end{proof}


It is important to analyze the tail behaviour of the survival distribution for $\sigma_{a,\Qx}^1$. The following lemma, whose proof is deferred to Appendix \ref{sec:AppendixA}, establishes such behaviour.

\begin{lemma} \label{lemma:tail:sigma}
Let $\mathcal{C}=(\sqrt{\mu}-\sqrt{\lambda})^2$. Then, for a sufficiently large $T$,
\[
\Px[\sigma_{a,\Qx}^1\more T\;|\;q_0^a=x]\sim\left\{\begin{array}{rrl}  \left(\frac{\mu}{\lambda}\right)^{x/2} \frac{x}{\mathcal{C}\sqrt{\pi\sqrt{\lambda\mu}}} \frac{e^{-T}}{T^{1/2}}&\text{if}& \lambda\less\mu\\ \frac{x}{\lambda\sqrt{\pi}} \frac{1}{\sqrt{T}}&\text{if}& \lambda=\mu\end{array}\right.
\]
Consequently, as expected, if $\lambda=\mu$, $\Ex[\sigma_{a,\Qx}^1]=\infty$, whereas if $\lambda\less\mu$, for every $k\in\N$,
\[
\Ex\left[\left(\sigma_{a,\Qx}^1\right)^k\right]\less\infty.
\]
\end{lemma}

\begin{rem}\label{Rem:Cont:incons}
Notice that if $\lambda=\mu$, the results in Lemma \ref{lemma:tail:sigma} agree with the results obtained in Eq. (6) in \cite{ContLarrard2012}. However, if $\lambda\less \mu$, Eq. (5) in \cite{ContLarrard2012} says that
\[
\Px[\sigma_{a,\Qx}^1\more T\;|\;q_0^a=x] \sim \frac{x(\lambda+\mu)}{2\lambda(\mu-\lambda)}\frac{1}{T},
\]
which is not correct, due to the well known fact that a Birth and death process with a death rate larger than its birth rate, its extinction time, $\sigma_{a, \Qx}^1$, has moments of all orders (an easy way to see this is to use the Moment Generating Function (MGF) computed on Proposition 1 in \cite{ContLarrard2012} and observe that if $\lambda\less \mu$, then the MGF is defined on an open interval around 0 (c.f. \cite{Billingsley:1995}[Section 21]).
\end{rem}

Lemma \ref{lemma:Sol:u:Cont} allows a closed formula to be obtained for the distribution of $\sigma_{a,\;\Hx}^1$, when the rates are proportional to each other, as in Assumption 1. Such a formula is described in the following proposition, whose proof is deferred to Appendix \ref{sec:AppendixA}.

\begin{prop}\label{prop:Sol:u:A1}
Under Assumption 1, for $A_t=\int_0^t\alpha_sds$, the distribution of $\sigma_{a,\Hx}^1$ is given by
\[
\Px[\sigma_{a,\Hx}^1\more T|\;q_0^a=x]=\left(\frac{\mu}{\lambda}\right)^{x/2}\int_{A_T}^{\infty}\frac{x}{s}I_x\left(2s\sqrt{\lambda\mu}\right)e^{-s(\lambda+\mu)}ds,
\]
where $I_\nu(\cdot)$ is the modified Bessel function of the first kind.
\end{prop}

\begin{rem}\label{rem:rel:H:Q} Lemma \ref{lemma:Sol:u:Cont} and Proposition \ref{prop:Sol:u:A1} imply that
\[
\Px[\sigma_{a,\Hx}^1\more T|\;q_0^a=x]=\Px[\sigma_{a,\Qx}\more A_T|\;q_0^a=x].
\]
This implies that the distribution of the time between price changes in the present model is comparable to the distribution of the inter-arrival time between price changes for the model presented in \cite{ContLarrard2012}
\end{rem}

Finally, we present the distribution of the time for the first price change.

\begin{cor}\label{cor:Distr:tau} Under Assumption 1, for $A_t=\int_0^t\alpha_sds$, the distribution of $\tau_{\Hx}^1$ is given by
$$
\Px[\tau_{\Hx}^1\more T\;|\;q_0^a=x, q_0^b=y]= \Px[\sigma_{a,\Hx}^1\more T|\;q_0^a=x] \Px[\sigma_{a,\Hx}^1\more T|\;q_0^a=y],
$$
where using the formula in proposition \ref{prop:Sol:u:A1}.
\end{cor}

\begin{proof}
The result follows from the fact that $\tau_{\Hx}^n=\sigma_{a,\Hx}^{n}\wedge\sigma_{\Hx,b}^{n}$, Lemma \ref{prop:Sol:u:A1} and the independence between $\sigma_{a,\Hx}^{n}$ and $\sigma_{b,\Hx}^{n}$.
\end{proof}

Now, the asymptotic behaviour of the survival distribution function of $\tau_{\Hx}^1$ is presented and its proof is deferred to Appendix \ref{sec:AppendixA}.

\begin{lemma}\label{lemma:asympt:tau}
Let $\mathcal{C}=(\sqrt{\mu}-\sqrt{\lambda})^2$. Then,
	\[
	\Px\left[\tau_{\Hx}^1\more T\;\Big|\;q_0^a=x, q_0^b=y\right] \sim \left\{\begin{array}{rrl}  \left(\frac{\mu}{\lambda}\right)^{(x+y)/2} \frac{xy}{\pi\mathcal{C}^2\sqrt{\lambda\mu}} \frac{\exp(-2A_{T}\mathcal{C})}{A_{T}}&\text{if}& \lambda\less\mu\\[.4cm] \dfrac{xy}{\lambda^2\pi} \dfrac{1}{A_T}&\text{if}& \lambda=\mu\end{array}\right.
	\]
Moreover, if
\begin{itemize}
	\item $\alpha_t\sim t^s\log^m(t)$ as $t\to\infty$ for some $s\neq-1$, $m\geq0$, and $n\in\N$,
	\[
	\Ex\left[\left(\tau_{\Hx}^1\right)^n\;\Big|\;q_0^a=x, q_0^b=y\right]\left\{\begin{array}{rrl} \less\infty &\text{if}& \lambda\less\mu\\ \less\infty &\text{if}& \lambda=\mu\text{ and }n< s+1\\ =\infty &\text{if}& \lambda=\mu\text{ and }n\geq s+1 \end{array}\right.
	\]
		\item $\alpha_t\sim k/t$ as $t\to\infty$ for some $k>0$,
	\[
	\Ex\left[\left(\tau_{\Hx}^1\right)^n\;\Big|\;q_0^a=x, q_0^b=y\right]\left\{\begin{array}{rrl} \less\infty &\text{if}& n<2k\mathcal{C}\text{ and }\lambda\less\mu\\ =\infty &\text{if}& n\geq2k\mathcal{C}\text{ and }\lambda\less\mu\\  =\infty &\text{if}& \lambda=\mu\end{array}\right.
	\]
\end{itemize}
\end{lemma}

\subsection{Long-run dynamics of the price process}

We are interested in analyzing the asymptotic behaviour of the number of price changes up to time $t$. That is, in describing
\begin{equation}\label{Eqn:Nt:1}
N_t^{\sbt}:=\max\{n\more0\;|\;\tau_{\sbt}^1+\tau_{\sbt}^2+\ldots+\tau_{\sbt}^n\leq t\},
\end{equation}
where, $\tau_{\Qx}^n$ and $\tau_{\Hx}^n$ are defined above.

The next proposition, whose proof is  deferred to Appendix \ref{sec:AppendixA}, provides an expression which relates the distribution of the partial sums for the waiting times between price changes for the models with the generators $\Hx_t$ and $\Qx$.

\begin{prop}\label{prop:Distr:Sn}
Let $S_{\Hx}^n:=\tau_{\Hx}^1+\tau_{\Hx}^2+\ldots+\tau_{u\Hx}^n$ and $S_{\Qx}^n:=\tau_{\Qx}^1+\tau_{\Qx}^2+\ldots+\tau_{\Qx}^n$. Then,
\[
\Px[S^n_{\Hx}\leq t]=\Px[S^n_{\Qx}\leq A_t],
\]
where $A_t=\int_0^t\alpha_sds$ in accordance with Assumption 1.
\end{prop}

The following results provide the convergence of the price process. For presentation purposes we separate them into the case when $\lambda<\mu$ and when $\lambda=\mu$.

\begin{thm}\label{thm:main:part1:1} Assume $\lambda\less\mu$ and let $\mathcal{C}=(\sqrt{\mu}-\sqrt{\lambda})^2$. Then, under Assumption 1, for $A_t=\int_0^t\alpha_sds$, 
	\begin{itemize}
		\item If $\dfrac{\alpha_t}{t^{s}} \to \widetilde{K}$ with $s\neq -1$ or if $\dfrac{\alpha_t}{t^{-1}} \to K$ as $t\to\infty$, with $2\mathcal{C}K> 1$, the rescaled price process converges and for the sequence $t_n=nt$ and a constant $\sigma$, in distribution,
		\[
		\frac{s_{t_n}}{\sqrt{n}}\Rightarrow \sigma W_t
		\]
		\item If $\dfrac{\alpha_t}{t^{-1}} \sim K$ with $2\mathcal{C}K\leq 1$ as $t\to\infty$, the rescaled price process converges and for the sequence $t_n=tn^{1/2k\mathcal{C}}$ and a constant $\sigma$, in distribution,
		\[
		\frac{s_{t_n}}{\sqrt{n}}\Rightarrow \sigma \int_0^t \sqrt{\frac{1}{u^{1-2\mathcal{C}K}}}dW_u
		\]
	\end{itemize}
\end{thm}

\begin{thm}\label{thm:main:part1:2} Assume $\lambda = \mu$. Then, under Assumption 1, for $A_t=\int_0^t\alpha_sds$, 
	\begin{itemize}
		\item If $\alpha_t\sim t^{-1+s}$ as $t\to\infty$, for any $s\more 1$, the rescaled price process converges and for the sequence $t_n=nt$ and a constant $\sigma$, in distribution,
		\[
		\frac{s_{t_n}}{\sqrt{n}}\Rightarrow \sigma W_t
		\]
		\item If $\alpha_t\sim t^{-1+s}$ as $t\to\infty$ for any $s\in(0,1]$, the rescaled price process converges and for the sequence $t_n=t(n)^{1/s}\log(n)$ and a constant $\sigma$, in distribution,
		\[
		\frac{s_{t_n}}{\sqrt{n}}\Rightarrow \sigma\int_0^t\sqrt{\frac{1}{u^{s}}}dW_u
		\]
		\item If $\alpha_t=o(t^{-1+s})$, for some $s>0$, the price process converges cannot be rescaled to ensure convergence.
	\end{itemize}

\end{thm}

\begin{rem} \label{rem:Comparison:cont} It is important to notice that we recover the model proposed by \cite{ContLarrard2012} in our current setting. Indeed, for recovering their model, $\alpha_t$ should be chosen such that $\alpha_t\equiv 1=x^0$. In such case,
\begin{itemize}
	\item If $\lambda<\mu$ and $\alpha_t=1$, then $\alpha_t/t^{-1}\to\infty$ and the result in \cite{ContLarrard2012} follows immediately from Theorem \ref{thm:main:part1:1}.
	\item If $\lambda=\mu$ and $\alpha_t=t^{-1+1}$, it becomes the borderline case in Theorem \ref{thm:main:part1:2}. In such case, we can see from the proof of the aforementioned theorem that we need to choose the rescaling sequence $t_n=tn\log(n)$, but we get convergence to a Brownian motion with constant volatility.
\end{itemize}

\end{rem}

\section{Empirical Results} \label{sec:empirical}
In this paper six different stocks will be analyzed. These stocks vary on the type of sector they belong to and in their frequency of trading as well as other properties. The main goal is to see the different quantities being fit to the model and to contrast them to the model in \cite{ContLarrard2012}. The selected stocks were CSCO, FB, INTC, MSFT, LBTYK and VOD and all were analyzed on the week of Nov 3rd to Nov 7th of 2014.

The first quantity that will be analyzed in this paper will be the distribution of the time between price changes. From Lemma \ref{lemma:asympt:tau}, it follows that 
	\[
	\Px\left[\tau_{\Hx}^1\more T\;\Big|\;q_0^a=x, q_0^b=y\right] \sim \left\{\begin{array}{rrl}  \left(\frac{\mu}{\lambda}\right)^{(x+y)/2} \frac{xy}{\pi\mathcal{C}^2\sqrt{\lambda\mu}} \frac{\exp(-2A_{T}\mathcal{C})}{A_{T}}&\text{if}& \lambda\less\mu\\[.4cm] \dfrac{xy}{\lambda^2\pi} \dfrac{1}{A_T}&\text{if}& \lambda=\mu\end{array}\right.
	\]
Then, by adding over all possible initial positions of the queues and multiplying by the prbability that such position occurs, we have that
	\[
	\Px\left[\tau_{\Hx}^1\more T\right] \sim \left\{\begin{array}{rrl}  \sum\limits_{x=1}^\infty\sum\limits_{y=1}^\infty \left(\frac{\mu}{\lambda}\right)^{(x+y)/2} \frac{xyf(x,y)}{\pi\mathcal{C}^2\sqrt{\lambda\mu}} \frac{\exp(-2A_{T}\mathcal{C})}{A_{T}}&\text{if}& \lambda\less\mu\\[.4cm] \sum\limits_{x=1}^\infty\sum\limits_{y=1}^\infty \dfrac{xyf(x,y)}{\lambda^2\pi} \dfrac{1}{A_T}&\text{if}& \lambda=\mu\end{array}\right.
	\]
By taking derivatives and using L'Hopital rule, we have that if $f_{\tau_{\Hx}^1}(t)$ is the density of $\tau_{\Hx}^1$ then,
	\begin{equation}\label{eqn:density:tau}
	f_{\tau_{\Hx}^1}(t) \sim \left\{\begin{array}{rrl}  \sum\limits_{x=1}^\infty\sum\limits_{y=1}^\infty \left(\frac{\mu}{\lambda}\right)^{(x+y)/2} \frac{xyf(x,y)}{\pi\mathcal{C}^2\sqrt{\lambda\mu}} \frac{(2\mathcal{C}A_T+1)\alpha_T\exp(-2A_{T}\mathcal{C})}{A_{T}^2}&\text{if}& \lambda\less\mu\\[.4cm] \sum\limits_{x=1}^\infty\sum\limits_{y=1}^\infty \dfrac{xyf(x,y)}{\lambda^2\pi} \dfrac{\alpha_T}{A_T^2}&\text{if}& \lambda=\mu\end{array}\right.
	\end{equation}

The next figure show the empirical densities of $f_{\tau_{\Hx}^1}(t)$ for the six picked stocks. As it will be shown in Table \ref{tab:Quotinten}, in all six cases it happens that,in average, $\mu<\lambda$ but in almost all cases $\mu/\lambda>0.9$.

\begin{center}
\begin{minipage}[c]{0.85\textwidth}
\centering
\begin{minipage}{0.49\textwidth}
    \includegraphics[width=\textwidth]{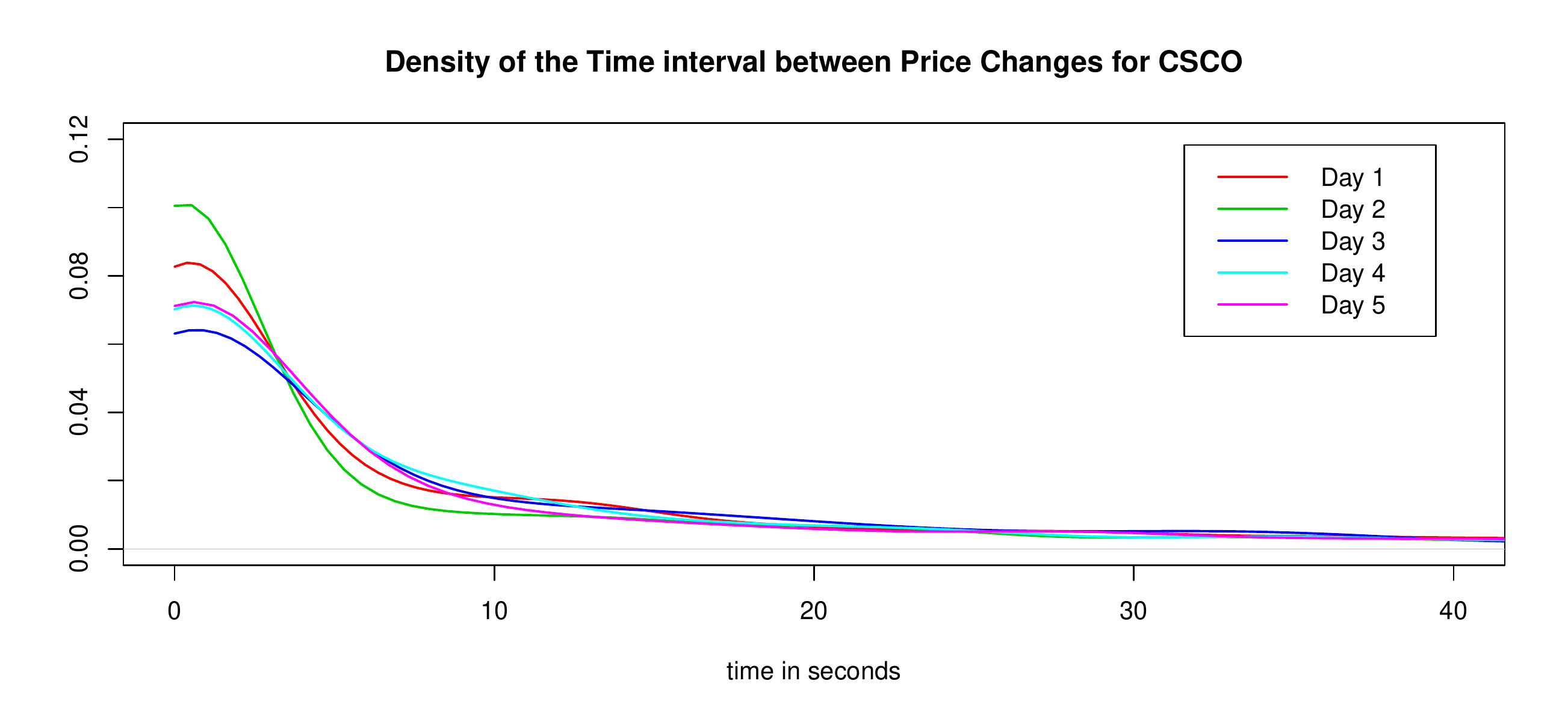}
		\includegraphics[width=\textwidth]{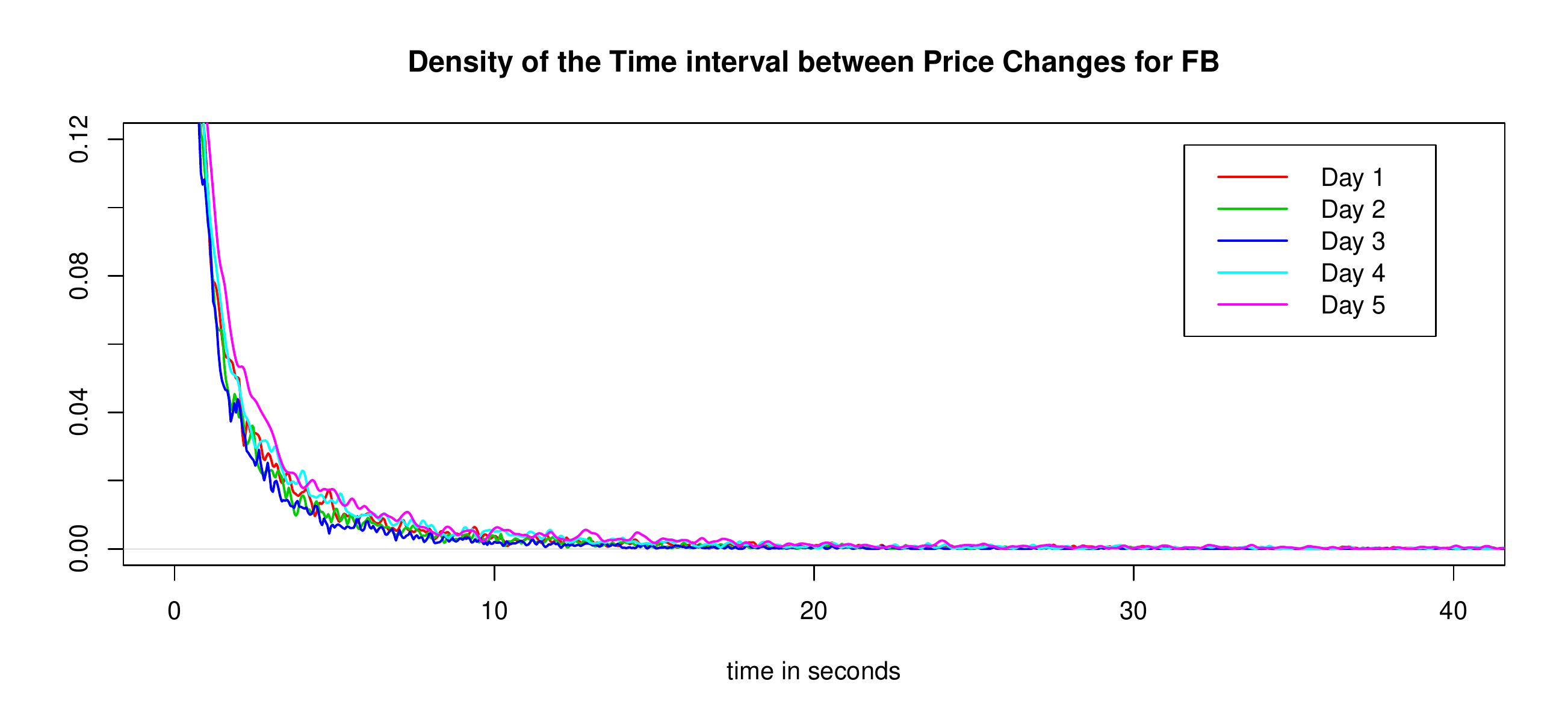}
		\includegraphics[width=\textwidth]{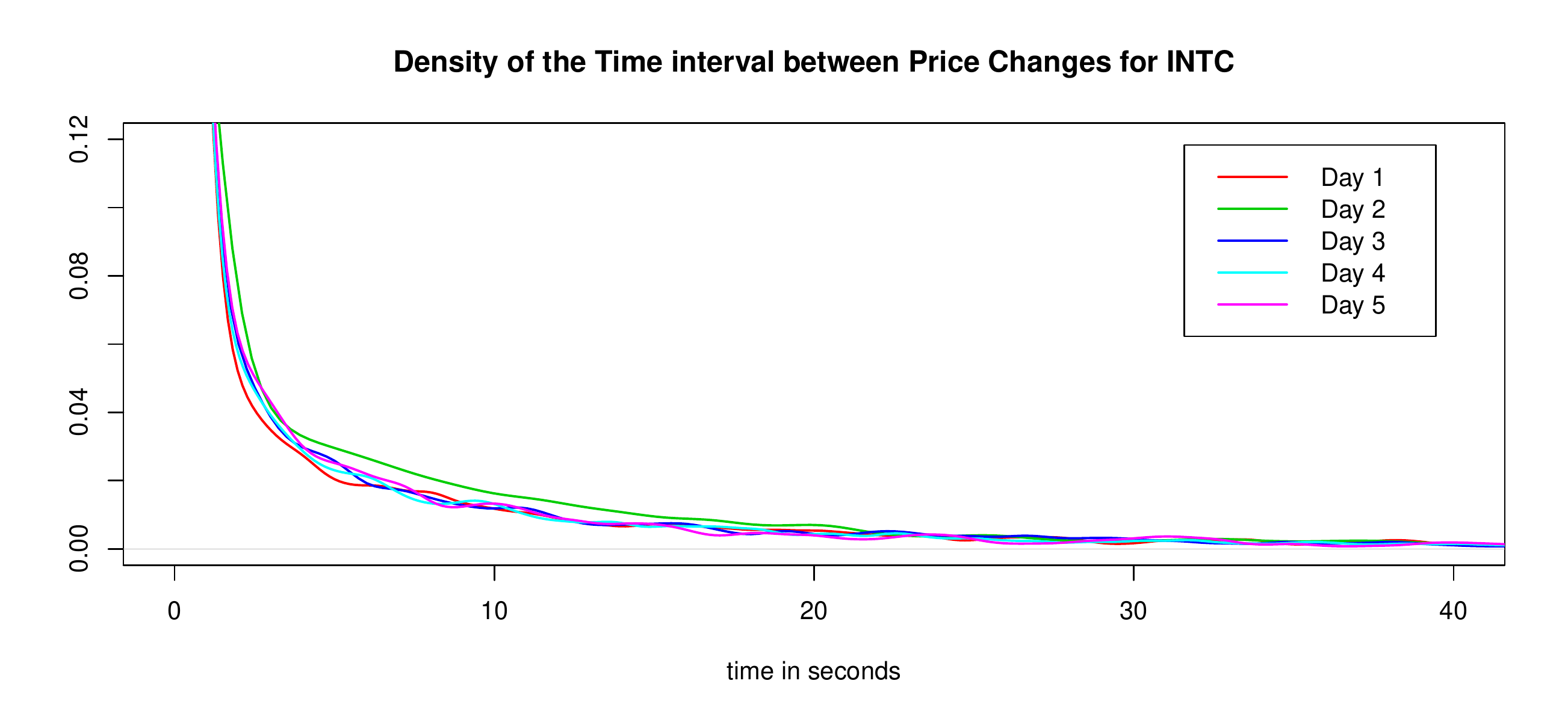}
\end{minipage}
\begin{minipage}{0.49\textwidth}
		\includegraphics[width=\textwidth]{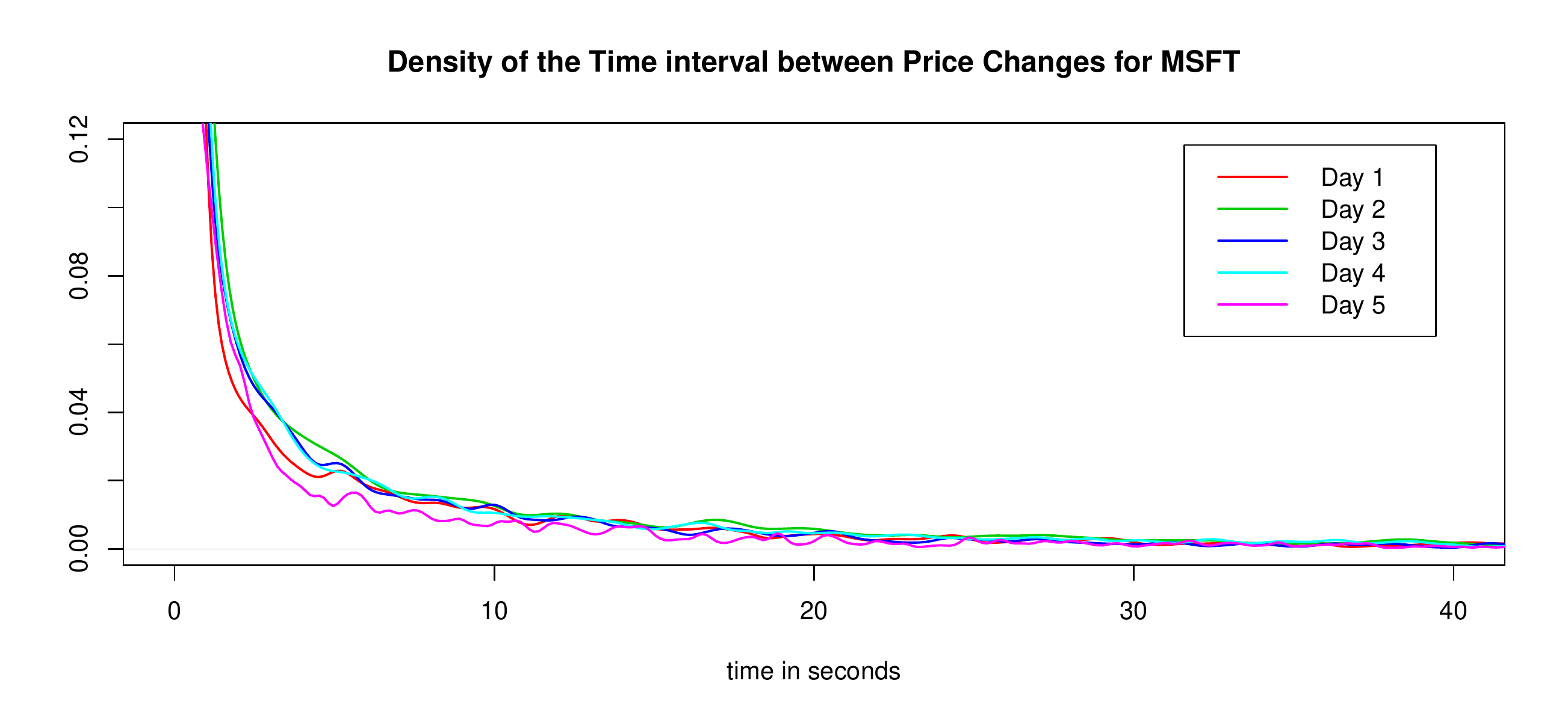}
		\includegraphics[width=\textwidth]{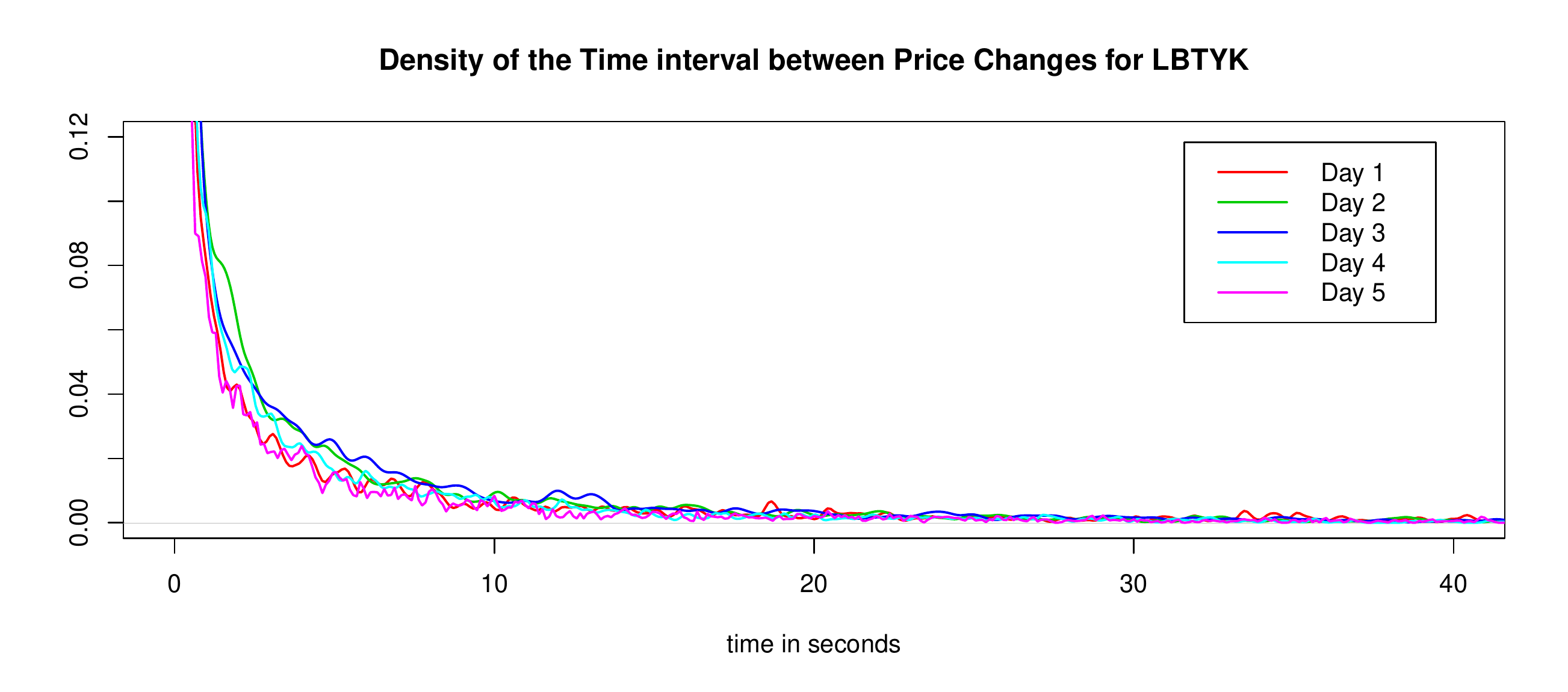}
		\includegraphics[width=\textwidth]{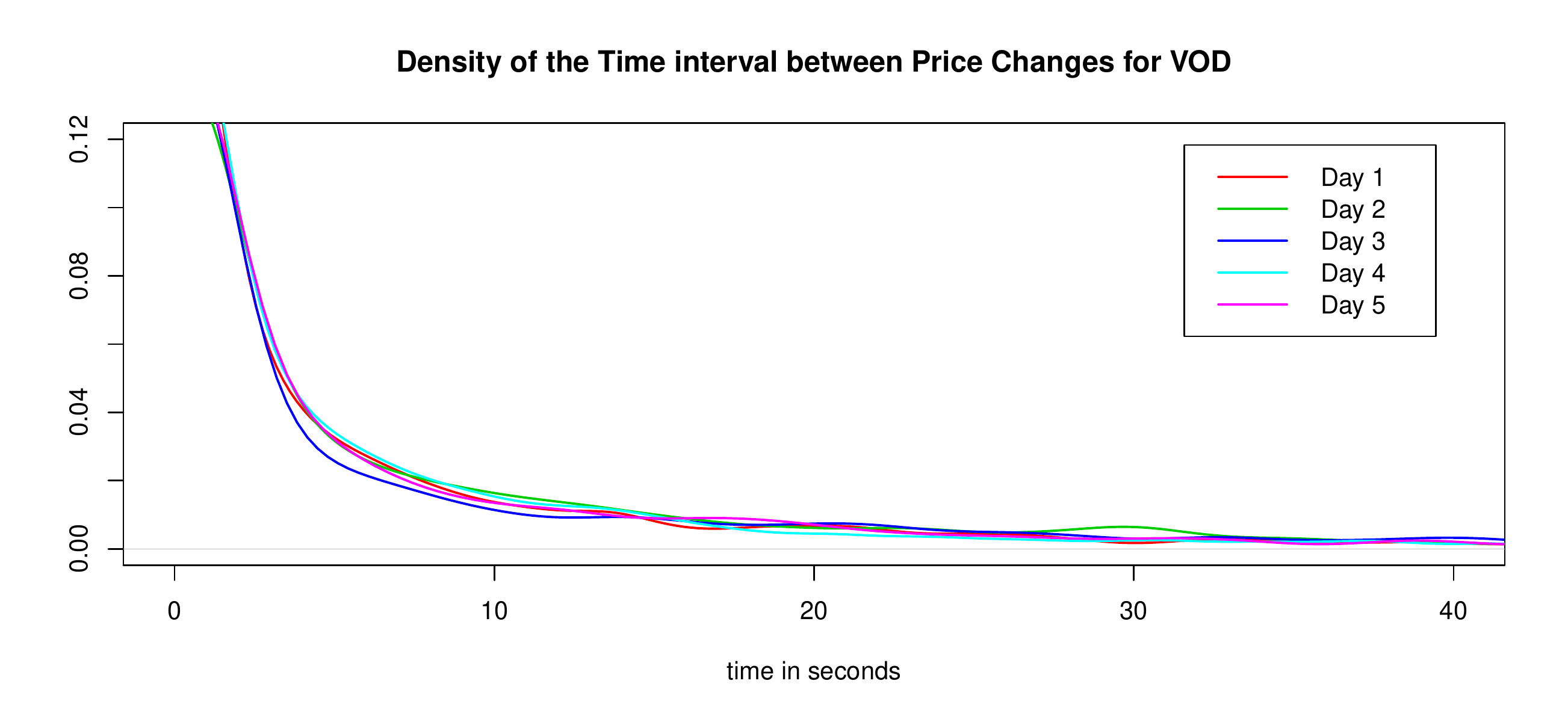}
\end{minipage}
    \captionof{figure}{Densities of the inter-arrival time between price changes on the six stocks for the week of Nov 3rd to Nov 7th of 2014.}
    \label{fig:timeinterval}
		\bigskip
		\medskip
\end{minipage}
\end{center}

Next, the intensities of Limit orders at the ask side, $\lambda_t^a$, and Market orders plus Cancellations, $\mu_t^a$, are plotted for each stock. In each plot the intensity $\lambda_t^a$ or $\mu_t^a$ is computed for each day of the week and a power-law fit is found using regression. The results of the regression are summarized in Table \ref{tab:Askinten} after the corresponding figure.

\begin{center}
\begin{minipage}[c]{0.85\textwidth}
\centering
\begin{minipage}{0.49\textwidth}
    \includegraphics[width=\textwidth]{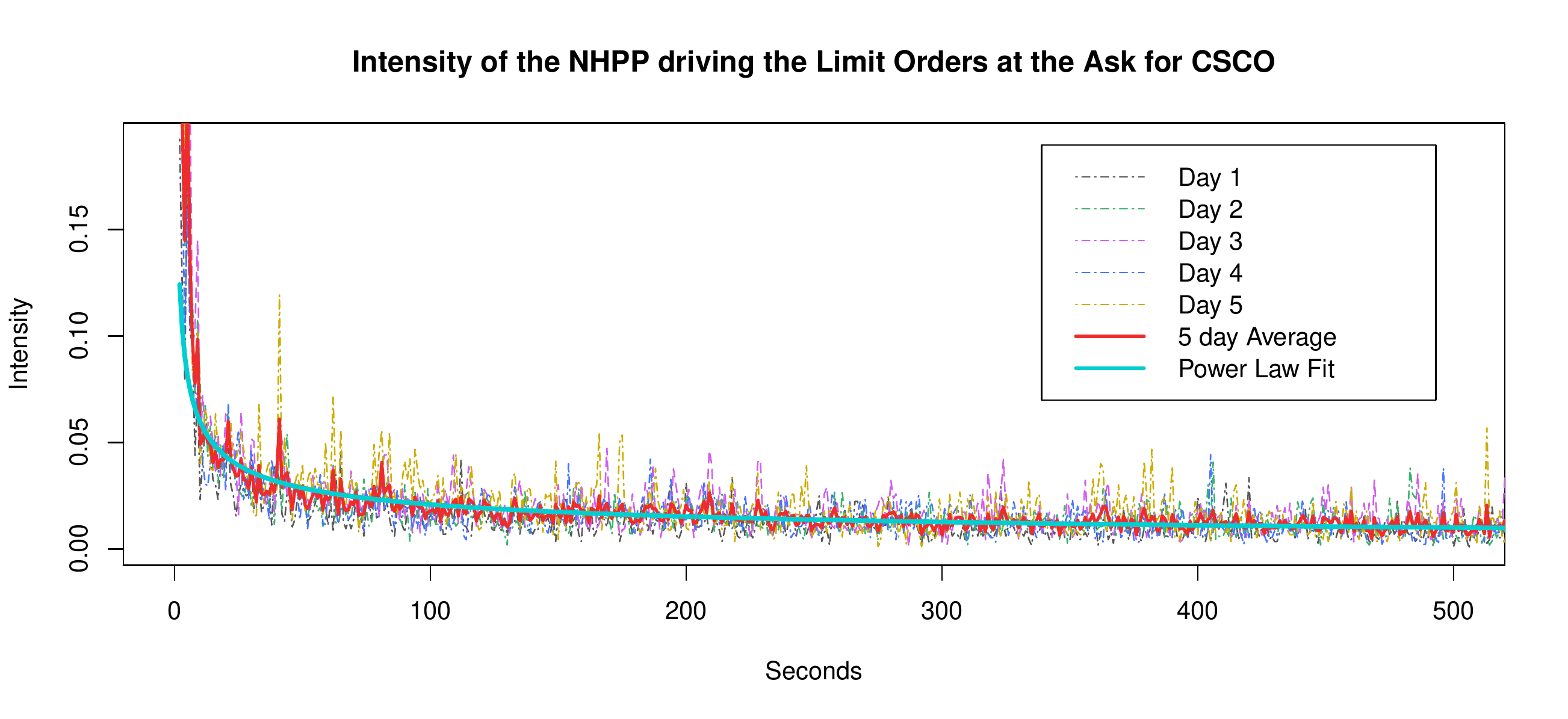}
		\includegraphics[width=\textwidth]{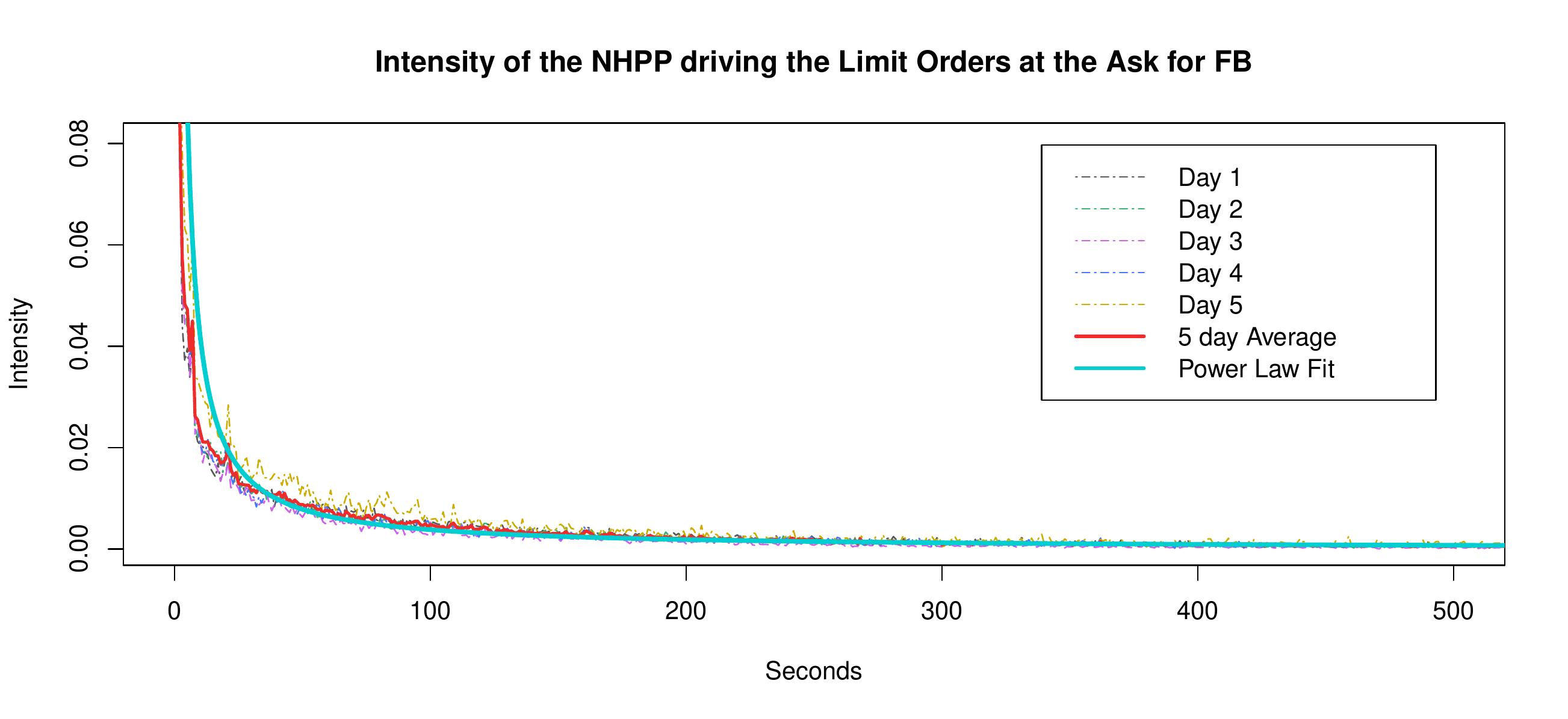}
		\includegraphics[width=\textwidth]{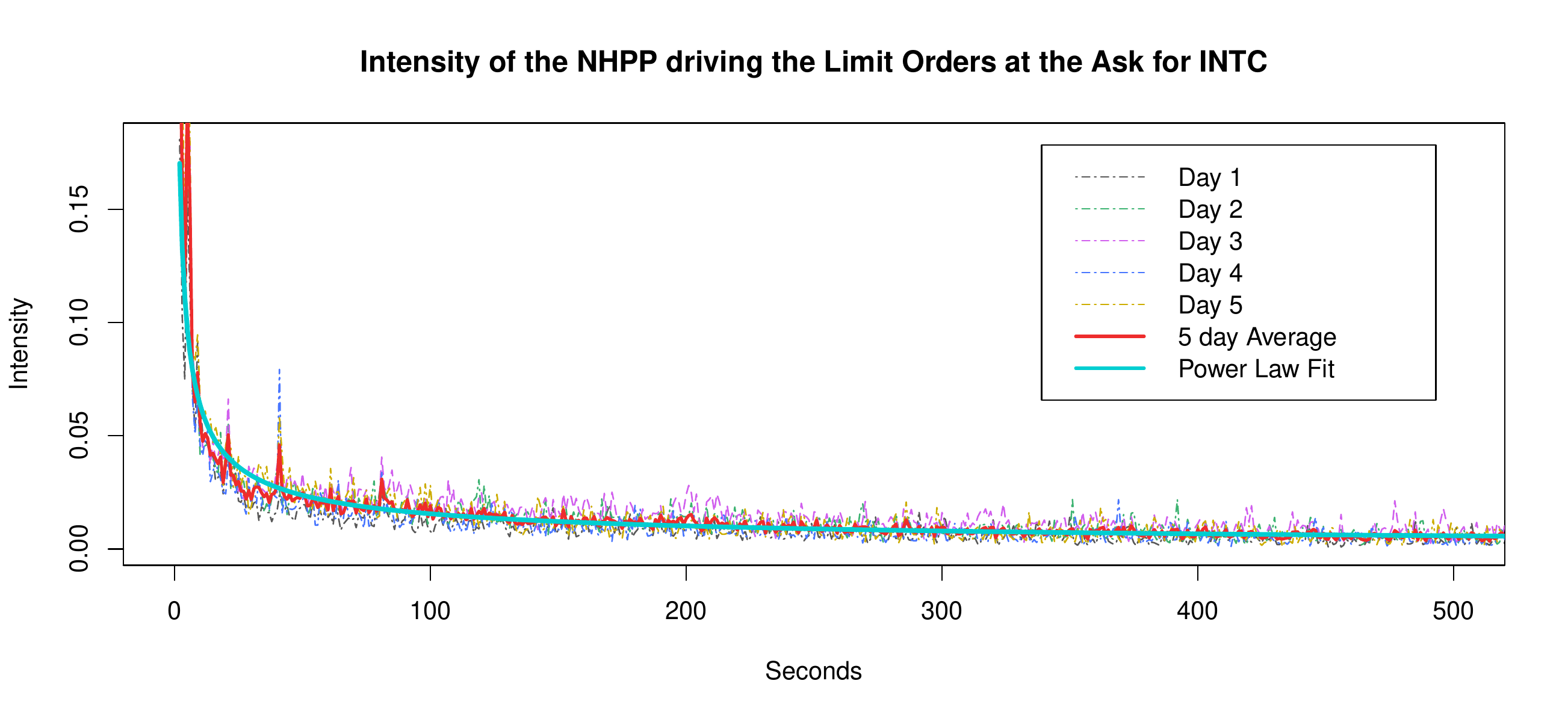}
\end{minipage}
\begin{minipage}{0.49\textwidth}
		\includegraphics[width=\textwidth]{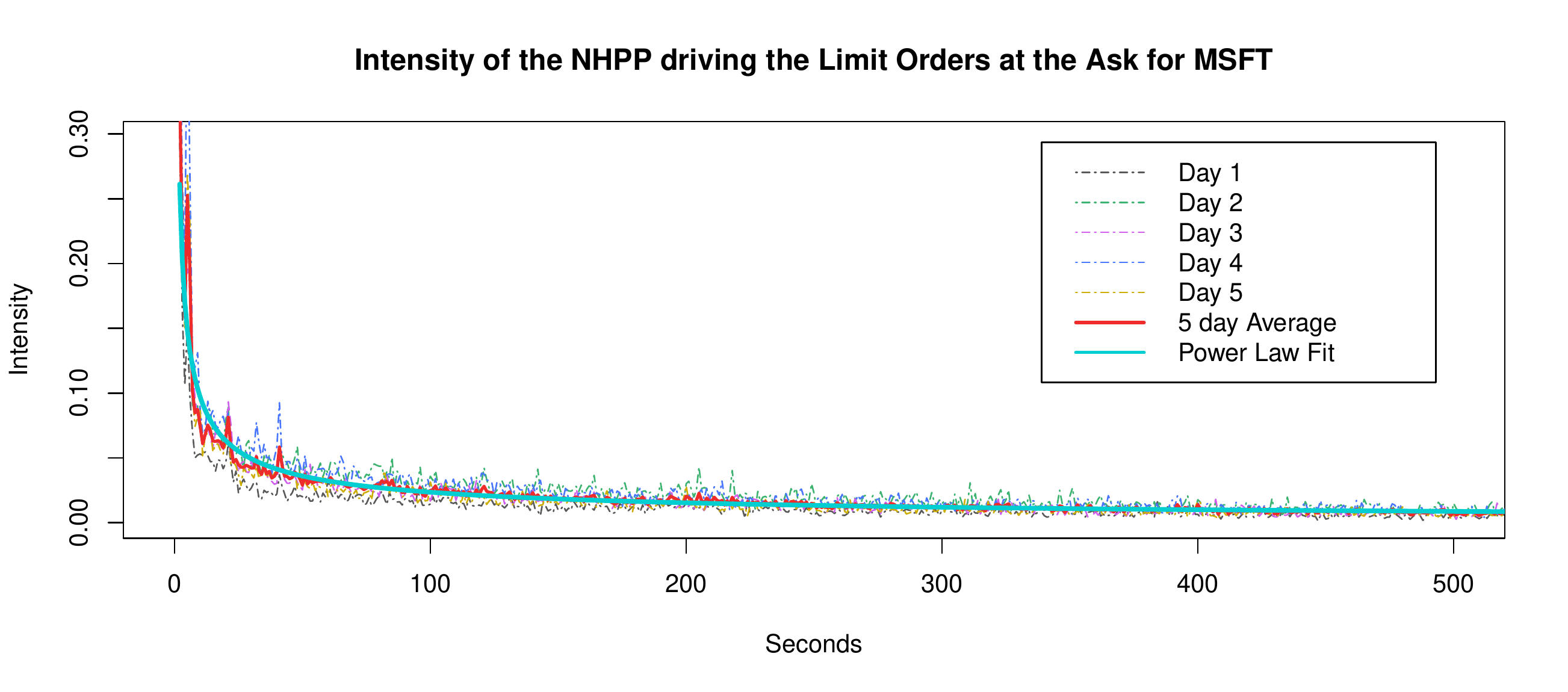}
		\includegraphics[width=\textwidth]{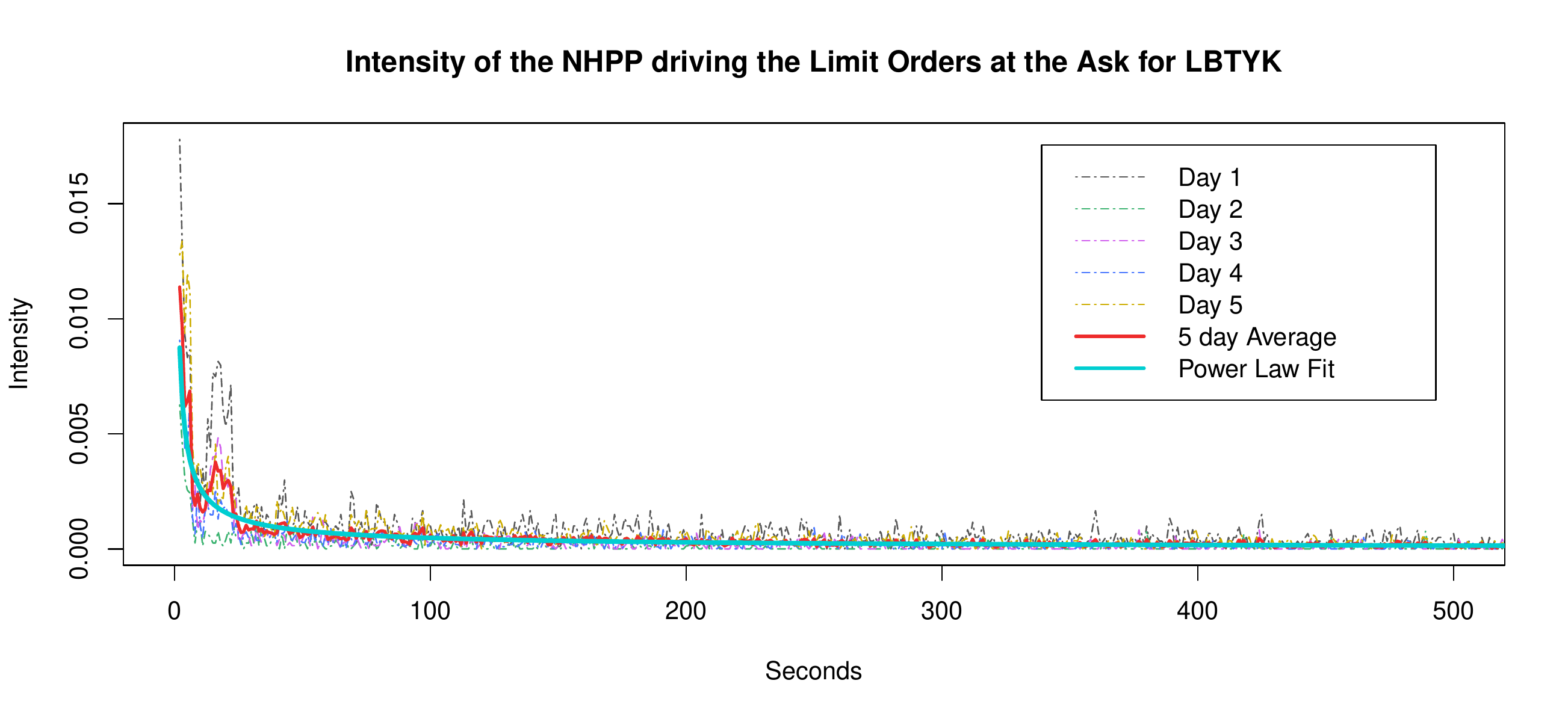}
		\includegraphics[width=\textwidth]{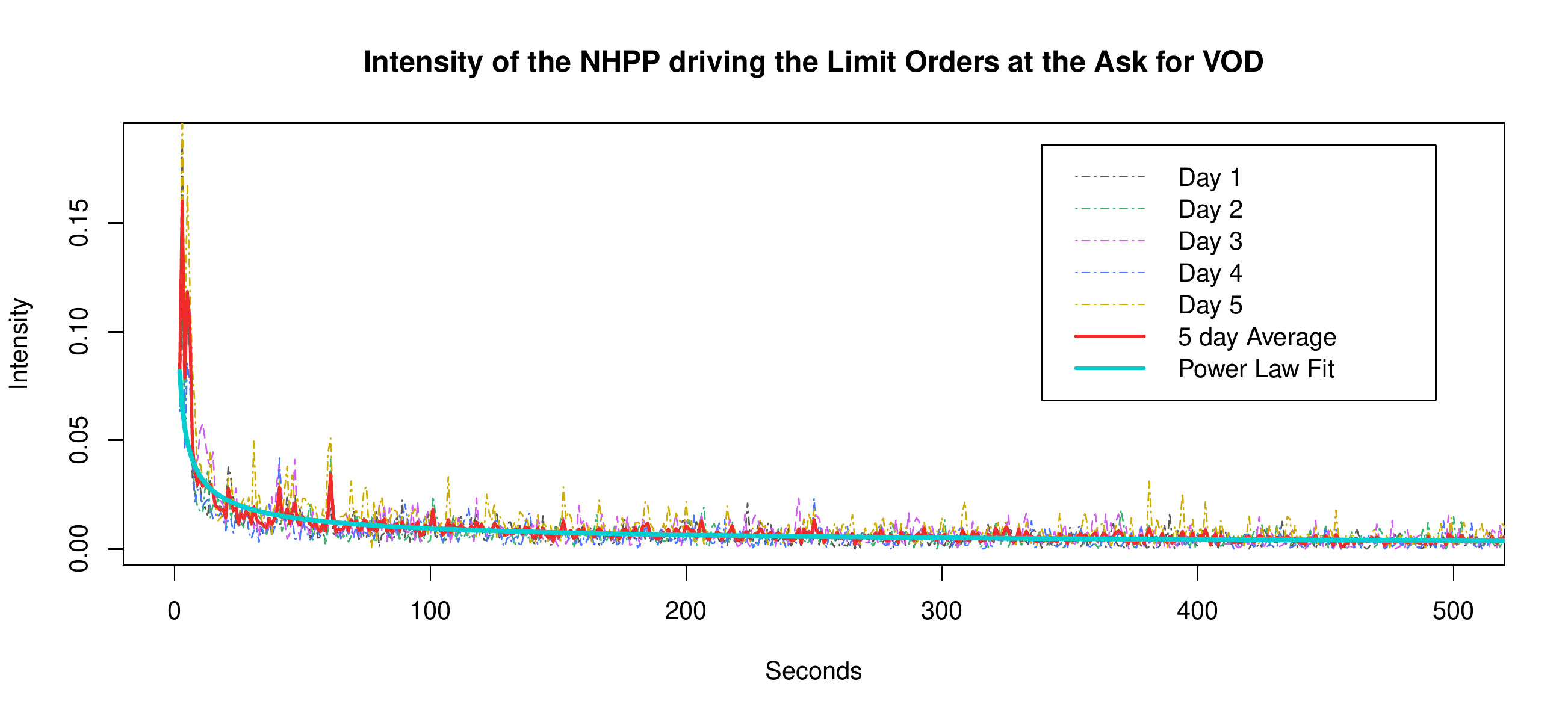}
\end{minipage}
    \captionof{figure}{Daily intensities of Limit Orders at the Ask side for the six stocks considered on the week of Nov 3rd to Nov 7th of 2014 and their corresponding power law fit.}
    \label{fig:LOAintensity}
		\bigskip
		\medskip
\end{minipage}
\end{center}

\begin{center}
\begin{minipage}[c]{0.85\textwidth}
\centering
\begin{minipage}{0.49\textwidth}
    \includegraphics[width=0.9\textwidth]{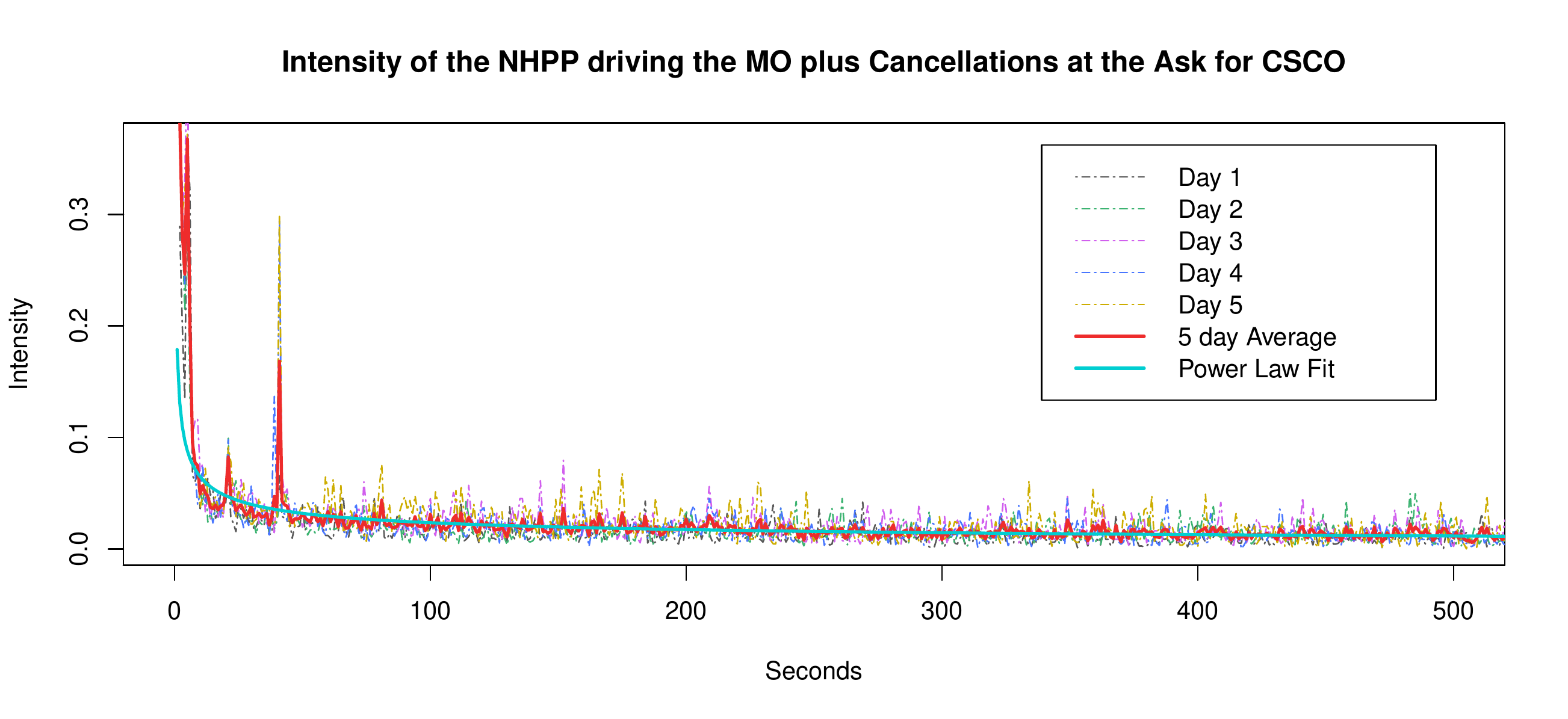}
		\includegraphics[width=0.9\textwidth]{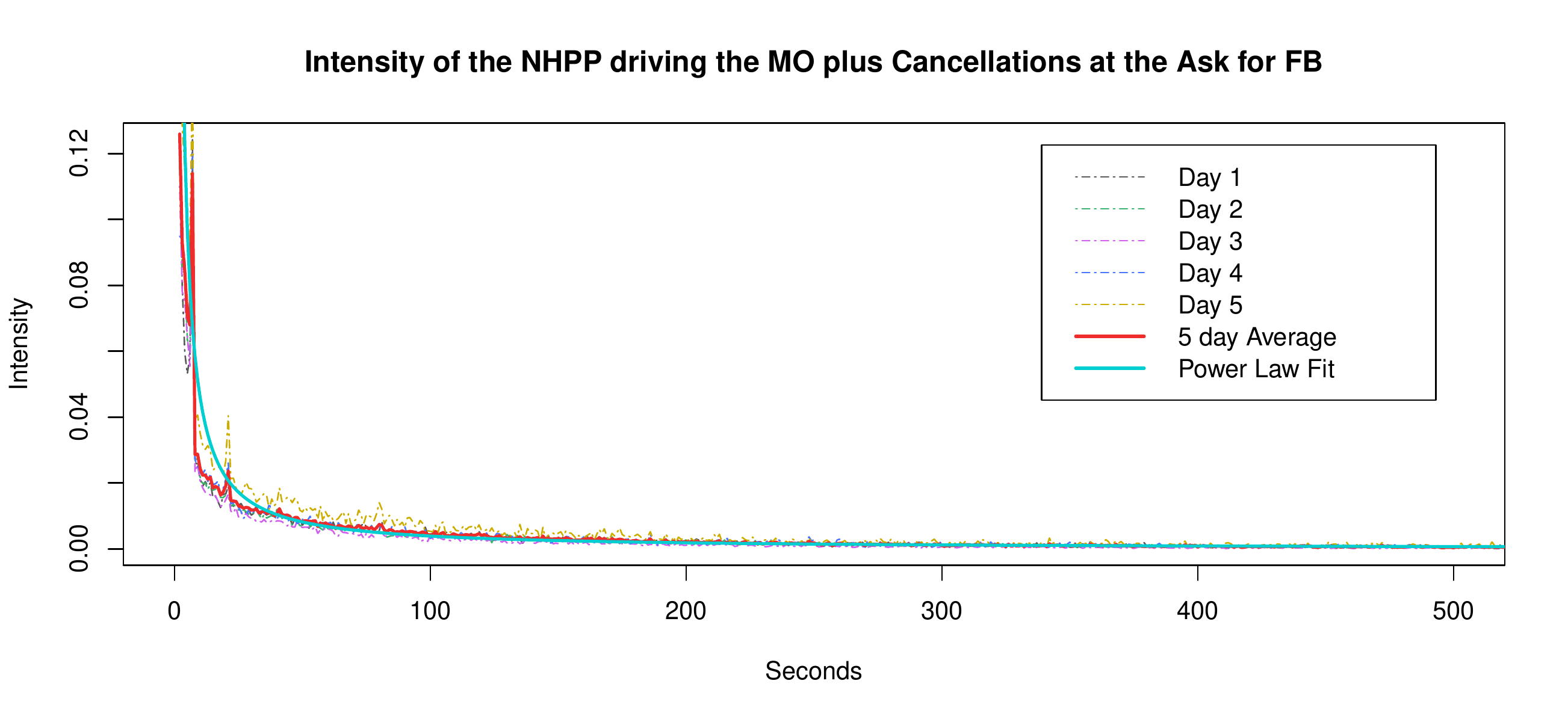}
		\includegraphics[width=0.9\textwidth]{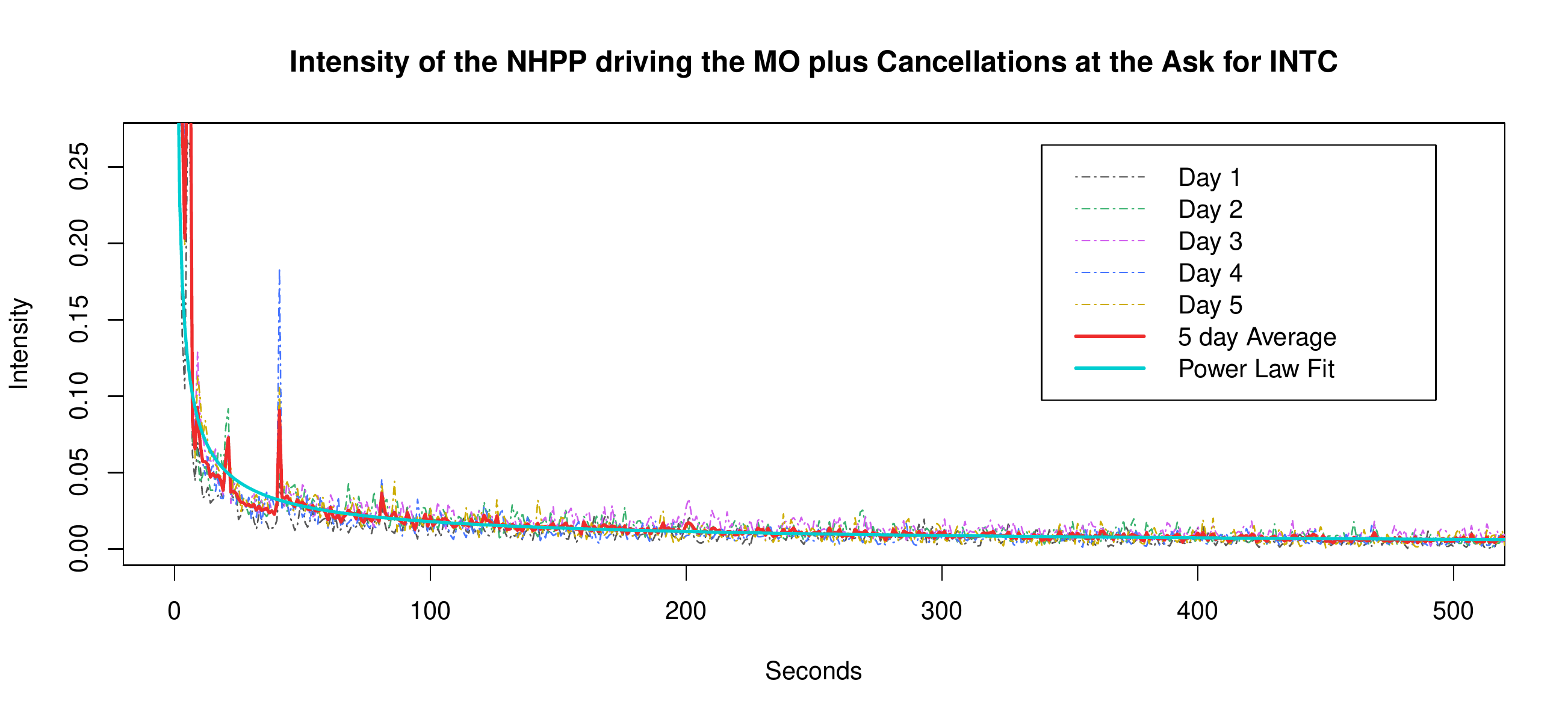}
\end{minipage}
\begin{minipage}{0.49\textwidth}
		\includegraphics[width=0.9\textwidth]{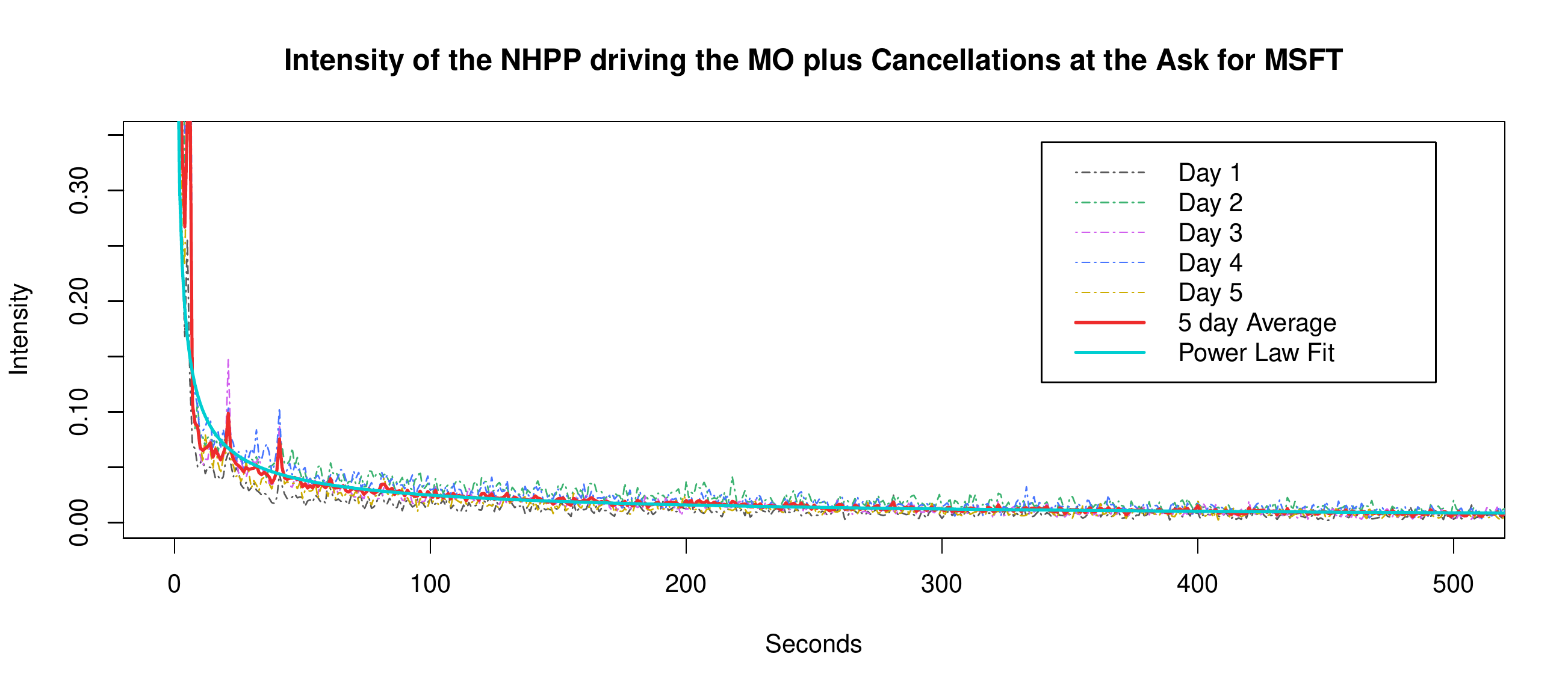}
		\includegraphics[width=0.9\textwidth]{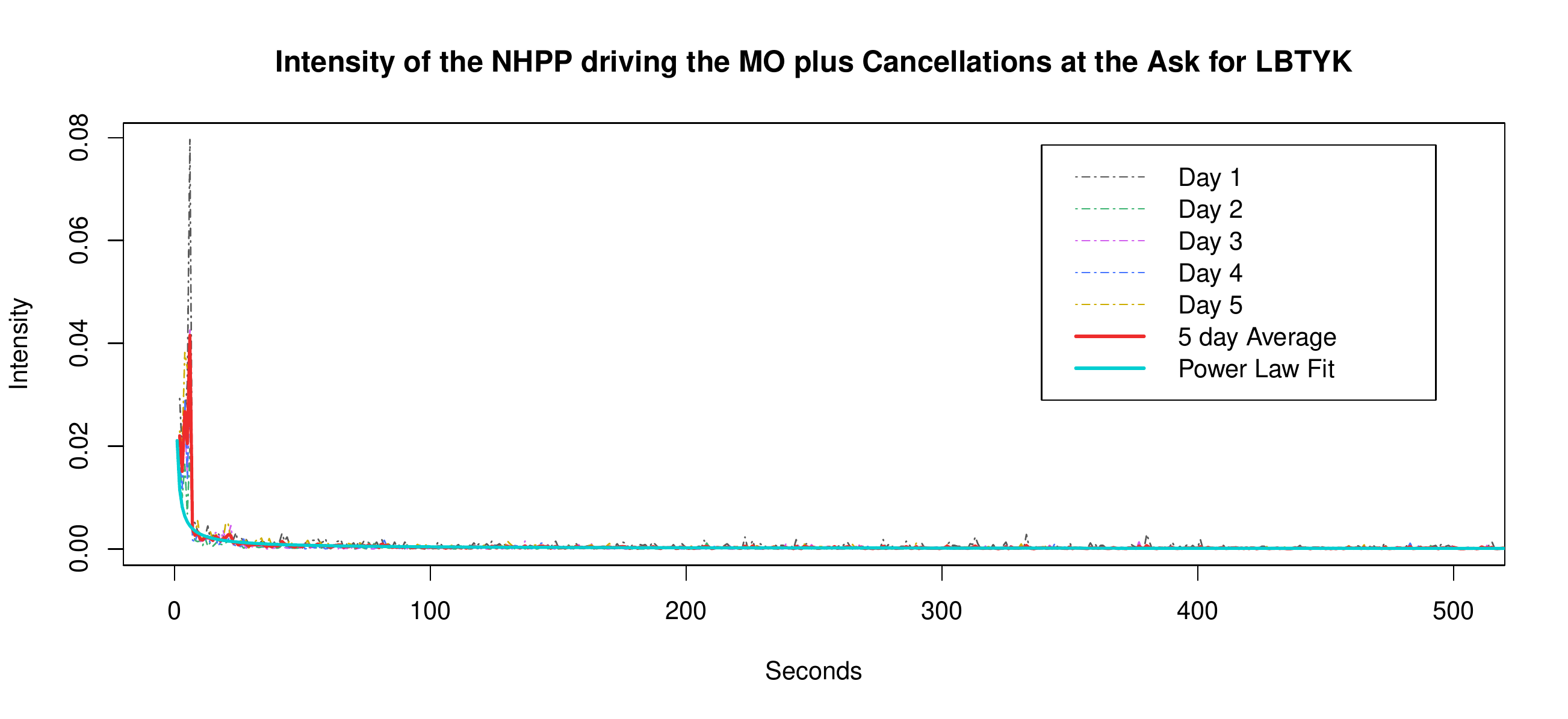}
		\includegraphics[width=0.9\textwidth]{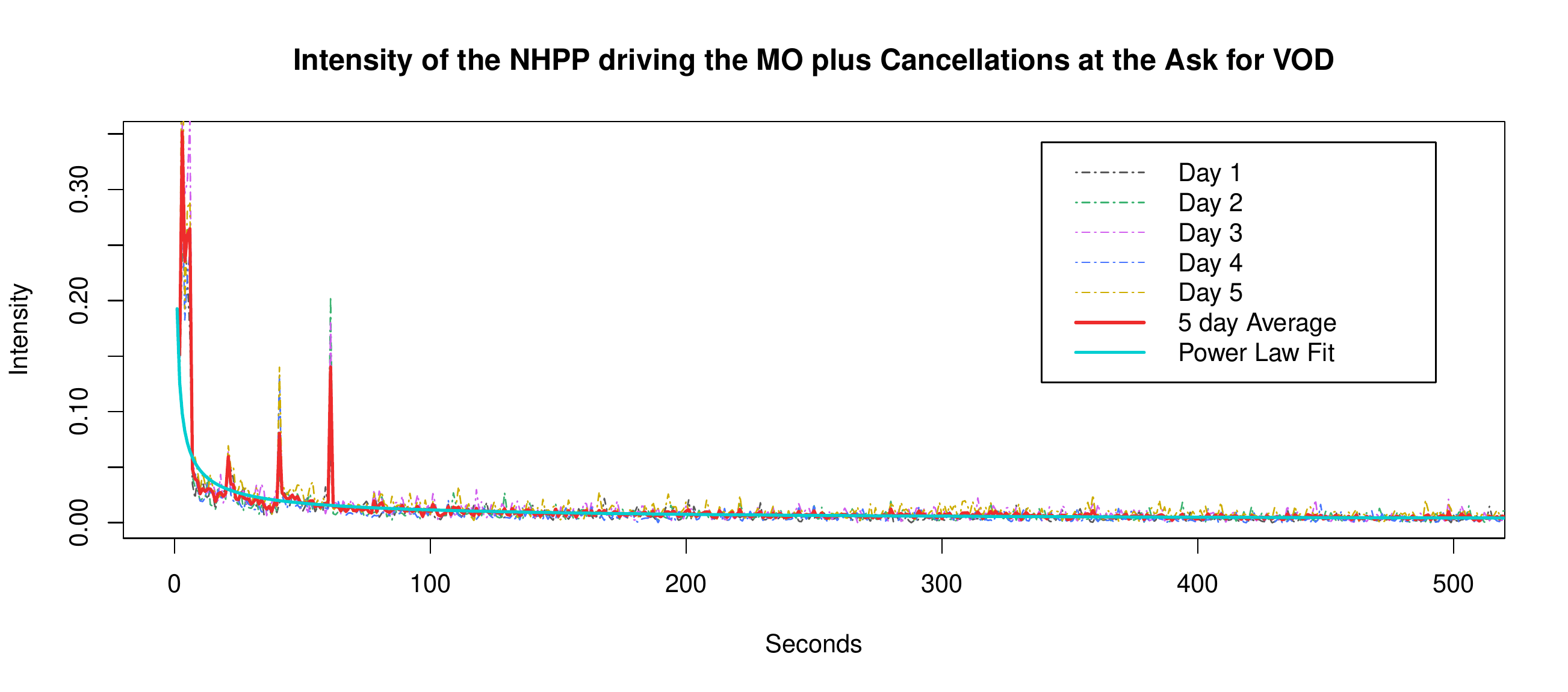}
\end{minipage}
    \captionof{figure}{Daily intensities of Marker Orders plus Cancellations on the Ask side for the six stocks considered on the week of Nov 3rd to Nov 7th of 2014 and their corresponding power law fit.}
    \label{fig:MOAintensity}
		\bigskip
		\medskip
\end{minipage}
\end{center}

In order to approximate the the long-run dynamics of the price process as stated in Theorems \eqref{thm:main:part1:1}-\eqref{thm:main:part1:2}, a power law fit to the intensity of Limit Orders $\lambda_t^a$ at the ask side was fit in each of the six stocks analyzed. That is, a regression is performed to fit $\lambda_t^a\approx\frac{K_{\lambda,a}}{t^s}$. Similarly, a power law fit to the intensity of Market Orders plus Cancellations $\mu_t^a$ at the ask side was fit in each of the six stocks analyzed. In this case, a regression is performed to fit $\mu_t^a\approx\frac{K_{\mu,a}}{t^r}$. The following table summarizes the power law fit to the intensities  of the analyzed stocks.

\renewcommand{\arraystretch}{1.4}

\begin{table}[h]
	\centering
		\begin{tabular}{|l|c|c|c|c|}\hline
			\multirow{2}{*}{\textbf{Stock}} & \multicolumn{2}{c|}{\textbf{Fit for }$\lambda_t^a\approx K_{\lambda,a}t^{-s}$} & \multicolumn{2}{c|}{\textbf{Fit for }$\mu_t^a\approx K_{\mu,a}t^{-r}$}\\ \cline{2-5}
			& \textbf{Coefficient} $K_{\lambda,a}$ & \textbf{Exponent} $s$ & \textbf{Coefficient} $K_{\mu,a}$ & Exponent $r$\\ \hline\hline
			CSCO  &  0.1703         & 0.4560       & 0.1790          & 0.4412  \\ \hline
			FB    &  0.4664         & 1.0045       & 0.5429          & 1.0073  \\ \hline
			INTC  &  0.2604         & 0.6127       & 0.3582          & 0.6515  \\ \hline
			MSFT  &  0.4002         & 0.6153       & 0.4671          & 0.6363  \\ \hline
			LBTYK &  0.0146         & 0.7438       & 0.0211          & 0.8640  \\ \hline
			VOD   &  0.1199         & 0.5536       & 0.1927          & 0.6116  \\ \hline
		\end{tabular}
	\caption{Regression fit for $\lambda_t^a\approx Ct^{-s}$ and $\mu_t^a\approx Dt^{-r}$ for all six stocks analyzed.}
	\label{tab:Askinten}
	\bigskip
	\medskip
\end{table}

Now, the intensities of Limit orders at the bid side, $\lambda_t^b$, and Market orders plus Cancellations, $\mu_t^b$, are now plotted for each stock. Similarly as before, in each plot the intensity $\lambda_t^b$ or $\mu_t^b$ is computed for each day of the week and a power-law fit is found using regression. The results of the regression are summarized in Table \ref{tab:Bidinten} after the corresponding figure.

\begin{center}
\begin{minipage}[c]{0.85\textwidth}
\centering
\begin{minipage}{0.49\textwidth}
    \includegraphics[width=\textwidth]{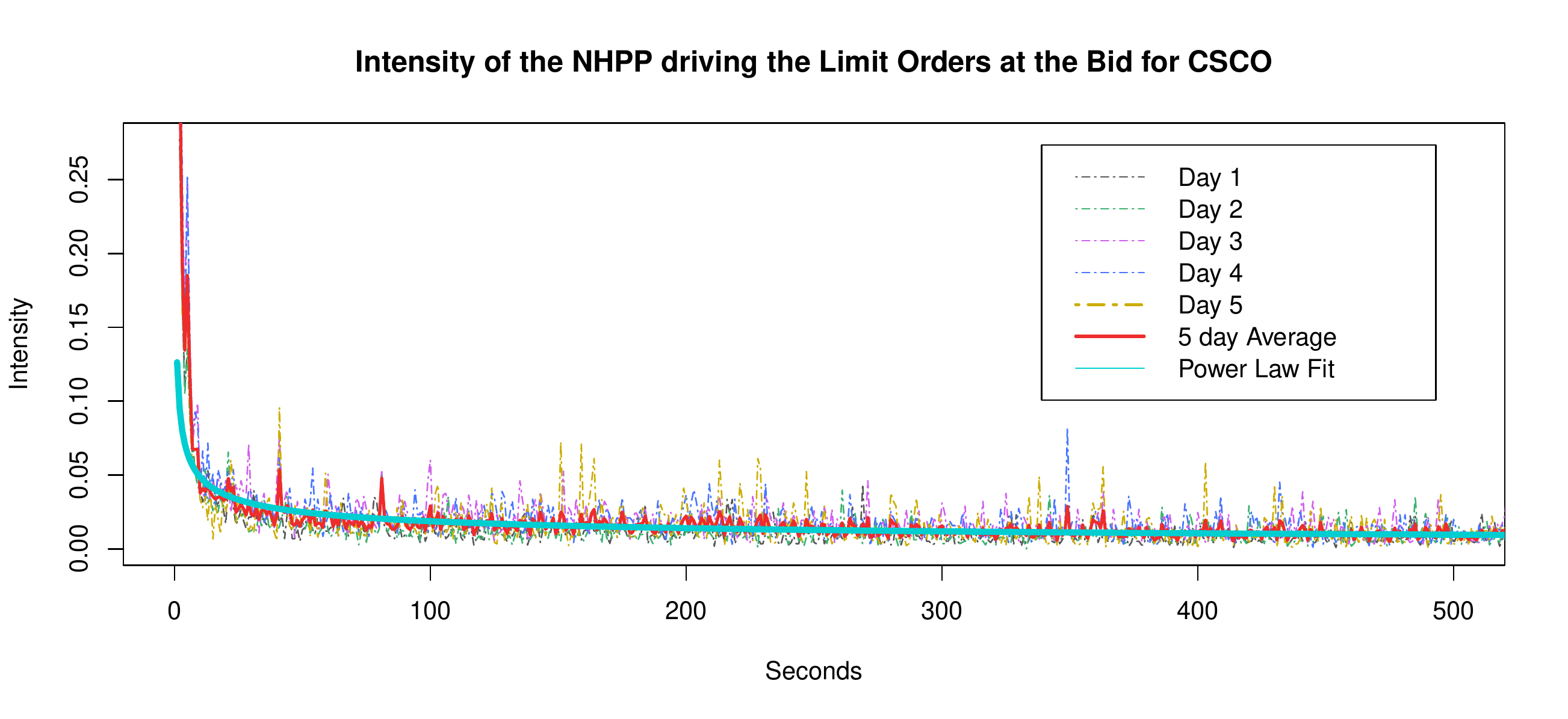}
		\includegraphics[width=\textwidth]{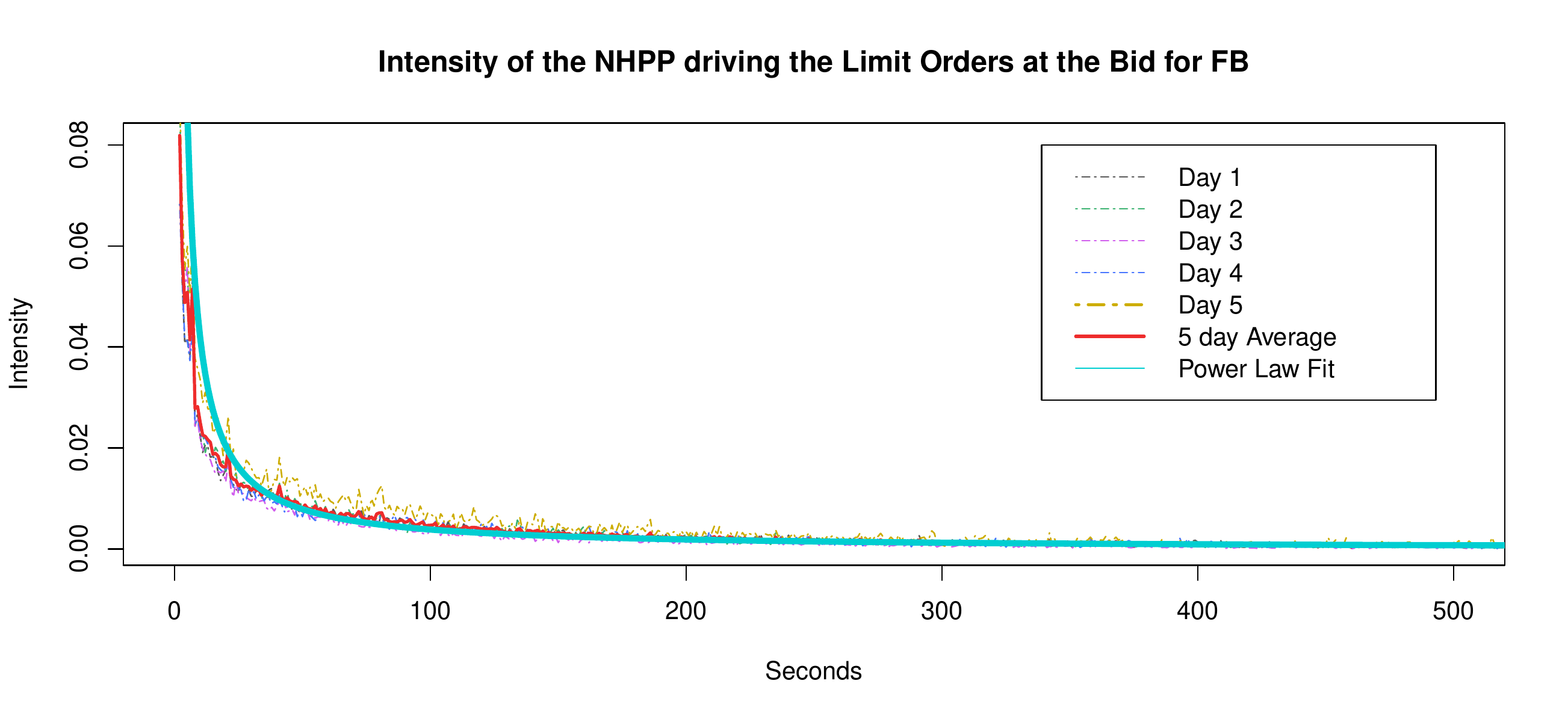}
		\includegraphics[width=\textwidth]{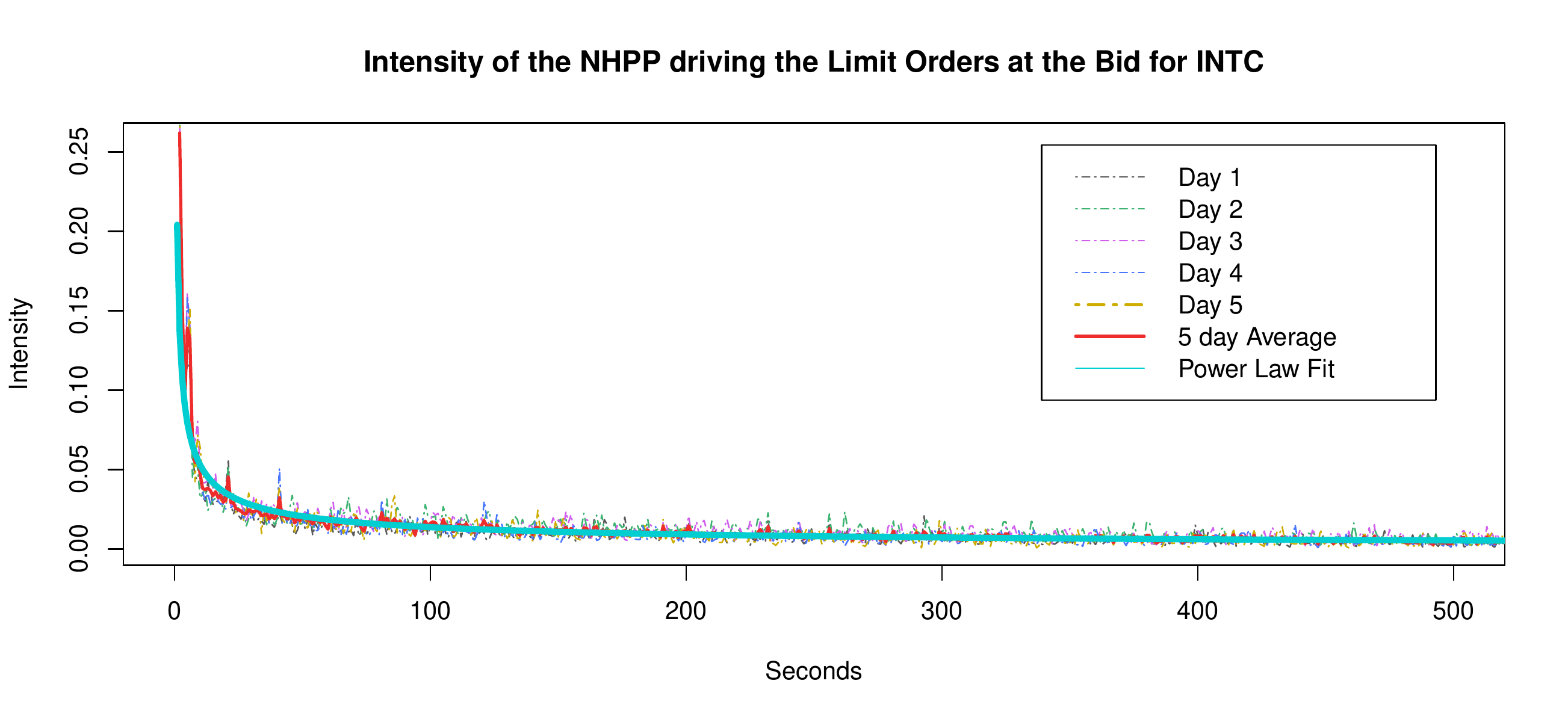}
\end{minipage}
\begin{minipage}{0.49\textwidth}
		\includegraphics[width=\textwidth]{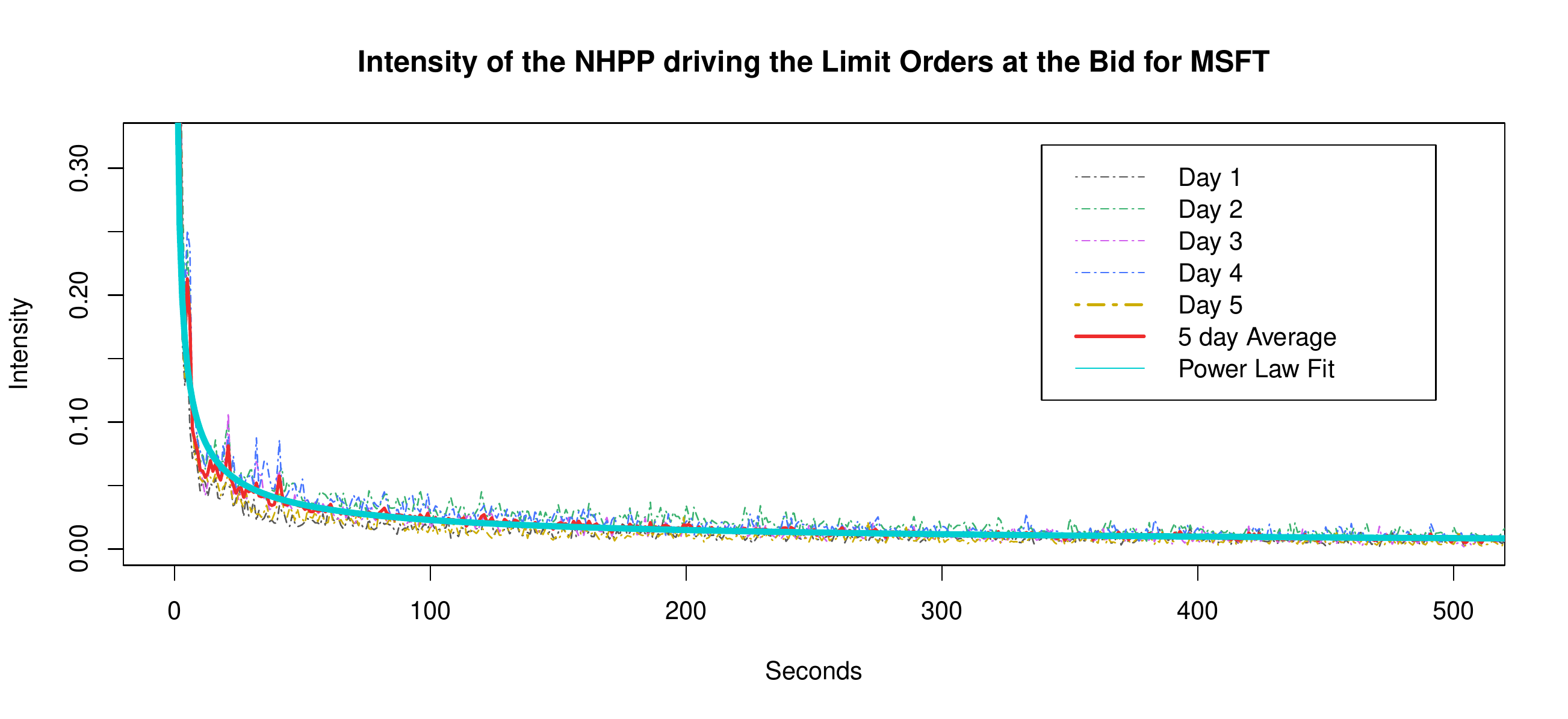}
		\includegraphics[width=\textwidth]{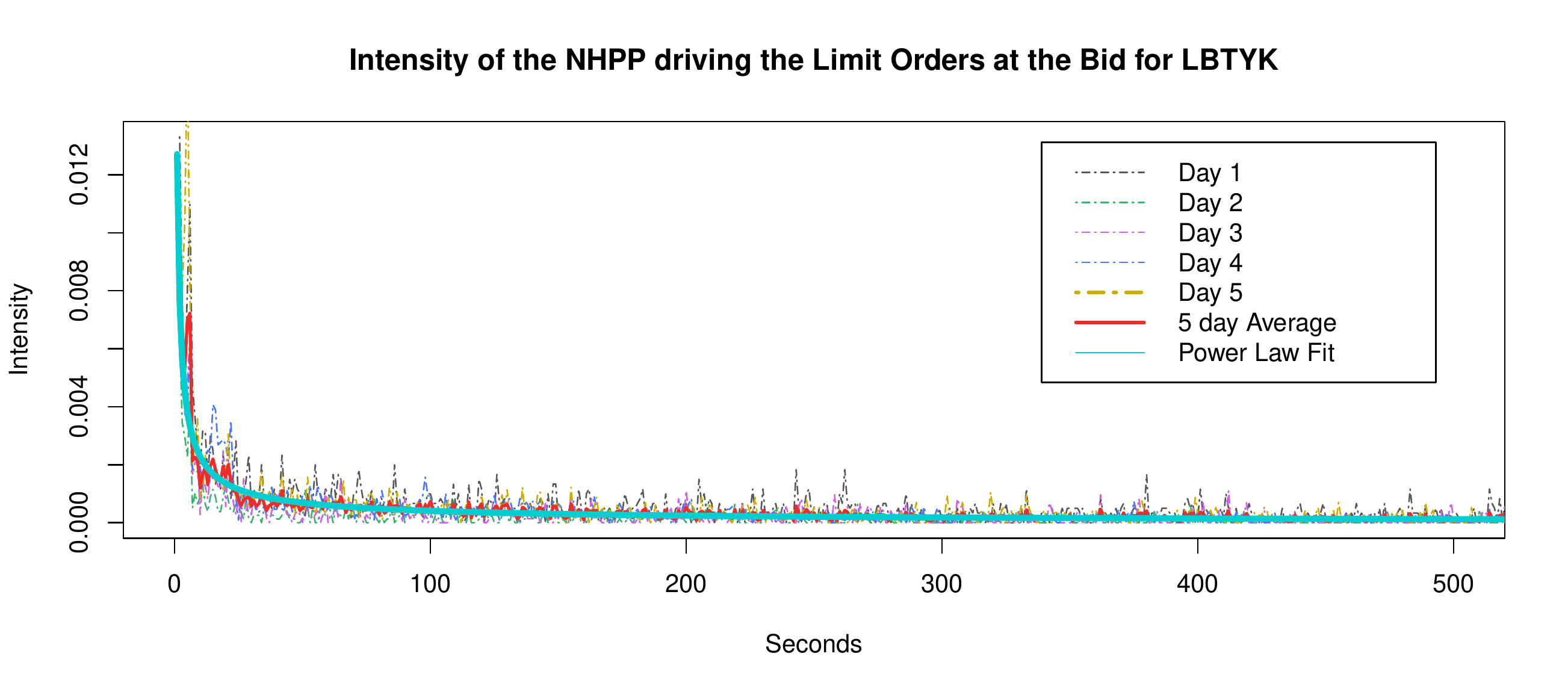}
		\includegraphics[width=\textwidth]{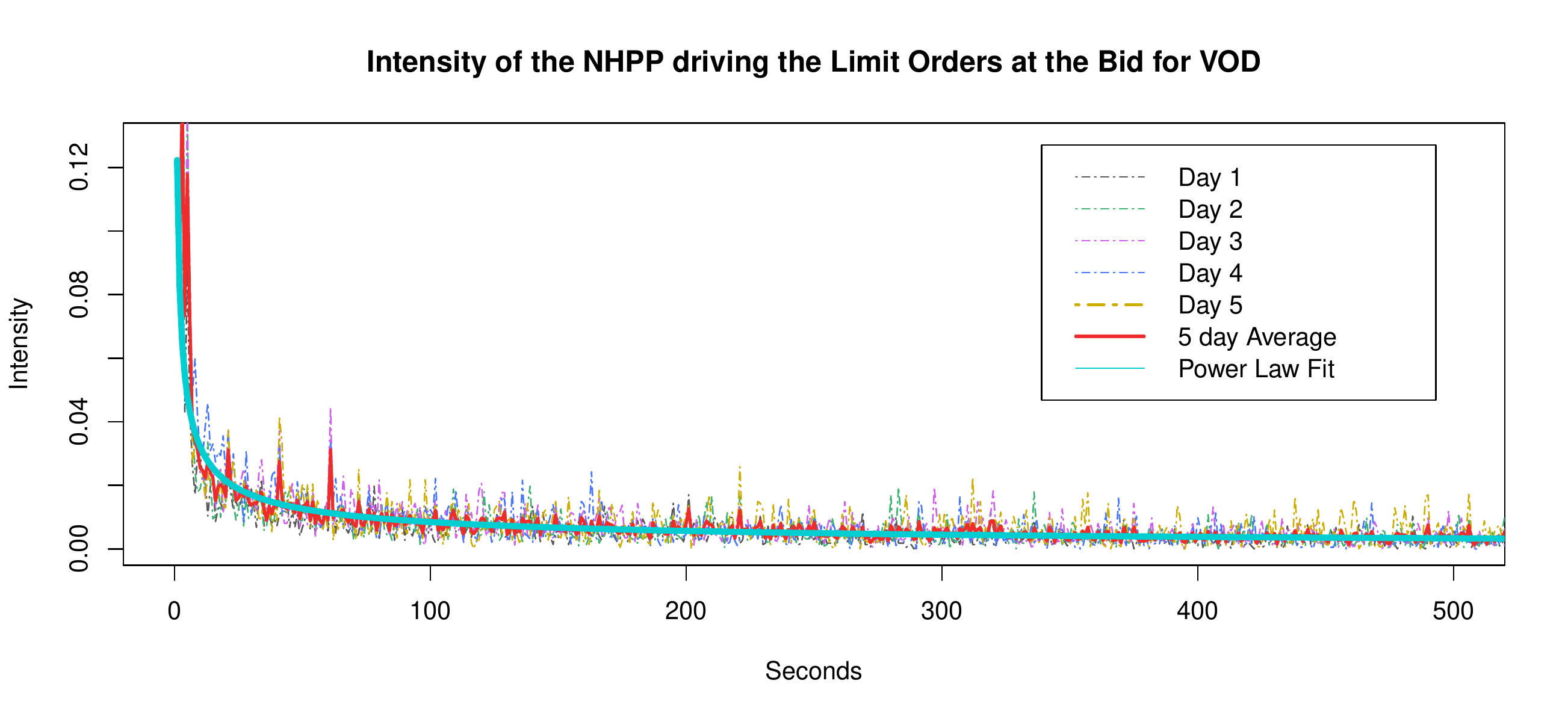}
\end{minipage}
    \captionof{figure}{Daily intensities of Limit Orders at the Bid side for the six stocks considered on the week of Nov 3rd to Nov 7th of 2014 and their corresponding power law fit.}
    \label{fig:LOBintensity}
		\bigskip
		\medskip
\end{minipage}
\end{center}

\begin{center}
\begin{minipage}[c]{0.85\textwidth}
\centering
\begin{minipage}{0.49\textwidth}
    \includegraphics[width=0.9\textwidth]{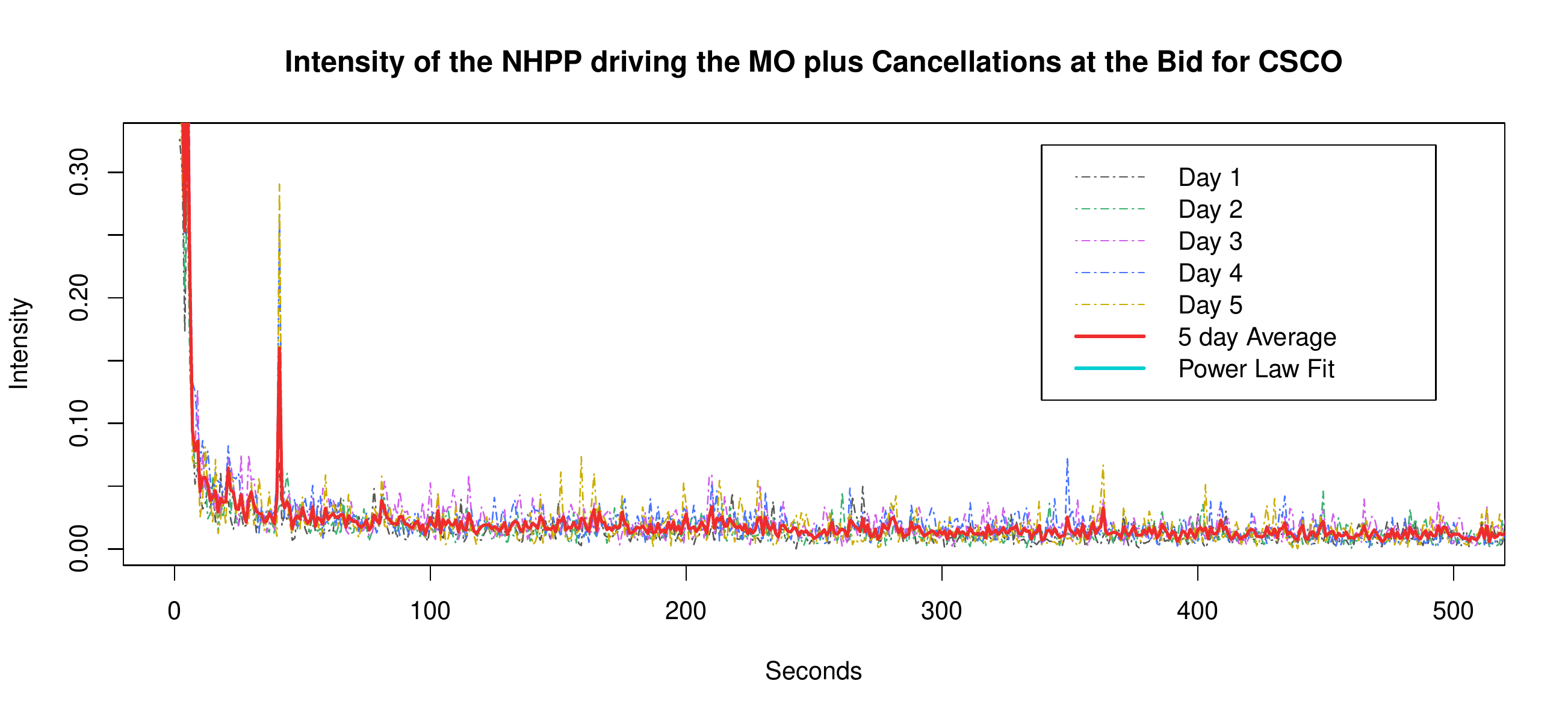}
		\includegraphics[width=0.9\textwidth]{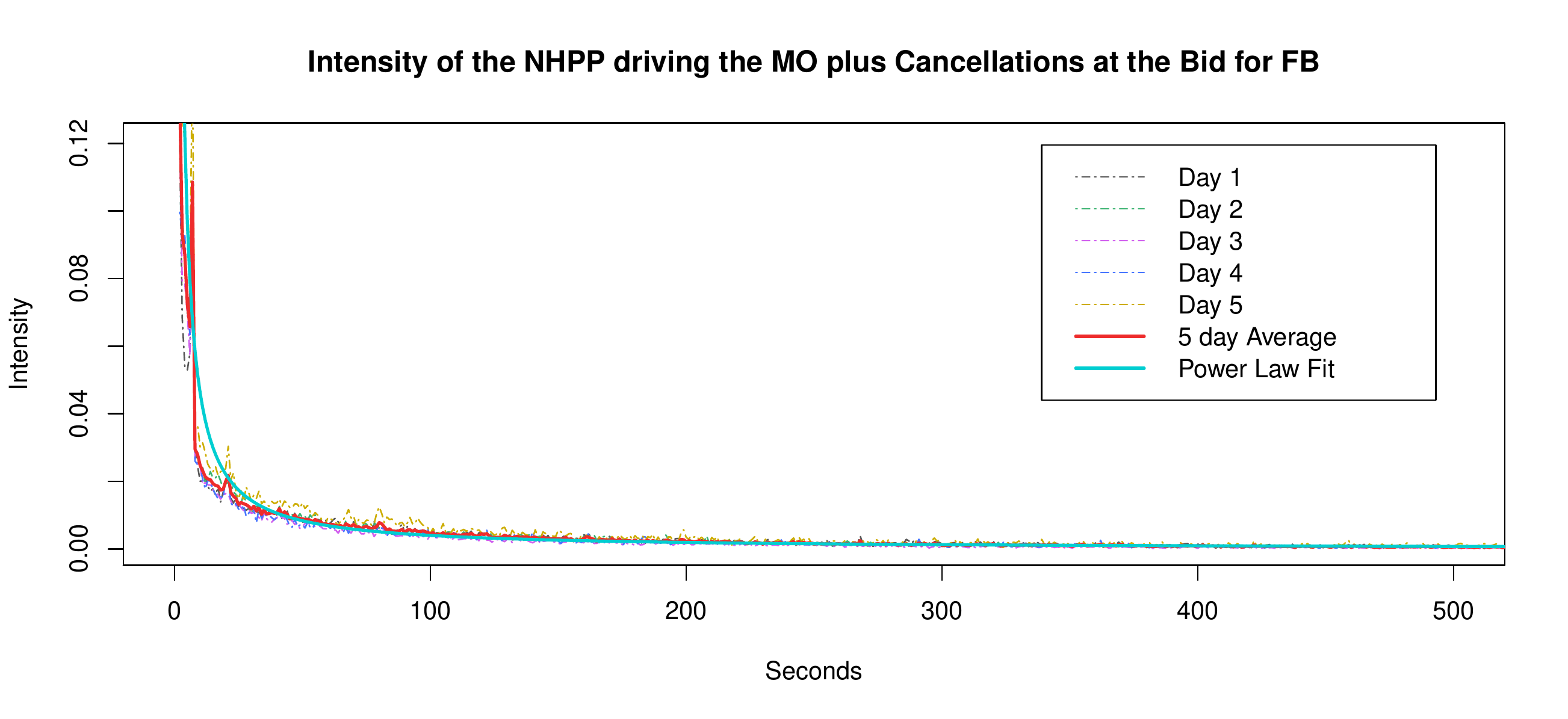}
		\includegraphics[width=0.9\textwidth]{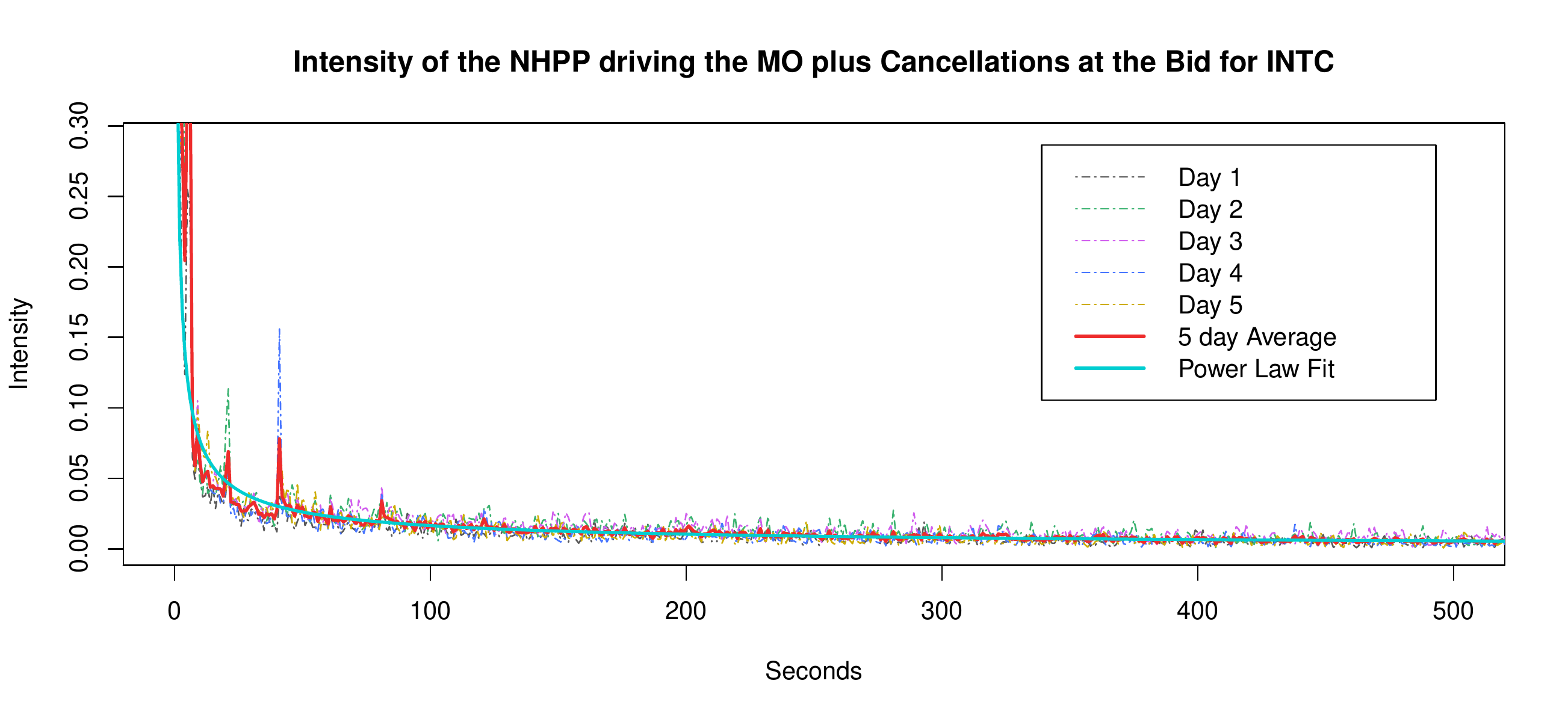}
\end{minipage}
\begin{minipage}{0.49\textwidth}
		\includegraphics[width=0.9\textwidth]{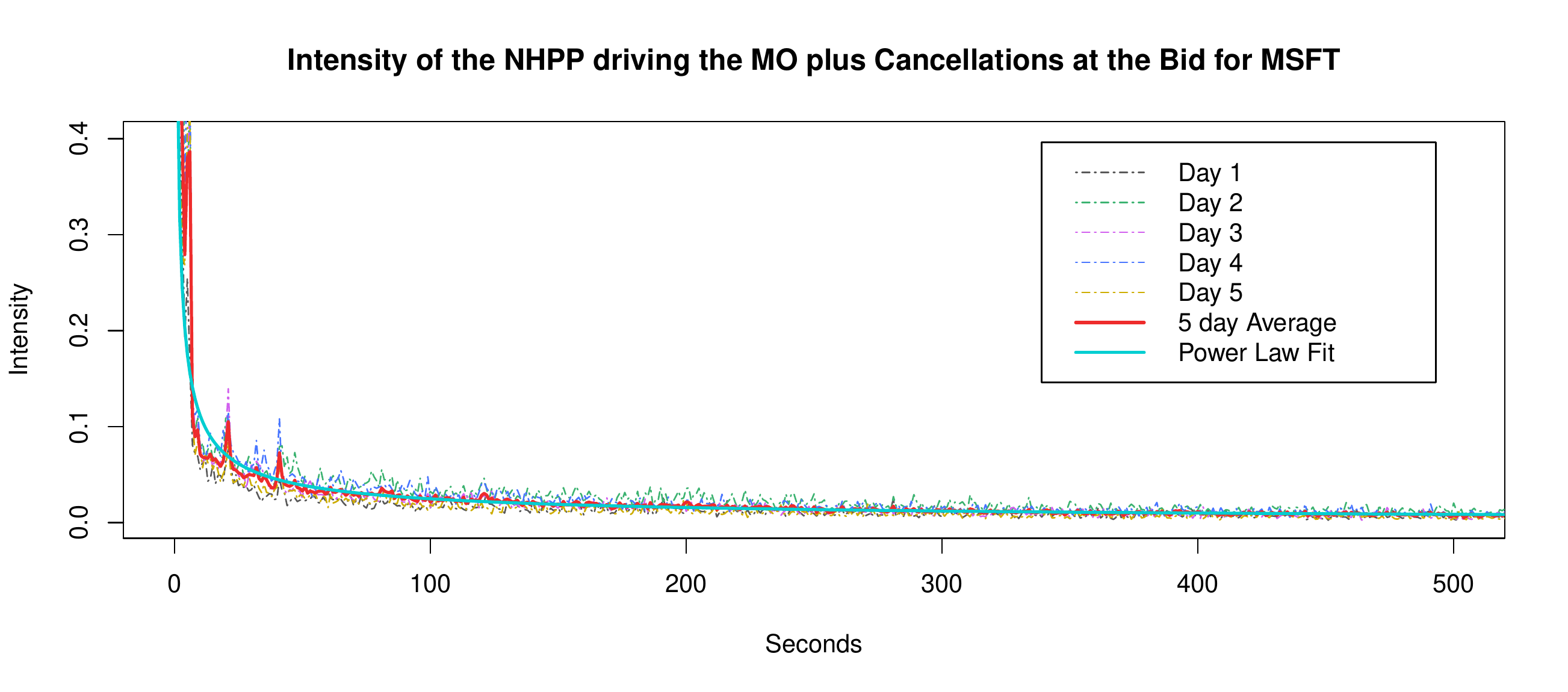}
		\includegraphics[width=0.9\textwidth]{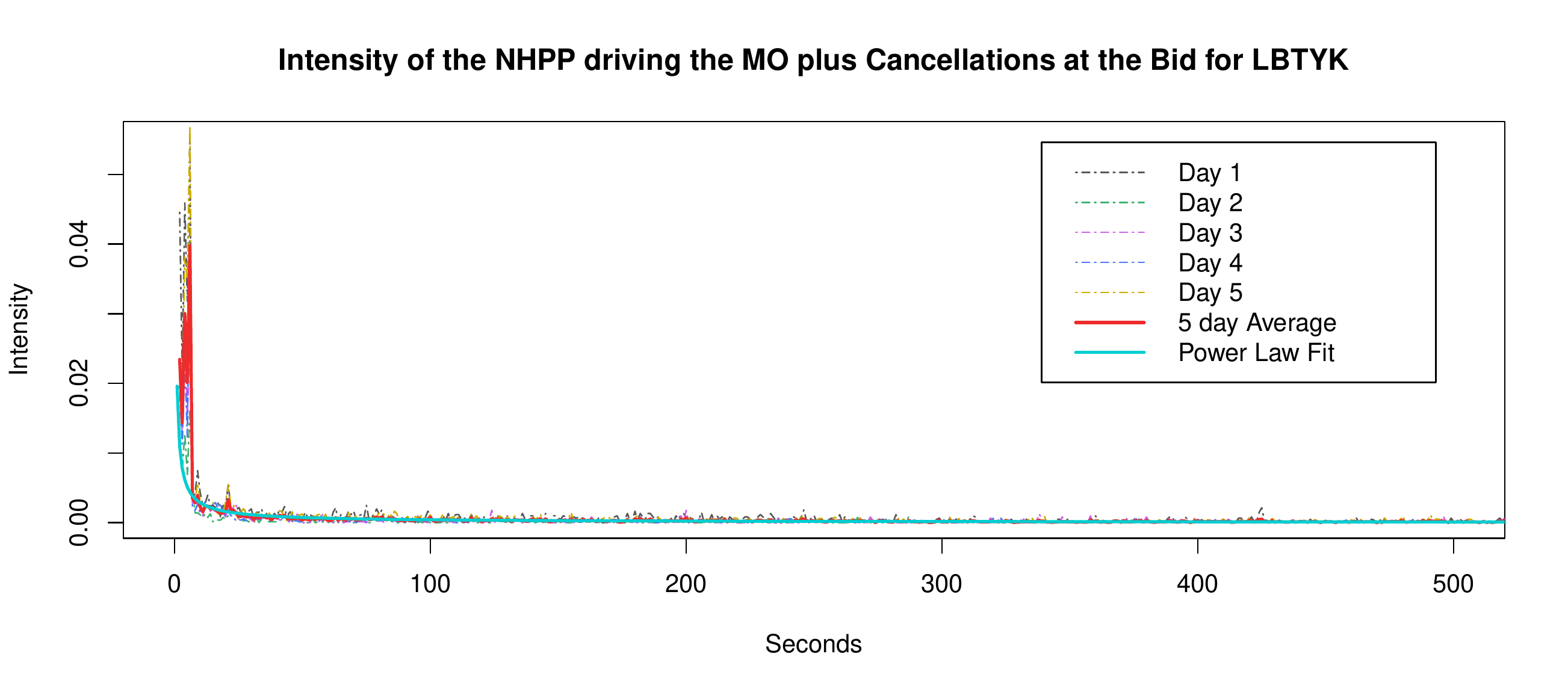}
		\includegraphics[width=0.9\textwidth]{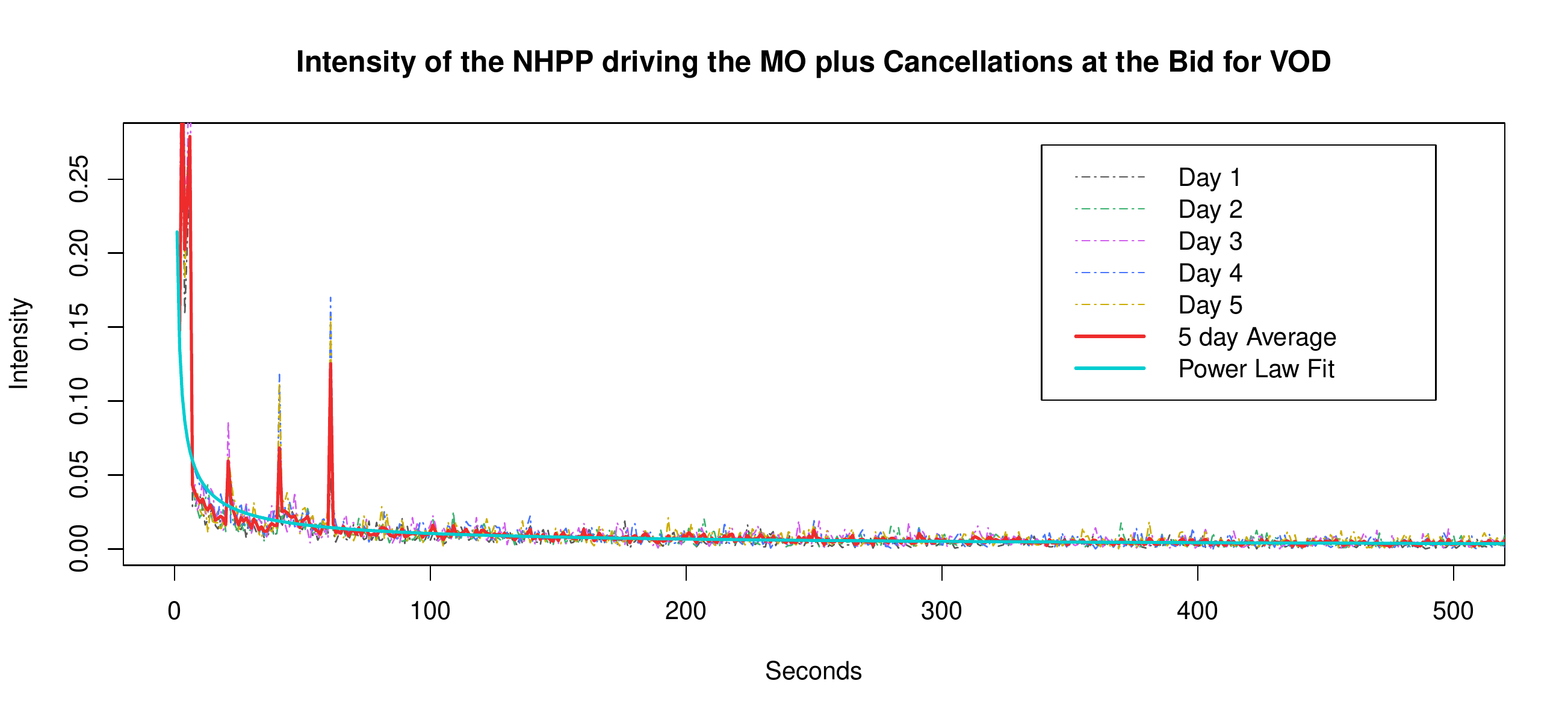}
\end{minipage}
    \captionof{figure}{Daily intensities of Marker Orders plus Cancellations on the Bid side for the six stocks considered on the week of Nov 3rd to Nov 7th of 2014 and their corresponding power law fit.}
    \label{fig:MOBintensity}
		\bigskip
		\medskip
\end{minipage}
\end{center}

In the same fashion as before, a power law fit to the intensity of Limit Orders $\lambda_t^b$ at the bid side was fit in each of the six stocks analyzed. That is, a regression is performed to fit $\lambda_t^b\approx\frac{K_{\lambda,b}}{t^s}$. Similarly, a power law fit to the intensity of Market Orders plus Cancellations $\mu_t^b$ at the bid side was fit in each of the six stocks analyzed. In this case, a regression is performed to fit $\mu_t^b\approx\frac{K_{\mu,b}}{t^r}$. The following table summarizes the power law fit to the intensities  of the analyzed stocks.

\renewcommand{\arraystretch}{1.4}

\begin{table}[h]
	\centering
		\begin{tabular}{|l|c|c|c|c|}\hline
			\multirow{2}{*}{\textbf{Stock}} & \multicolumn{2}{c|}{\textbf{Fit for} $\lambda_t^b\approx K_{\lambda,b}t^{-s}$} & \multicolumn{2}{c|}{\textbf{Fit for} $\mu_t^b\approx K_{\mu,b}t^{-r}$}\\ \cline{2-5}
			& \textbf{Coefficient} $K_{\lambda,b}$ & \textbf{Exponent} $s$ & \textbf{Coefficient} $K_{\mu,b}$ & \textbf{Exponent} $r$\\ \hline\hline
			CSCO  &  0.1264         & 0.4149       & 0.1775          & 0.4509  \\ \hline
			FB    &  0.4584         & 1.0039       & 0.5359          & 1.0064  \\ \hline
			INTC  &  0.2041         & 0.5872       & 0.3525          & 0.6649  \\ \hline
			MSFT  &  0.3887         & 0.6163       & 0.5014          & 0.6522  \\ \hline
			LBTYK &  0.0127         & 0.7466       & 0.0196          & 0.8352  \\ \hline
			VOD   &  0.1223         & 0.5806       & 0.2143          & 0.6566  \\ \hline
		\end{tabular}
	\caption{Regression fit for $\lambda_t^b\approx Ct^{-s}$ and $\mu_t^abapprox Dt^{-r}$ for all six stocks analyzed.}
	\label{tab:Bidinten}
	\bigskip
	\medskip
\end{table}

Finally, in order to assess how close the quotients $\lambda_t^a/\mu_t^a$ and $\lambda_t^b/\mu_t^b$ behave like constants, a plot of this quotients is presented along with their average. First the quotients at the ask side are presented and then the quotients at the bid. As it can be observed, in all cases, the quotient is less than 1 indicating that the queues are in a stationary case.

\begin{center}
\begin{minipage}[c]{0.85\textwidth}
\centering
\begin{minipage}{0.49\textwidth}
    \includegraphics[width=0.9\textwidth]{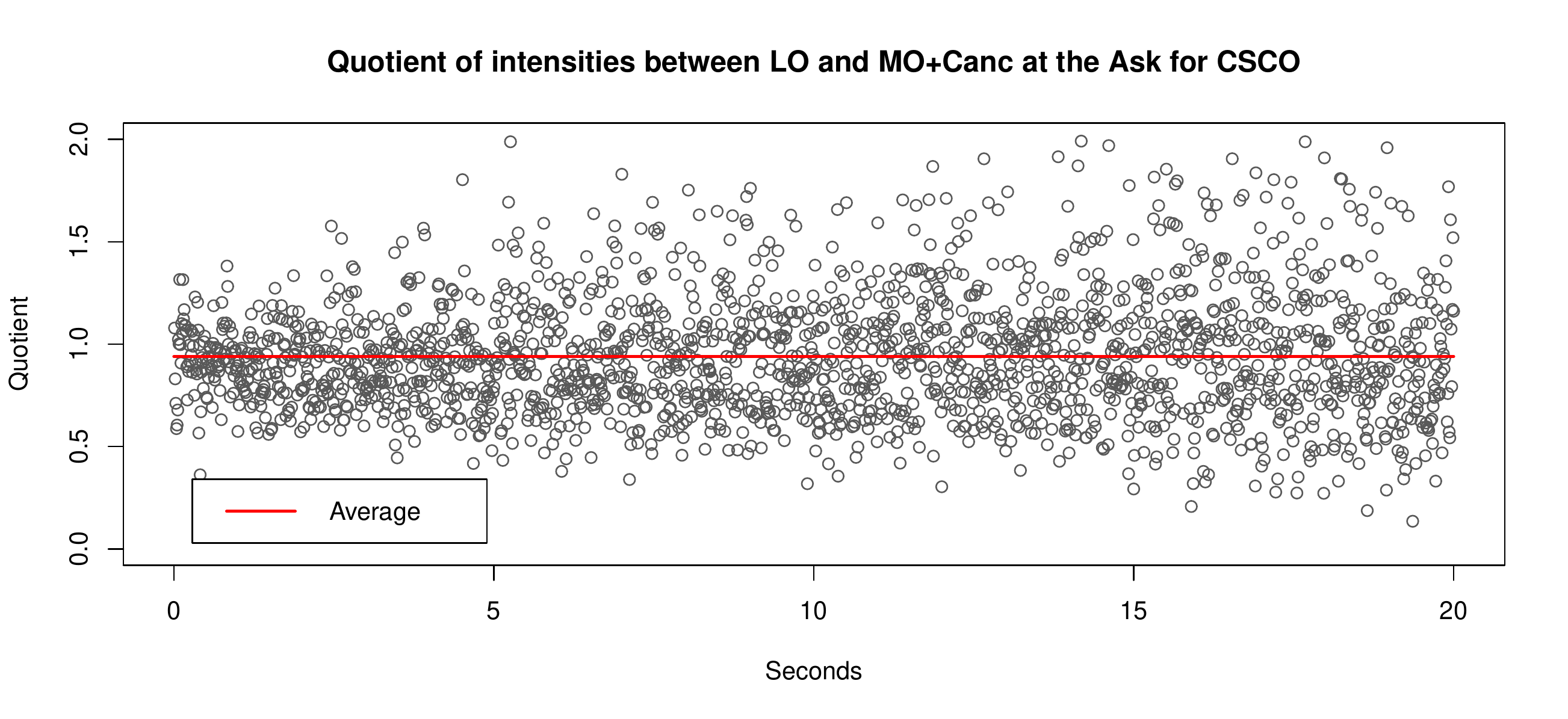}
		\includegraphics[width=0.9\textwidth]{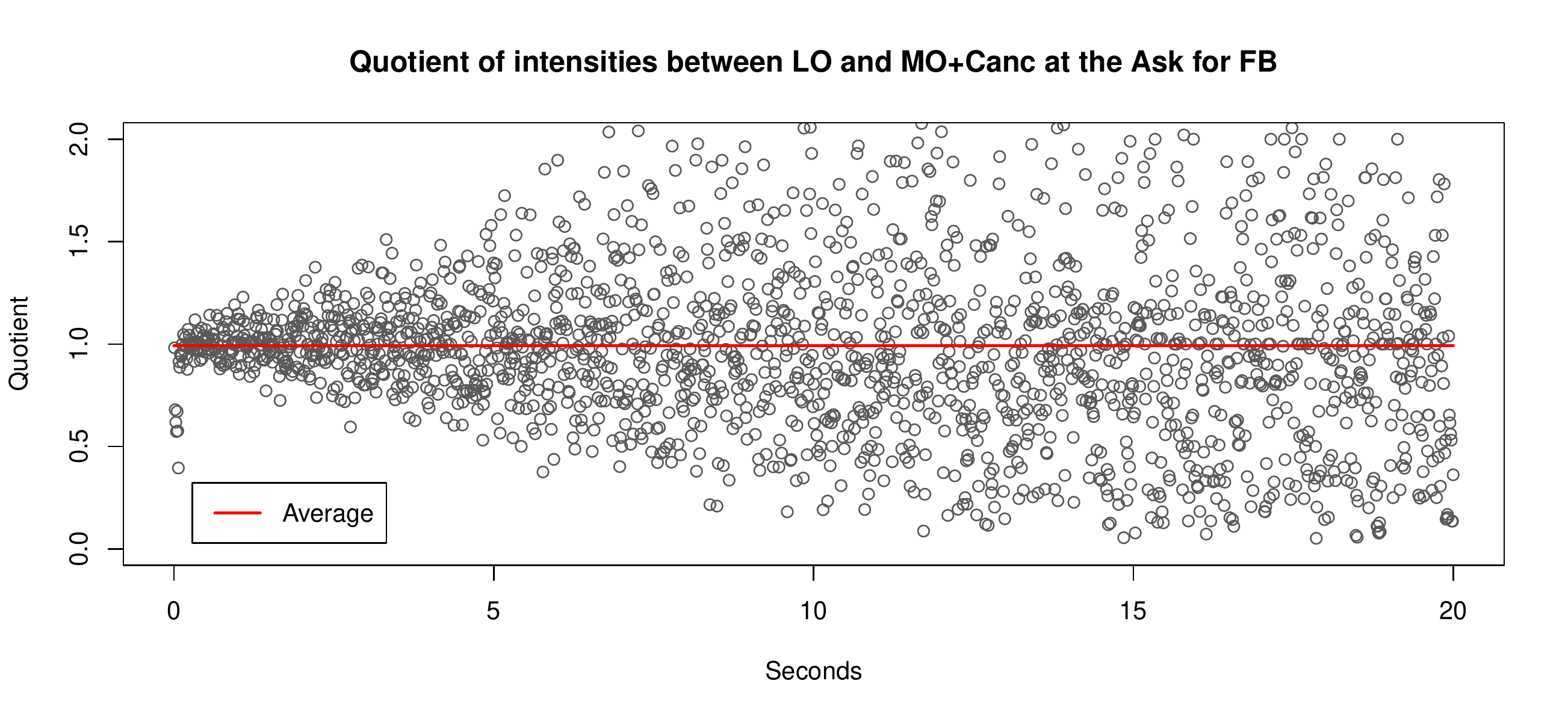}
		\includegraphics[width=0.9\textwidth]{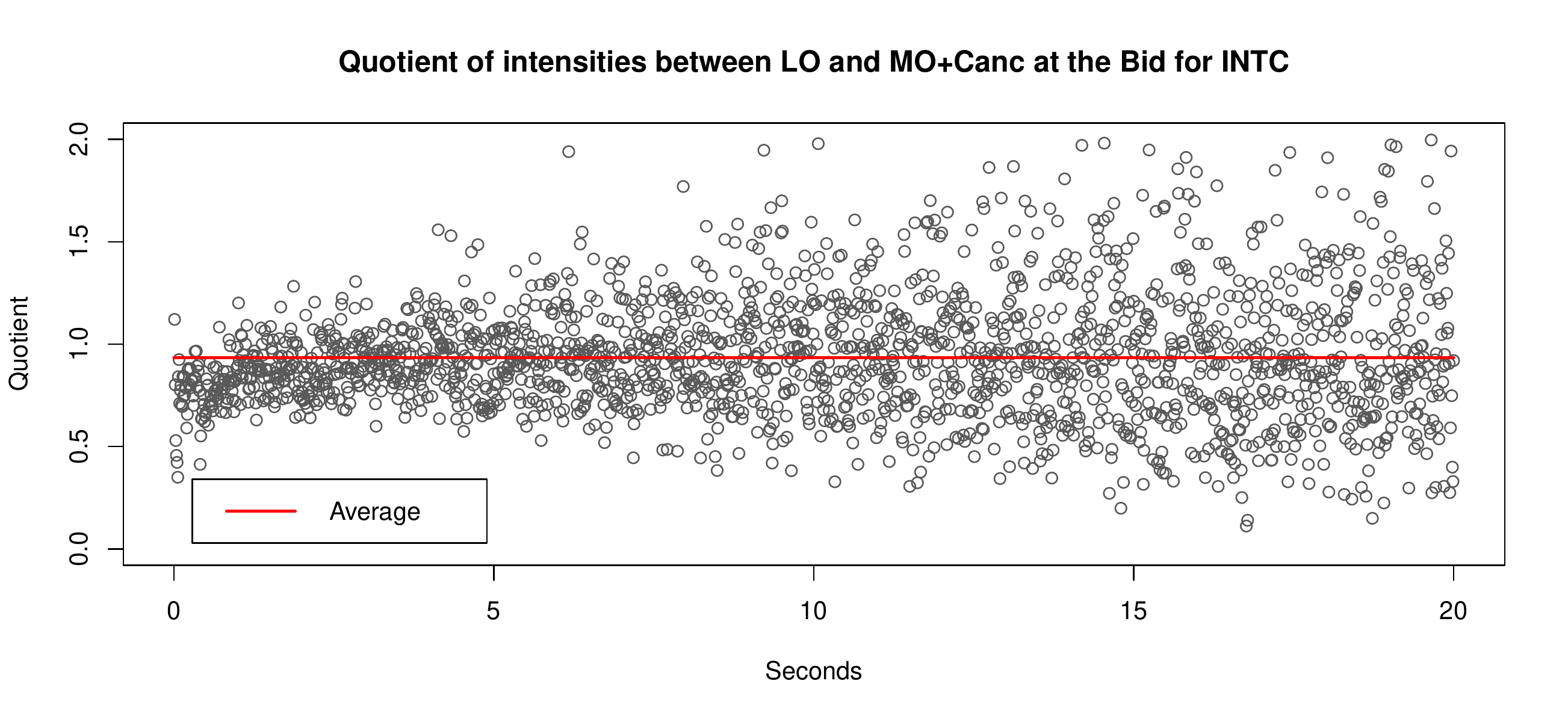}
\end{minipage}
\begin{minipage}{0.49\textwidth}
		\includegraphics[width=0.9\textwidth]{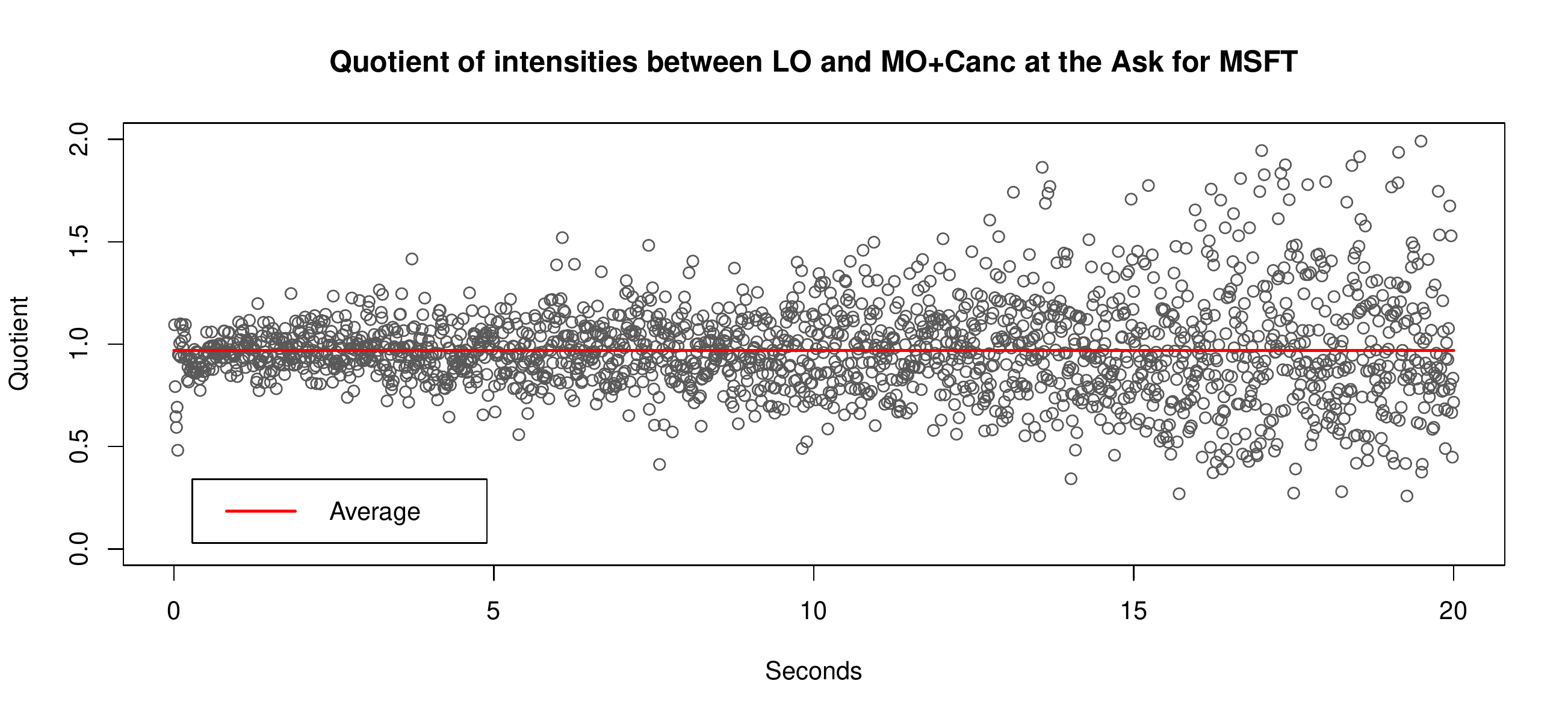}
		\includegraphics[width=0.9\textwidth]{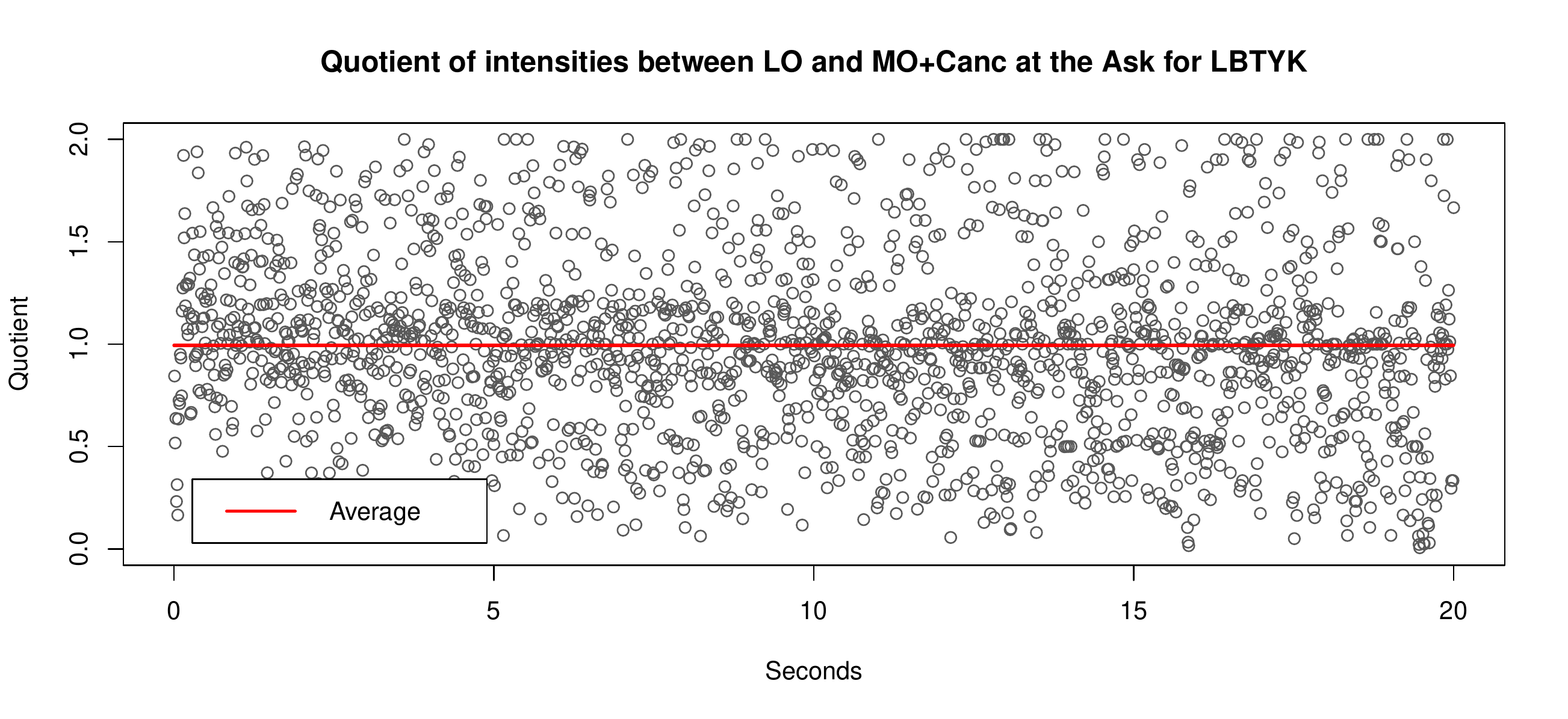}
		\includegraphics[width=0.9\textwidth]{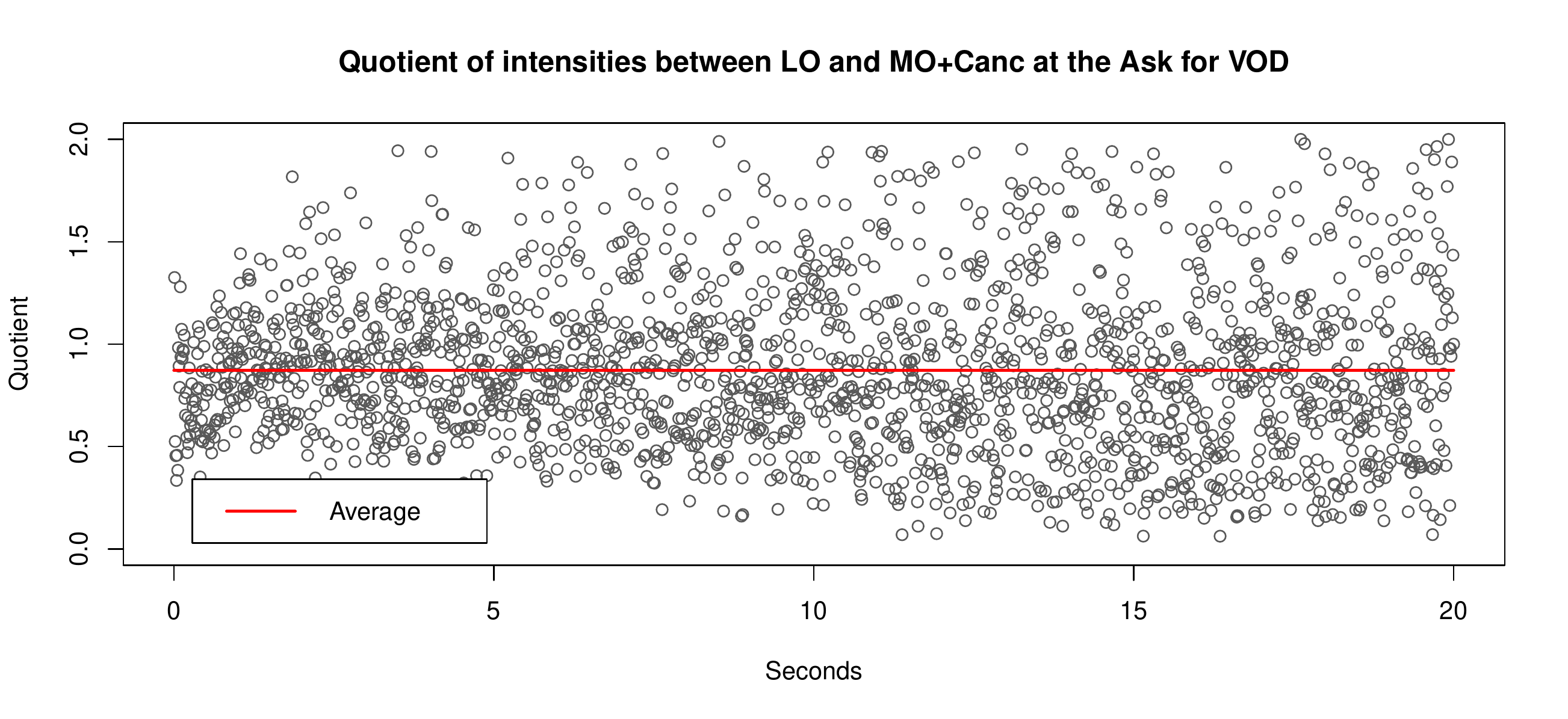}
\end{minipage}
    \captionof{figure}{Plot of the quotient $\lambda_t^a/\mu_t^a$ versus the time $t$. The assumption of a constant quotient is contrasted here.}
    \label{fig:QuotA}
		\bigskip
		\medskip
\end{minipage}
\end{center}

Next, the quotient on the bid side is displayed.

\begin{center}
\begin{minipage}[c]{0.85\textwidth}
\centering
\begin{minipage}{0.49\textwidth}
    \includegraphics[width=0.9\textwidth]{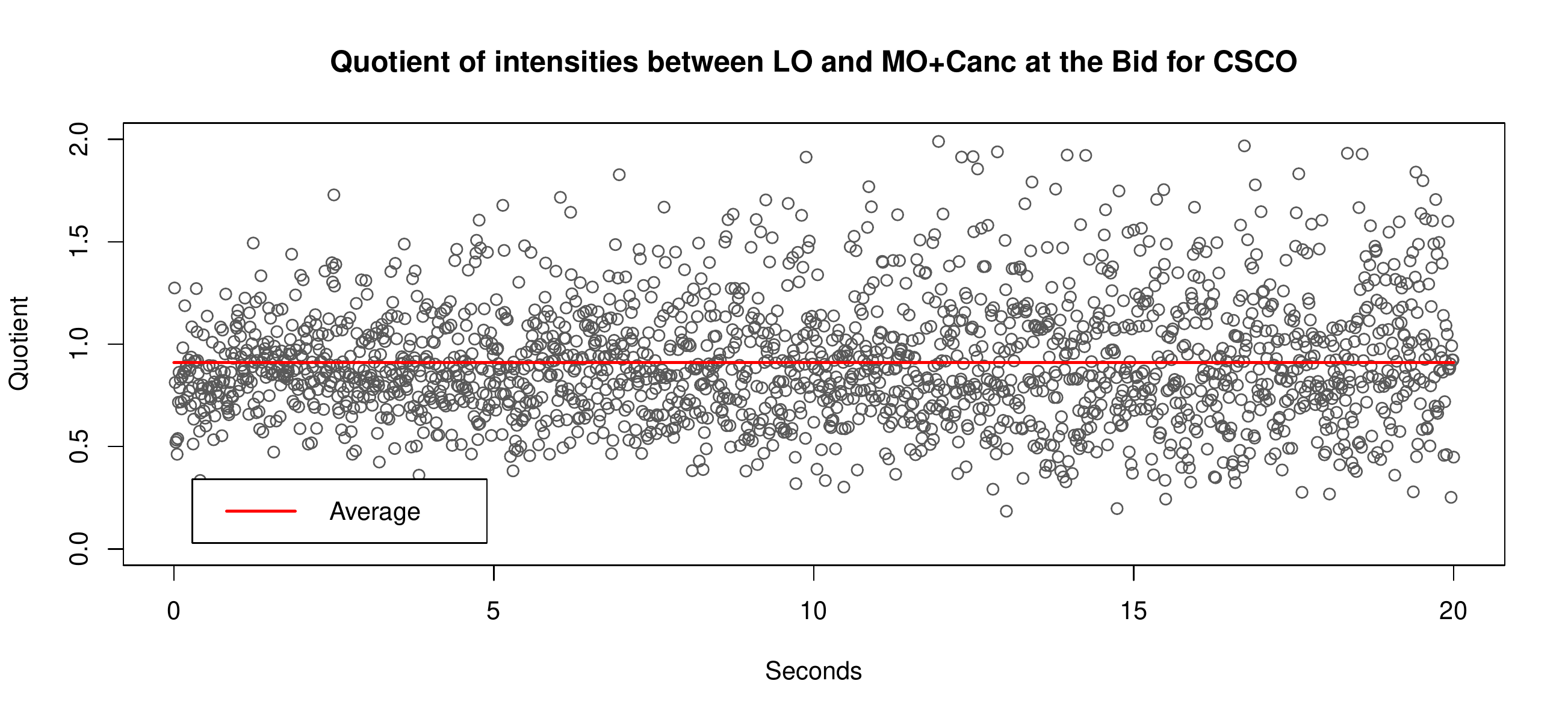}
		\includegraphics[width=0.9\textwidth]{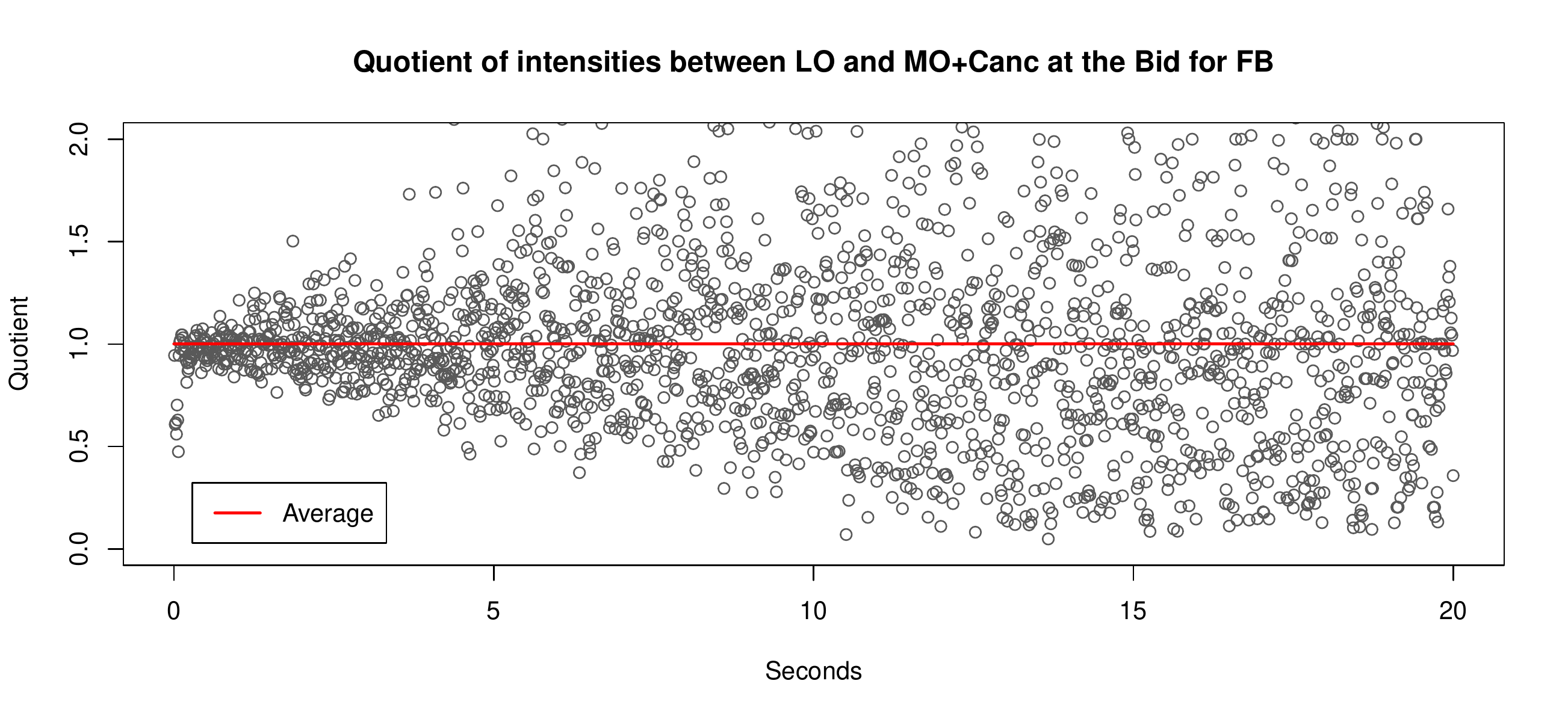}
		\includegraphics[width=0.9\textwidth]{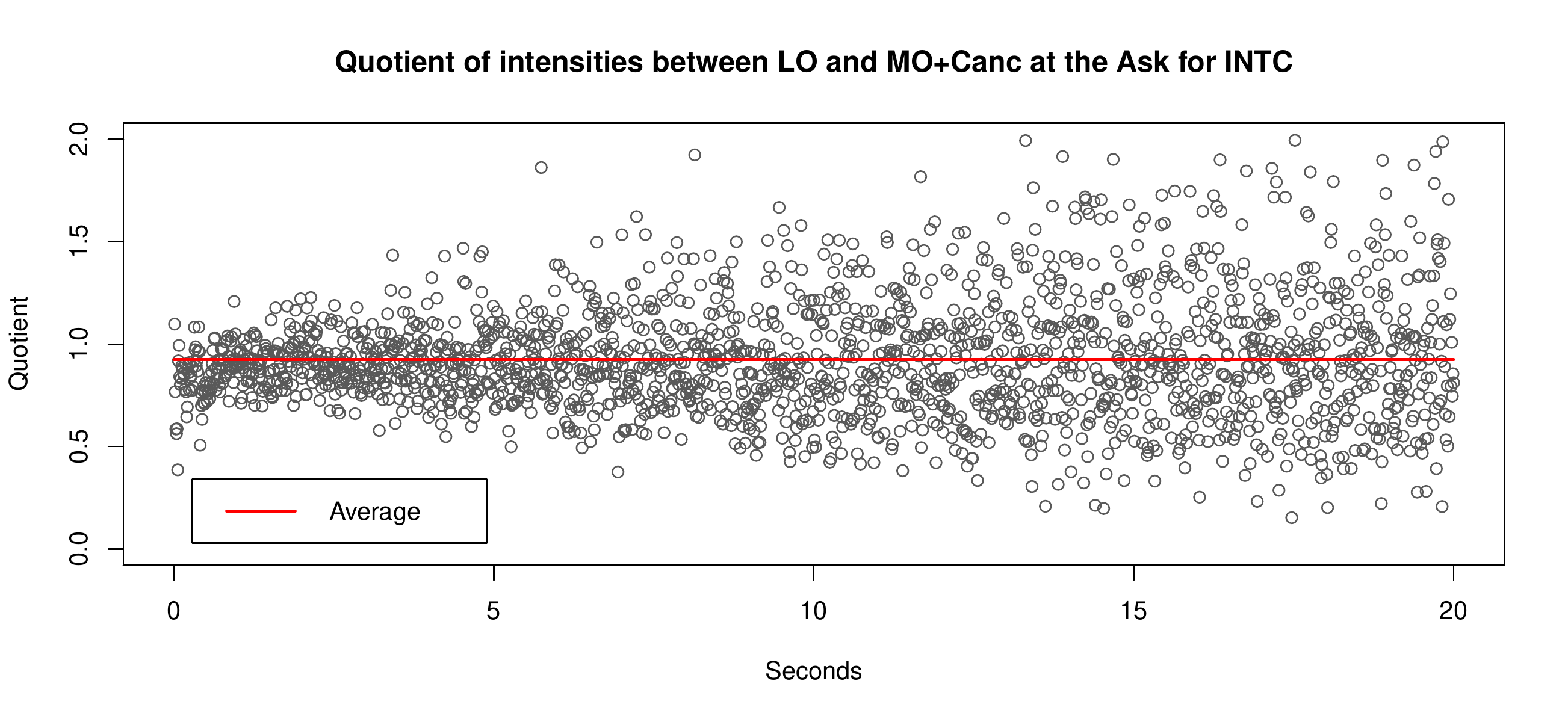}
\end{minipage}
\begin{minipage}{0.49\textwidth}
		\includegraphics[width=0.9\textwidth]{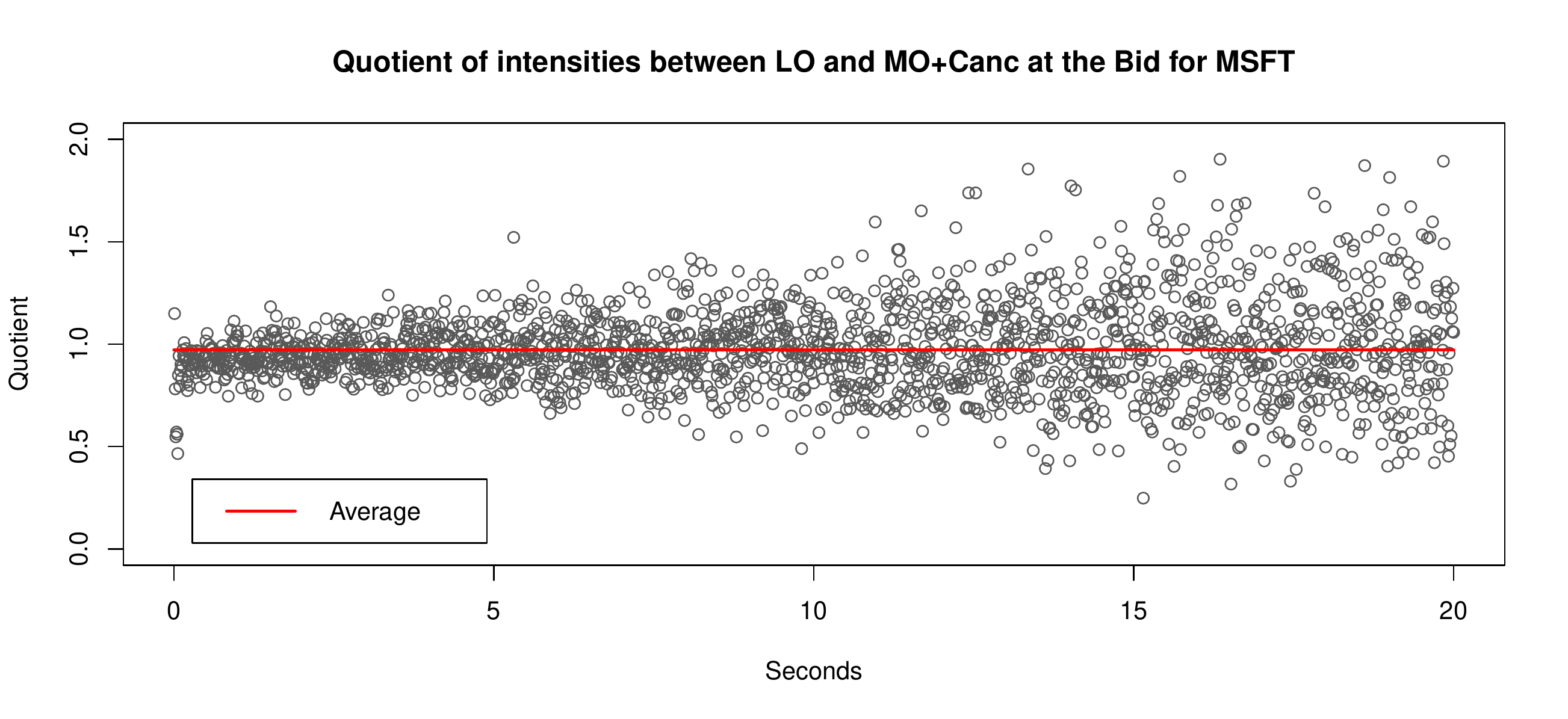}
		\includegraphics[width=0.9\textwidth]{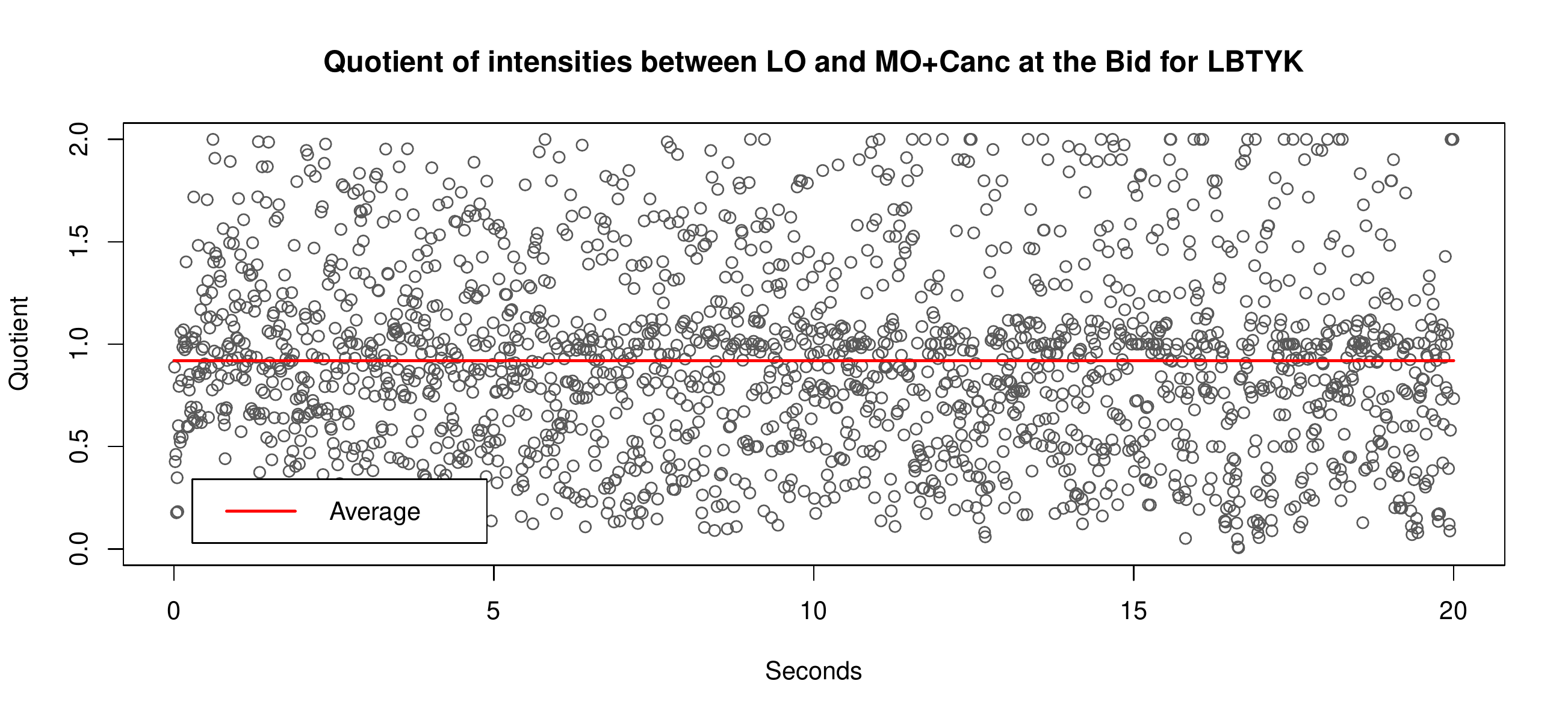}
		\includegraphics[width=0.9\textwidth]{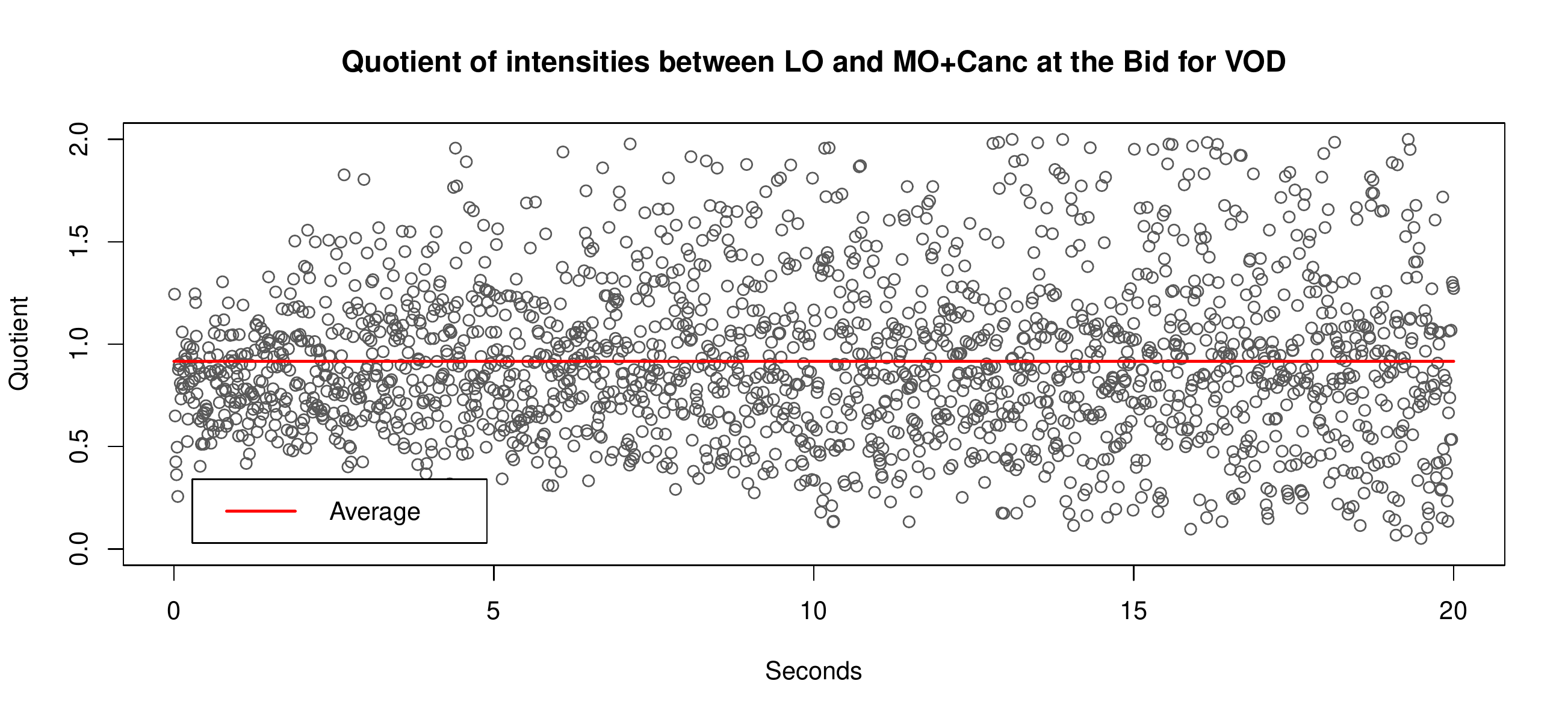}
\end{minipage}
    \captionof{figure}{Plot of the quotient $\lambda_t^b/\mu_t^b$ versus the time $t$. The assumption of a constant quotient is contrasted here.}
    \label{fig:QuotB}
		\bigskip
		\medskip
\end{minipage}
\end{center}

The last component of this section is a table that compares the mean of the quotient $\lambda_t^a/\mu_t^a$ for the ask side of all six stocks with the same quotient $\lambda_t^b/\mu_t^b$ for the bid side for all six stocks.

\begin{table}[h]
	\centering
		\begin{tabular}{|l|c|c|}\hline
			\textbf{Stock} & \textbf{Mean of the quotient} $\lambda_t^a/\mu_t^a$ & \textbf{Mean of the quotient} $\lambda_t^b/\mu_t^b$\\ \hline\hline
			CSCO  &  0.9598         & 0.9392  \\ \hline
			FB    &  0.9927         & 0.9993  \\ \hline
			INTC  &  0.9441         & 0.9544  \\ \hline
			MSFT  &  0.9901         & 0.9912  \\ \hline
			LBTYK &  0.9998         & 0.9498  \\ \hline
			VOD   &  0.8919         & 0.9255  \\ \hline
		\end{tabular}
	\caption{Comparison of the mean quotient of the intensities for limit orders vs market orders plus cancellations at the ask side and the bid side for all six stocks analyzed.}
	\label{tab:Quotinten}
	\bigskip
	\medskip
\end{table}

\section{Conclusions and Further Research} \label{sec:conclusions}
For this paper, a simple limit order book model was proposed with the intention to further study the empirical features of these complicated systems. In particular, this paper tries to focus and understand the empirical features of the inter-arrival times between order submissions and how do these empirical features may affect the fluid dynamics of the price process. Indeed, as shown in section 3, depending on the speed at which the density of the times between arrivals of orders decays, the long-run dynamics of the price process might have a time-dependent volatility. This is an important feature  because it will be interesting to find conditions under which the long-run dynamics of the price process possess converges to a more general Ito diffusion than a simple Brownian motion with constant volatility, say a geometric Brownian motion, which is one of the most used models for stock prices. Further, many of the existing models, to the knowledge of the author, that try to achieve these convergence define a point process that counts the arrivals of the price changes but not of the orders. That is, many work at a mesoscopic level but not at the microscopic level generated by accounting for the individual orders.

\smallskip

In this intent to create simple models that generalize the long-run dynamics of the price process, this paper has shown that different cases might arise. For example, while CSCO, INTC and VOD exhibit a quotient $\lambda_t/\mu_t<1$, implying that they will fall under the case covered in Theorem \ref{thm:main:part1:1} and since all of them have a tail that decays as a power law with exponent different from $-1$, they will converge to a simple Brownian motion with constant volatility. However, for the other three stocks (FB, MSFT and LBTYK), it can be seen that their quotient $\lambda_t/\mu_t<1$ is significantly close to 1, implying that they will fall under the case covered in Theorem \ref{thm:main:part1:2}, and in here, two cases arise: while MSFT and LBTYK have a tail that decays slower than $t^{-1}$ and thus will converge to a Brownian motion with a time-dependent volatility, FB exhibits a tail that is barely heavier that $t^{-1}$, implying that it will converge to a normal Brownian motion with constant volatility.

\smallskip

As it can be seen with this small sample of stocks, many different scenarios have arose, implying that these conditions imposed in the model are attainable. Of course, the model has some limitations and many simplifications took place, but the author believes that this is the first step towards working in obtaining more realistic models such as a GBM. A good example of how these models have been found but where the taken scale is a mesoscopic one is provided on \cite{Rosenbaum2015}, where the authors use almost unstable Hawkes processes to achieve convergence to a GBM starting from modeling the times of arrivals of the price changes. An interesting model would then become to consider how to use similar results to achieve convergence to such processes starting from the arrival of individual orders and no from the aggregated data. The difficult part in all of these models is to understand the tail behaviour of the stopping time that signals a price change, such as the one provided in Lemma \ref{lemma:asympt:tau}. While the author believes that many more interesting features can be achieved by substituting the inhomogeneous Poisson process by a more general point process such as a Hawkes process, or even better, a state-dependent Hawkes process (one where the intensity depends on the state of the process) the complicated part is to unravel the behavior of the aforementioned stopping time and this will become an interesting research direction for the near future.

\bibliographystyle{apalike}
\bibliography{all2005}

\newpage
\appendix

\section{Proofs of Section \ref{Sec:Properties}}\label{sec:AppendixA}

\begin{proof}[Proof of Proposition \ref{prop:distrsigma}]. The arrival processes $L_t$ and $M_t$ are Markov processes. Moreover, the queue process, describing the amount of orders at the ask, $q_t^a$ is also a Markov process with its generator given by Equation \ref{eqn:generatorXY}. That is, for any function $u\in\text{Dom}(\mathscr{L}_t)$, $t\in\R^+$ and $z\in\N$,
\begin{equation} \label{eqn:NHPP:generator}
\mathscr{L}_tu(t,z)=\lambda_t u(t,z+1) + \mu_t u(t,z-1) -(\lambda_t+\mu_t)u(t,z).
\end{equation}

Let $\bar{u}(t,x)$ be an arbitrary bounded function such that $t\mapsto \bar{u}(t,x)$ is $C_1$ for all $x$ and $(t,x)\mapsto \partial\bar{u}(t,x)/\partial t$ is bounded. Fix $T\more0$ and let $f(t,x)=\bar{u}(T-t,x)$. Under the stated conditions, $\bar{u}$ belongs to the domain of the generator $\mathscr{L}_t$ and, thus, the process
\[
f(t,q_t^a)-f(0,q_0^a)-\int_0^t\left(\frac{\partial}{\partial r}+\mathscr{L}_r\right)f(r,q_r^a)dr,\quad t\in[0,T]
\]
is a local martingale. Therefore,
\[
M_t:=\bar{u}(T-t,q_t)-\bar{u}(T,q_0)-\int_0^t\left(\frac{\partial}{\partial {r}}+\mathscr{L}_r\right)\bar{u}(T-r,q_r)dr,\quad {t\in [0,T]},
\]
is a martingale with $M_0=0$. Let $\varsigma:=T\wedge\sigma_a^1$. By the Optional Sampling Theorem,
\begin{equation}\label{eqn:OST:NHPP}
\bar{u}(T,x)=\Ex_x[\bar{u}(T-\varsigma,{q_\varsigma})]-\Ex_x\left[\int_0^\varsigma\left(\frac{\partial}{\partial {r}}+\mathscr{L}_r\right)\bar{u}(T-r,q_r)dr\right],
\end{equation}
where $\Ex_x[\cdot]:=\Ex[\;\cdot\;|\;q_0=x]$.

On the other hand, suppose that $\bar{u}(t,x)$ solves the initial value problem \ref{eqn:PDE:HNPP}. That is, $\bar{u}(t,x)$ satisfies
\[
\left\{\begin{array}{rcl}
\left(\frac{\partial}{\partial r} + \mathscr{L}_r\right) \bar{u} (T-r,z)=0 & \text{for} & 0\less r \less T,\;
z\in\N.\\
\bar{u}(T-r,0)=0& \text{for} & 0\leq r\less T.\\
\bar{u}(0,z)=1 & \text{for} & z\in\N.\end{array}\right.
\]
In that case, by \eqref{eqn:OST:NHPP},
\begin{eqnarray*}
\bar{u}(T,x) &= & \Ex [\bar{u}(T-\varsigma,q_\varsigma)]\\
		     &= &  \Ex [\bar{u}(T-\varsigma,q_\varsigma)\I\{\sigma_a^1\leq T\}+\bar{u}(T-\varsigma,q_\varsigma)\I\{\sigma_a^1\more T\}]   \\
				&= &\Ex [\bar{u}(T-\sigma_a^1,0)\I\{\sigma_a^1\leq T\} + \bar{u}(0,q_T)\I\{\sigma_a^1\more T\}]\\
				&= &\Px [\sigma_a^1(x)>T].
\end{eqnarray*}
This implies that $\bar{u}(T,x)=\Px[\sigma_a^1(x)> T]$.

\end{proof}

\begin{proof}[Proof of Lemma \ref{lemma:tail:sigma}] According to \cite{Olver:2010}[Formula 10.30.4], for fixed $\nu$,
\[
I_\nu(x)\sim \frac{e^x}{\sqrt{2\pi x}}\qquad\qquad\text{as}x\to\infty.
\]
Thus, as $T\to\infty$,
\begin{eqnarray*}
\Px[\sigma_{a,\Qx}^1\more T\;|\;q_0^a=x]&= &\left(\frac{\mu}{\lambda}\right)^{x/2}\int_T^{\infty}\frac{x}{s}I_x\left(2s\sqrt{\lambda\mu}\right)e^{-s(\lambda+\mu)}ds\\
		&\sim &\left(\frac{\mu}{\lambda}\right)^{x/2}\int_T^{\infty}\frac{x}{s} \frac{e^{2s\sqrt{\lambda\mu}}}{\sqrt{4s\pi\sqrt{\lambda\mu}}}e^{-s(\lambda+\mu)}ds\\
		&\sim &\left(\frac{\mu}{\lambda}\right)^{x/2}\int_T^{\infty}\frac{x}{2\sqrt{\pi\sqrt{\lambda\mu}}}s^{-3/2}e^{-s(\sqrt{\mu}-\sqrt{\lambda})^2}ds\\
\end{eqnarray*}
Consequently, if $\lambda=\mu$,
\begin{align*}
\Px[\sigma_{a,\Qx}^1\more T\;|\;q_0^a=x] &\sim \int_T^{\infty}\frac{x}{2\lambda\sqrt{\pi}}s^{-3/2}ds\\
		&\sim \frac{x}{2\lambda\sqrt{\pi}} \frac{2}{\sqrt{T}}\\
		&\sim \frac{x}{\lambda\sqrt{\pi}} \frac{1}{\sqrt{T}}.
\end{align*}
This agrees with the result proved in \cite{ContLarrard2012}. However, if $\lambda\less\mu$,
\begin{align*}
\Px[\sigma_{a,\Qx}^1\more T\;|\;q_0^a=x] &\sim \left(\frac{\mu}{\lambda}\right)^{x/2}\int_{T(\sqrt{\mu}-\sqrt{\lambda})^2}^{\infty}\frac{x}{2\sqrt{\pi\sqrt{\lambda\mu}}}\frac{(\sqrt{\mu}-\sqrt{\lambda})^3}{u^{3/2}}e^{-u}\frac{du}{(\sqrt{\mu}-\sqrt{\lambda})^2}\\
		&\sim \left(\frac{\mu}{\lambda}\right)^{x/2} \frac{x(\sqrt{\mu}-\sqrt{\lambda})}{2\sqrt{\pi\sqrt{\lambda\mu}}} \int_{T(\sqrt{\mu}-\sqrt{\lambda})^2}^{\infty}u^{-3/2}e^{-u}du\\
		&= \left(\frac{\mu}{\lambda}\right)^{x/2} \frac{x(\sqrt{\mu}-\sqrt{\lambda})}{\sqrt{\pi\sqrt{\lambda\mu}}} \left[\frac{e^{-T(\sqrt{\mu}-\sqrt{\lambda})^2}}{(\sqrt{\mu}-\sqrt{\lambda})\sqrt{T}} - \Gamma\left(\frac{1}{2},T(\sqrt{\mu}-\sqrt{\lambda})^2\right)\right]\\
		&\sim \left(\frac{\mu}{\lambda}\right)^{x/2} \frac{x(\sqrt{\mu}-\sqrt{\lambda})}{\sqrt{\pi\sqrt{\lambda\mu}}} \left[\frac{e^{-T(\sqrt{\mu}-\sqrt{\lambda})^2}}{(\sqrt{\mu}-\sqrt{\lambda})\sqrt{T}} - T^{-3/2}e^{-T}+O(T^{-1})\right]\\
		&= \left(\frac{\mu}{\lambda}\right)^{x/2} \frac{x(\sqrt{\mu}-\sqrt{\lambda})}{\sqrt{\pi\sqrt{\lambda\mu}}} \left[\frac{e^{-T(\sqrt{\mu}-\sqrt{\lambda})^2}}{(\sqrt{\mu}-\sqrt{\lambda})\sqrt{T}} +o(T^{-1/2})\right]\\
		&\sim \left(\frac{\mu}{\lambda}\right)^{x/2} \frac{x(\sqrt{\mu}-\sqrt{\lambda})}{\sqrt{\pi\sqrt{\lambda\mu}}}\frac{e^{-T\mathcal{C}}}{(T\mathcal{C})^{1/2}},
\end{align*}
where in the second to last asymptotic expansion we used Formula 8.11.2 in \cite{Olver:2010}

Let $\mathcal{C}=(\sqrt{\mu}-\sqrt{\lambda})^2$. To compute the expectation in the case where $\lambda=\mu$, notice that, for large enough $T$,
\[
\Ex\left[\sigma_{a,\Qx}^1;|\;q_0^a=x\right]= \int_0^\infty \Px[\sigma_{a,\Qx}^1\more t;|\;q_0^a=x] dt\geq\frac{x}{\lambda\sqrt{\pi}} \int_0^\infty \frac{1}{\sqrt{t}}dt = \infty,
\]
whereas if $\lambda\less \mu$, for a sufficiently large $T$, there are finite constants $\widehat{C_1}$ and $\widehat{C_2}$ such that for any $n\geq1$,
\begin{align*}
\Ex\left[\left(\sigma_{a,\Qx}^1\right)^k;|\;q_0^a=x\right]&= n\int_0^\infty t^{n-1}\Px[\sigma_{a,\Qx}^1\more t] dt\\
		&\leq \widehat{C_1} + \widehat{C_2}\int_T^\infty t^{n-1}\Big( (t\mathcal{C})^{-1/2}e^{-t\mathcal{C}}\Big)dt\\
		&\leq  \widehat{C_1} + \widehat{C_2}\int_T^\infty t^{n-1 - 1/2}e^{-t\mathcal{C}}dt\\
		&= \widehat{C_1} + \widehat{C_2}\Gamma(n-1/2,T\mathcal{C}) \less\infty
\end{align*}
\end{proof}

\begin{proof}[Proof of Proposition \ref{prop:Sol:u:A1}] For $0\leq t\leq A_T$, let $\bar{w}(t,x)$ be the function defined in Equation \ref{eqn:soln:IVP:C}. For $0\leq t\leq T$, set $\bar{v}(t,x)=\bar{w}(A_{t},x)$. Then, $\bar{v}(t,x)$ belongs to the Domain of $\mathscr{L}_t$ and
\[
\frac{\partial}{\partial r}\bar{v}(T-r,x) = \frac{\partial}{\partial r}\Bigg[\bar{w}(A_{T-r},x)\Bigg] = \frac{\partial}{\partial r}\bar{w}(A_{T-r},x)\cdot (-a_r)= -Q\bar{w}(A_{T-r},x)a_r,
\]
for all $r\in(0,T)$ and $x\in\N$. However, since $\mathscr{L}_r=Qa_r$,
\[
\frac{\partial}{\partial t}\bar{v}(t,x) + \mathscr{L}_r\bar{v}(t,x) = 0
\]
for all $x\in\N$, $t\in(0,T)$. Moreover, since $\bar{w}(t,x)$ satisfies the IVP \ref{eqn:PDE:HNPP:C}, for $0\leq r\less T$, $\bar{v}(T-r,0)= \bar{w}(A_{T-r},0) = 0$ and for any $z\in\N$, $\bar{v}(0,z)=\bar{w}(A_{0},z)=\bar{w}(0,z)=1$. Thus, by Proposition \ref{prop:distrsigma}, the result follows.
\end{proof}

\begin{proof}[Proof of Lemma \ref{lemma:asympt:tau}]
By Remark \ref{rem:rel:H:Q} and Lemma \ref{lemma:tail:sigma}, the first part is straightforward. If

\begin{itemize}
	\item $\alpha_t\sim t^s\log^m(t)$ as $t\to\infty$ for some $s\neq-1$, $m\in\N\cup\{0\}$. Then, for sufficiently large $t$, $A_t\geq\hat{c}t^{s+1}$. Therefore,
		\begin{itemize}
		\item If $\lambda\less\mu$, (by the Proof of Lemma ), there are finite constants $\mathcal{C}$, $C_1$, $C_2$ and $C_3$ such that,
		\begin{eqnarray*}
		\Ex\left[\left(\tau_{\Hx}^1\right)^n\;|\;q_0^a=x, q_0^b=y\right] &= & n\int_0^\infty t^{n-1}\Px\left[\tau_{\Hx}^1\more t\;|\;q_0^a=x, q_0^b=y\right]dt\\
		&= &n\int_0^\infty t^{n-1}\Px\left[\tau_{\Qx}^1\more A_t\;|\;q_0^a=x, q_0^b=y\right]dt\\
		&\leq & T^n+n\int_T^\infty t^{n-1}\Px\left[\tau_{\Qx}^1\more A_t\;|\;q_0^a=x, q_0^b=y\right]dt\\
		&\leq & T^n + C_1 \int_T^\infty t^{n-1}\left[ \left(\mathcal{C}A_t\right)^{-1/2}e^{-\mathcal{C}A_t} \right]^2 dt\\
		&\leq & T^n + C_1 \int_T^\infty t^{n-1}\left[ \left(\mathcal{C}t^{s+1}\right)^{-1/2}e^{-\mathcal{C}t^{s+1}} \right]^2 dt\\
		&\leq & T^n + C_2 \int_T^\infty t^{n-s-2}e^{-2\mathcal{C}t^{s+1}} dt\\
		&=& T^n + C_3 \int_T^\infty u^{(n-2(s+1))/(s+1)}(u/2\mathcal{C})e^{-u}du\\
		&=& T^n + C_4 \int_T^\infty u^{n/(s+1)-2}e^{-u}du\less\infty.
		\end{eqnarray*}
		\item If $\lambda=\mu$, there are finite constants $C_1$, $C_2$ and $C_3$ such that,
		\begin{eqnarray*}
		\Ex\left[\left(\tau_{\Hx}^1\right)^n\;|\;q_0^a=x, q_0^b=y\right] &= & n\int_0^\infty t^{n-1} \Px\left[\tau_{\Hx}^1\more t\;|\;q_0^a=x, q_0^b=y\right]dt\\
		&= & n\int_0^\infty t^{n-1}\Px\left[\tau_{\Qx}^1\more A_t\;|\;q_0^a=x, q_0^b=y\right]dt\\
		&=&C_1+n\int_T^\infty t^{n-1}\Px\left[\tau_{\Qx}^1\more A_t\;|\;q_0^a=x, q_0^b=y\right]dt\\
		&= & C_1 + C_2 \int_T^\infty \dfrac{xy}{\lambda^2\pi} t^{n-1}\dfrac{1}{A_T} dt\\
		&= & C_1 + C_2 \int_T^\infty \frac{xy}{\lambda^2\pi}t^{n-s-2}dt\\
		&= &  C_3\indicator{n<s+1}+\infty\indicator{n\geq s+1}.
		\end{eqnarray*}		
	\end{itemize}

\item $\alpha_t\sim k/t$ for some $k>0$. Then $A_t\sim k\log(t)$.
		\begin{itemize}
		\item If $\lambda\less\mu$, for sufficiently large $T$, there are finite constants $\mathcal{C}$, $C_1$, $C_2$ and $C_3$ such that,
		\begin{eqnarray*}
		\Ex\left[\left(\tau_{\Hx}^1\right)^n\;|\;q_0^a=x, q_0^b=y\right] &= & n\int_0^\infty t^{n-1}\Px\left[\tau_{\Hx}^1\more t\;|\;q_0^a=x, q_0^b=y\right]dt\\
		&= & n\int_0^\infty t^{n-1}\Px\left[\tau_{\Qx}^1\more A_t\;|\;q_0^a=x, q_0^b=y\right]dt\\
		&= & C_1 + n\int_T^\infty t^{n-1} \left[A_t^{-1/2}e^{-\mathcal{C}A_t}\right]^2dt\\
		&= & C_1 + n\int_T^\infty t^{n-1} \left[(k\log(t))^{-1/2}e^{-k\mathcal{C}\log(t)}\right]^2dt\\
		&= & C_1 + C_2\int_T^\infty k t^{n-1-2k\mathcal{C}}\log(t) dt\\
		&= &  C_3(\I\{n<2\mathcal{C}k\}+ \infty(\I\{n\geq2\mathcal{C}k\}).
		\end{eqnarray*}
		\item If $\lambda=\mu$, there are finite constants $C_1$, $C_2$ and $C_3$ such that,
		\begin{align*}
		\Ex\left[\left(\tau_{\Hx}^1\right)^n\;|\;q_0^a=x, q_0^b=y\right] &=n\int_0^\infty t^{n-1}\Px\left[\tau_{\Qx}^1\more A_t\;|\;q_0^a=x, q_0^b=y\right]dt\\
		&=C_1 + C_2 \int_T^\infty \dfrac{xy}{\lambda^2\pi} t^{n-1}\dfrac{1}{A_T} dt\\
		&=C_1 + C_2 \int_T^\infty \dfrac{xy}{\lambda^2\pi} t^{n-1}\dfrac{1}{k\log(t)} dt\\\\
		&=C_1 + C_3 \int_T^\infty \frac{t^{n-1}}{\log(t)} dt=\infty.
		\end{align*}
	\end{itemize}

\end{itemize}
\end{proof}

\begin{proof}[Proof of Proposition \ref{prop:Distr:Sn}] \label{proof:Distr:Sn}
Let $F_{n,\Qx}(t)$ and $F_{n,\Hx}(t)$ denote the cdf of $S^n_{\Qx}$ and $S^n_{\Hx}$, respectively. Moreover, let $f_{n,\Qx}(t)$ and $f_{n,\Hx}(t)$ denote their corresponding densities. The result will be proven by induction. The base case, $n=1$ is given in Corollary \ref{cor:Distr:tau}. Assume the result is true for any $m\leq n\in\N$. Then by Corollary \ref{cor:Distr:tau} and the induction hypothesis,
\begin{equation}\label{eqn:rel:Cx:Ax}
F_{\Hx}(t)=F_{\Qx}(A_t)\qquad\qquad\text{and}\qquad\qquad f_{n,\Hx}(t)=f_{n,\Qx}(A_t)\alpha_t
\end{equation}

Furthermore, by the definition of $\tau^n$ and $S^n$ 

\begin{align} \nonumber
\Px_{x,y}[S^{n+1}_{\Hx}\leq t]&=\Px_{x,y}[S^{n}_{\Hx}\leq t, \tau^{n+1}_{\Hx}\leq t-S^{n}_\Qx]\\ \nonumber
		&=\int_0^{t} \Px_{x,y}[S^{n}_{\Hx}\leq t, \tau^{n+1}_{\Hx}\leq t-S^{n}_{\Hx}\;|\;S^{n}_{\Hx}=u]\Px_{x,y}[S^{n}_{\Hx}=u]du\\ \label{eqn:proof:Sn1}
		&=\int_0^{t} \Px[\tau^{n+1}_{\Hx}\leq t-S^{n}_{\Hx}\;|\;S^{n}_{\Hx}=u]\Px_{x,y}[S^{n}_{\Hx}=u]du\\ \label{eqn:proof:Sn2}
		&=\int_0^{t} \Px[\tau^{n+1}_{\Hx}\leq t-u]f_{n,\Hx}(u)du\\ \nonumber
		&=\int_0^{t} \Px[\tau^{n+1}_{\Qx}\leq A_{t-u}]f_{n,\Qx}(A_u)\alpha_u du\\ \nonumber
		&=\int_0^{t} F_{1,\Qx}(A_t-A_u)f_{n,\Qx}(A_u)\alpha_u du\\ \nonumber
		&=\int_0^{A_t} F_{1,\Qx}(A_t-u)f_{n,\Qx}(u) du\\ \nonumber
		&=\int_0^{A_t} F_{1,\Qx}(A_t-u)dF_{n,\Qx}(u)\\ \nonumber
		&=\Px_{x,y}[S^{n+1}_{\Qx}\leq A_t],
\end{align}
In the last equality we used  the facts that $S_{n+1}=S_n+\tau_{n+1}$ and that for $X$ and $Y$, non-negative independent random variables,
\[
F_{X+Y}(t)=\Px[X+Y\leq t]=F_X*F_Y(t)=\int_0^tF_X(t-x)dF_Y(x),
\]
with $F_X$ and $F_Y$ denoting the cdfs of $X$ and $Y$.
\end{proof}

\begin{proof}[Proof of Theorem \ref{thm:main:part1:1}] By the dynamics of the order book described in Section \ref{Sec:Properties}, the sequence of random price changes $X_i\in\{-1,1\}$ is independent. Then, if $\lambda<\mu$ and,

\begin{itemize}
\item If $\alpha_t\sim t^{s}\log^m(t)$ for $s\neq-1, m\geq0$ or if $\dfrac{\alpha_t}{t^{-1}} \to K$ as $t\to\infty$, with $2\mathcal{C}K> 1$, by Lemma \ref{lemma:asympt:tau}, for every $x,y\in\N$,
\[
\Ex\left[\left(\tau_{\Hx}^1\right)^n\;\Big|\;q_0^a=x, q_0^b=y\right]\less\infty.
\]
Since $N_t=\max\{n\geq0\;|\; \tau_1+\tau_2+\ldots+\tau_n\leq t\}$ then,
\[
\tau_1+\tau_2+\ldots+\tau_{N_t}\leq t \leq \tau_1+\tau_2+\ldots+\tau_{N_t+1}.
\]
Dividing the previous inequality by $N_t$, since $N_t\to\infty$ as $t\to\infty$, by using the Strong Law of Large Numbers we obtain that a.s.
\[
\Ex[\tau]:=\sum\limits_{x,y\in\N} \Ex\left[\left(\tau_{\Hx}^1\right)^n\;\Big|\;q_0^a=x, q_0^b=y\right]f(x,y)\to \frac{t}{N_t}\qquad\qquad\text{ as }t\to\infty.
\]
 Therefore, by using the sequence $t_n=tn$, we decompose the process $s_{t_n}:=\sum\limits_{j=1}^{N_{t_n}}X_i$ as:
\small
\[
s_{t_n}= \underbrace{\frac{s_0}{\sqrt{n}}}_{\hbox{I}_n}+\underbrace{\frac{1}{\sqrt{n}}\sum\limits_{j=1}^{[tn/\Ex[\tau_1]]}\left(X_j\right)}_{\hbox{II}_n} + \underbrace{\left(\frac{1}{\sqrt{n}}\sum\limits_{j=1}^{N_{t_n}}X_j-\frac{1}{\sqrt{n}}\sum\limits_{j=1}^{[tn/\Ex_{\pi}(\tau_1)]}X_j\right)}_{\hbox{III}_n}
\]
\normalsize
\noindent
As $n\rightarrow\infty$, clearly, I$_n\Rightarrow0$.
Also, by Donsker's Invariance principle,
\begin{align*}
\hbox{II}_n&\Rightarrow  \sigma W_t,
\end{align*}
where $\sigma$ is a constant. Now, since $X_j\in\left\{1,-1\right\}$, for any $\epsilon\more0$,
\begin{align*}
\Px\left(\left|\sum\limits_{j=1}^{N_{t_n}}X_j-\sum\limits_{j=1}^{[tn/\Ex[\tau_1]]}\Pi_j\right|\geq\epsilon\sqrt{n}\right)&\leq\Px\left(\left|\sum\limits_{j=N_{t_n}\wedge[tn/\Ex[\tau_1]]}^{N_{t_n}\vee[tn/\Ex[\tau_1]]}X_j\right|\geq\epsilon\sqrt{n}\right)\\
	&\leq \Px\left(\frac{1}{2}\left|N_{t_n}-[tn/\Ex[\tau_1]]\right|\geq\epsilon\sqrt{n}\right)\\
	&\leq \Px\left(\left|\frac{N_{t_n}}{[tn/\Ex[\tau_1]]}-1\right|\geq\frac{2\epsilon\sqrt{n}}{[tn/\Ex[\tau_1]]}\right),
\end{align*}
which converges to 0 as $n\rightarrow\infty$. Thus, III$_n$ converges to 0 in probability and we conclude the proof.
		
		\item If $\dfrac{\alpha_t}{t^{-1}} \sim K$ with $2\mathcal{C}K\leq 1$ as $t\to\infty$, then $A_T\sim k\log(T)$. By Lemma \ref{lemma:asympt:tau}
\begin{align*}
\Px\left[\tau_{\Hx}^1\more T\;\Big|\;q_0^a=x, q_0^b=y\right] &\sim \left(\frac{\mu}{\lambda}\right)^{(x+y)/2} \frac{xy}{\pi\mathcal{C}^2\sqrt{\lambda\mu}} \frac{\exp(-2\mathcal{C}k\log(T))}{k\log(T)}\\
			&\sim \left(\frac{\mu}{\lambda}\right)^{(x+y)/2} \frac{xy}{\pi\mathcal{C}^2\sqrt{\lambda\mu}} \frac{T^{-2\mathcal{C}k}}{k\log(T)}\\
			&\sim \left(\frac{\mu}{\lambda}\right)^{(x+y)/2} \frac{xy}{\pi\mathcal{C}^2\sqrt{\lambda\mu}} \frac{1}{kT^{2\mathcal{C}k}\log(T)}
\end{align*}
Therefore,
\begin{align*}
n\Px\left[\tau_{\Hx}^1\more n^{1/2k\mathcal{C}}\;\Big|\;q_0^a=x, q_0^b=y\right]&\sim n\left[\left(\frac{\mu}{\lambda}\right)^{(x+y)/2} \frac{xy}{\pi\mathcal{C}^2\sqrt{\lambda\mu}} \frac{1}{kn\log(n^{1/2k\mathcal{C}})}\right]\\
				&\sim \left(\frac{\mu}{\lambda}\right)^{(x+y)/2} \frac{xy}{\pi\mathcal{C}^2\sqrt{\lambda\mu}} \frac{1}{k\log(n^{1/2k\mathcal{C}})}.
\end{align*}
Thus,
\[
n\Px\left[\tau_{\Hx}^1\more n^{1/2k\mathcal{C}}\;\Big|\;q_0^a=x, q_0^b=y\right]\to0\qquad\qquad\text{ as } n\to\infty.
\]
and by Theorem 6.4.2 in \cite{Gut:2013}, in probability,
\[
\frac{S_n-n\Ex\left[\tau_{\Hx}^1\indicator{\tau_{\Hx}^1<n^{1/2k\mathcal{C}}}\right]}{n^{1/2k\mathcal{C}}}\to 0\qquad\qquad\text{ as } n\to\infty.
\]
By Proposition \ref{prop:Second:moment:truncated1}, $\hat{\mathcal{A}}:=\lim_{n\to\infty}\frac{n\Ex[\tau\I\{\tau<n^{1/2k\mathcal{C}}]}{n^{1/2k\mathcal{C}}}$ is a constant. Thus, by a similar argument as in the previous bullet, in probability,
\[
\frac{t}{N_t^{1/2k\mathcal{C}}}\to\hat{\mathcal{A}}  \qquad\qquad\text{ as } t\to\infty,
\]
or equivalently, 
\[
N_t\to\left(\frac{t}{\hat{\mathcal{A}}}\right)^{2k\mathcal{C}}  \qquad\qquad\text{ as } t\to\infty.
\]	
As before, by using the sequence $t_n=tn^{1/2k\mathcal{C}}$, we decompose the process $s_{t_n}:=\sum\limits_{j=1}^{N_{t_n}}X_i$ as:
\small
\[
s_{t_n}= \underbrace{\frac{s_0}{\sqrt{n}}}_{\hbox{I}_n}+\underbrace{\frac{1}{\sqrt{n}}\sum\limits_{j=1}^{[n(t/\hat{\mathcal{A}})^{{2k\mathcal{C}}}]}\left(X_j\right)}_{\hbox{II}_n} + \underbrace{\left(\frac{1}{\sqrt{n}}\sum\limits_{j=1}^{N_{t_n}}X_j-\frac{1}{\sqrt{n}}\sum\limits_{j=1}^{[n(t/\hat{\mathcal{A}})^{{2k\mathcal{C}}}]}X_j\right)}_{\hbox{III}_n}
\]
\normalsize
\noindent
By similar arguments as above, as $n\rightarrow\infty$,
\begin{align*}
\hbox{I}_n&\Rightarrow0\\
\hbox{III}_n&\Rightarrow0\\
\hbox{II}_n&\Rightarrow W_{(t/\hat{\mathcal{A}})^{{2k\mathcal{C}}}}= W_{\hat{\mathcal{A}}^{-{2k\mathcal{C}}}\int_0^t \frac{1}{u^{1-2k\mathcal{C}}}du}.
\end{align*}
Moreover, in distribution,
\[
W_{(t/\hat{\mathcal{A}})^{{2k\mathcal{C}}}}=\hat{\mathcal{A}}^{-k\mathcal{C}}\int_0^t \sqrt{\frac{1}{u^{1-2\mathcal{C}k}}}dW_u,
\]
which concludes the proof.
	\end{itemize}

\end{proof}

\begin{proof}[Proof of Theorem \ref{thm:main:part1:2}] By the dynamics of the order book described in Section \ref{Sec:Properties}, the sequence of random price changes $X_i\in\{-1,1\}$ is independent. Then, if $\lambda=\mu$,

	\begin{itemize}
		\item If $\alpha_t\sim t^{s}\log^m(t)$ as $t\to\infty$, for any $s\more 0, m\geq0$, by Lemma \ref{lemma:asympt:tau}, for every $x,y\in\N$,
\[
\Ex\left[\left(\tau_{\Hx}^1\right)^n\;\Big|\;q_0^a=x, q_0^b=y\right]\less\infty.
\]
		and the proof follows in the same way as in the proof of Theorem \ref{thm:main:part1:1}.
		
		\item If $\alpha_t\sim t^{-1+s}$ as $t\to\infty$ for any $s\in(0,1]$, then $A_t\sim t^s/s$ and by Lemma \ref{lemma:asympt:tau}
		
\begin{align*}
\Px\left[\tau_{\Hx}^1\more T\;\Big|\;q_0^a=x, q_0^b=y\right] &\sim \dfrac{xy}{\lambda^2\pi} \dfrac{s}{T^s}
\end{align*}

Therefore,
\begin{align*}
n\Px\left[\tau_{\Hx}^1\more n^{1/s}\log(n)\;\Big|\;q_0^a=x, q_0^b=y\right]&\sim n\left[\dfrac{xy}{\lambda^2\pi} \dfrac{s}{n\log^s(n)}\right]\\
				&\sim \dfrac{xy}{\lambda^2\pi} \dfrac{s}{\log^s(n)}.
\end{align*}
Thus, since $s>0$,
\[
n\Px\left[\tau_{\Hx}^1\more n^{1/s}\log(s)\;\Big|\;q_0^a=x, q_0^b=y\right]\to0\qquad\qquad\text{ as } n\to\infty.
\]
and by Theorem 6.4.2 in \cite{Gut:2013}, in probability,
\[
\frac{S_n-n\Ex\left[\tau_{\Hx}^1\I\{\tau_{\Hx}^1<n^{1/s}\log(n)\right]}{n^{1/s}\log(n)}\to 0\qquad\qquad\text{ as } n\to\infty,
\]
where by Proposition \eqref{prop:Second:moment:truncated2} $\hat{\mathcal{B}}:=\lim_{n\to\infty}\frac{n\Ex[\tau\I\{\tau<n^{s}\log(n)]}{n^{s}\log(n)}$ is a constant. Thus, by a similar argument as in the proof of Theorem \ref{thm:main:part1:1}, in probability,
\[
\frac{t}{N_t^{s}\ln(N_t)}\sim\hat{\mathcal{B}}  \qquad\qquad\text{ as } t\to\infty,
\]
or what is the same, 
\begin{equation}\label{eqn:leqmu:case2:Nt}
N_t^s\log(N_t)\sim\frac{t}{\hat{\mathcal{B}}} \qquad\qquad\text{ as } t\to\infty.
\end{equation}
Let $\psi(t)$ be the inverse function of $f(t):=t^s\log(t)$. Notice that $\psi$ is well defined since $f$ is strictly increasing. By definition, $\psi(t)=u$ implies that $f(u)=t$ or, what is the same, $u^s\log(u)=t$, or $\psi(t)^s\log(\psi(t))=t$.Thus,
\begin{align*}
\psi(t)&\sim \frac{t^{1/s}}{\log(\psi(t))}\\
			 &\sim \frac{t^{1/s}}{(1/s)\log(\psi(t)^s)}\\
			 &\sim \frac{st^{1/s}}{\log(\psi(t)^s\log(\psi(t)))}\\
			 &\sim \frac{st^{1/s}}{\log(t)}
\end{align*}

Therefore, by Eq. \ref{eqn:leqmu:case2:Nt},
\[
f(N_t)\sim \frac{t}{\hat{\mathcal{B}}}
\]
and since $\psi$ is the inverse of $f$, then
\[
N_t\sim \psi\left(\frac{t}{\hat{\mathcal{B}}}\right) \sim \frac{s\left(\frac{t}{\hat{\mathcal{B}}}\right)^{1/s}}{\log\left(\frac{t}{\hat{\mathcal{B}}}\right)}\sim \frac{s}{\hat{\mathcal{B}}^{1/s}}\frac{t^{1/s}}{\log\left(t\right) -\log(\hat{\mathcal{B}})}
\]

Using the sequence $t_n=t(n\log(n))^s$, we have that,
\begin{align*}
N_{t_n}&\sim \frac{s}{\hat{\mathcal{B}}^{1/s}}\frac{(tn^{s}\log^s(n))^{1/s}}{\log\left(tn^{s}\log^s(n)\right) -\log(\hat{\mathcal{B}})}\\
&\sim \frac{s}{\hat{\mathcal{B}}^{1/s}} \frac{t^{1/s}n\log(n)}{\log\left(tn^{s}\log^s(n)\right) -\log(\hat{\mathcal{B}})}\\
&\sim \frac{s}{\hat{\mathcal{B}}^{1/s}} \frac{t^{1/s}n}{\frac{\log\left(tn^{s}\log^s(n)\right)}{\log(n)} -\frac{\log(\hat{\mathcal{B}})}{\log(n)}}\\
& \sim \frac{s}{\hat{\mathcal{B}}^{1/s}} \frac{t^{1/s}n}{s}\sim \frac{t^{1/s}n}{\hat{\mathcal{B}}^{1/s}}
\end{align*}
To conclude, we decompose the process $s_{t_n}:=\sum\limits_{j=1}^{N_{t_n}}X_i$ as in the proof of Theorem \ref{thm:main:part1:1} and use the same arguments therein.

		\item If $\alpha_t\sim t^s\log^m(t)$ for every  $s<0$ or if $\alpha_t\sim k/t$, then for any regularly varying sequence at infinity with exponent $1/\rho$ for some $\rho\in(0,1]$,
\[
n\Px\left[\tau_{\Hx}^1\more b_n\;\Big|\;q_0^a=x, q_0^b=y\right]\to\infty  \qquad\qquad\text{ as } n\to\infty.
\]
Thus, $N_t$ cannot be rescaled to ensure a Law of Large Numbers and the price process does not converge.
	\end{itemize}

\end{proof}

\section{Auxiliary Results}

\begin{prop} \label{prop:Second:moment:truncated1}
Let $\tau$ be positive random variable such that $\Px[\tau>t]\sim \frac{\Theta}{t^{2k\mathcal{C}}\log(t)}$ with $2k\mathcal{C}\leq1$, $\Theta$ a constant and define the sequence 
\[
X_n=\frac{n}{n^{1/2k\mathcal{C}}}\tau\indicator{\tau<n^{1/2k\mathcal{C}}}.
\]
Then, $\Psi_n:=\Ex[X_n]$ converges as $n\to\infty$.
\end{prop}

\begin{proof}

Let $F_\tau(t)$ denote the CDF of $\tau$. Then, since $\tau$ is a positive random variable,
\[
\Ex[g(\tau)]=\int_0^\infty g(x)dF_\tau(x),
\]
where the last integral should be understood as a Riemann-Stieltjes integral. Letting $g(x)=\frac{n}{n^{1/2k\mathcal{C}}}x\indicator{x<n^{1/2k\mathcal{C}}}$ and substituting in the above formula,
\begin{align*}
\Ex\left[\frac{n}{n^{1/2k\mathcal{C}}}\tau\indicator{\tau<n^{1/2k\mathcal{C}}}\right]&=\frac{n}{n^{1/2k\mathcal{C}}}\int_0^{n^{1/2k\mathcal{C}}}xdF_\tau(x)\\
		&=-\frac{n}{n^{1/2k\mathcal{C}}}\int_0^{n^{1/2k\mathcal{C}}}xd(1-F_\tau(x))\\
\end{align*}
Noting that the Left Hand Side is $\Psi_n$ and integrating by parts, 
\begin{align}\label{bound:Psin}\nonumber
\Psi_n&=-\frac{n}{n^{1/2k\mathcal{C}}}\cdot x\Px[\tau>x]\Big|_{x=0}^{n^{1/2k\mathcal{C}}} +\frac{n}{n^{1/2k\mathcal{C}}}\int_{0}^{n^{1/2k\mathcal{C}}}\Px[\tau>t]dt\\ \nonumber
	&=-n\Px\left[\tau>n^{1/2k\mathcal{C}}\right] +n^{1-1/2k\mathcal{C}}\int_{0}^{n^{1/2k\mathcal{C}}}\Px[\tau>t]dt\\ \nonumber
	&\sim -\left.n\frac{\Theta}{t^{2k\mathcal{C}}\log(t)}\right|_{t=n^{1/2k\mathcal{C}}}+n^{1-1/2k\mathcal{C}}\int_{0}^{n^{1/2k\mathcal{C}}}\frac{\Theta}{t^{2k\mathcal{C}}\log(t)}dt\\
	&\sim -\frac{\Theta}{2k\mathcal{C}\log(n)}+n^{1-1/2k\mathcal{C}}\int_{0}^{n^{1/2k\mathcal{C}}}\frac{\Theta}{t^{2k\mathcal{C}}\log(t)}dt 
\end{align}
Then, by using L'Hopital rule and considering $\Psi_n$ as a (differentiable) function of $n$ and setting $\widehat{\Theta}=\Theta/2k\mathcal{C}$,
\[
\frac{d}{dn}\Psi_n\sim \frac{\widehat{\Theta}}{n\log^2(n)} + n^{-1/2k\mathcal{C}}\int_{0}^{n^{1/2k\mathcal{C}}}\frac{\Theta}{t^{2k\mathcal{C}}\log(t)}dt + n^{1-1/2k\mathcal{C}} \frac{\widehat{\Theta}}{n\log(n)}
\]
which is clearly positive and thus $\Psi_n$ is an increasing sequence. Furthermore, by \eqref{bound:Psin}, there exist constants $\epsilon$ and $T_\epsilon$ such that if $n^{1/2k\mathcal{C}}>T_\epsilon$,
\[
\Psi_n\leq -\frac{\Theta}{2k\mathcal{C}\log(n)}+n^{1-1/2k\mathcal{C}}\int_{0}^{T_\epsilon}\frac{\Theta}{t^{2k\mathcal{C}}\log(t)}dt  + n^{1-1/2k\mathcal{C}}\int_{T_\epsilon}^{n^{1/2k\mathcal{C}}}\frac{\Theta}{t^{2k\mathcal{C}}\log(t)}dt
\]
and since $2k\mathcal{C}<1$, then $1-1/2\mathcal{C}\leq0$ and thus
\begin{align*}
\Psi_n&\leq -\frac{\Theta}{2k\mathcal{C}}+\int_{0}^{T_\epsilon}\frac{\Theta}{t^{2k\mathcal{C}}\log(t)}dt+ n^{1-1/2k\mathcal{C}}\int_{T_\epsilon}^{n^{1/2k\mathcal{C}}}\frac{\Theta}{t^{2k\mathcal{C}}}dt\\
	&=-\frac{\Theta}{2k\mathcal{C}}+\int_{0}^{T_\epsilon}\frac{\Theta}{t^{2k\mathcal{C}}\log(t)}dt+ n^{1-1/2k\mathcal{C}}\frac{n^{-1+1/2k\mathcal{C}} + C(T_\epsilon)}{-2k\mathcal{C}+1}\\
	&\leq -\frac{\Theta}{2k\mathcal{C}}+\int_{0}^{T_\epsilon}\frac{\Theta}{t^{2k\mathcal{C}}\log(t)}dt+ \frac{1}{1-2k\mathcal{C}} + \frac{C(T_\epsilon)}{1-2k\mathcal{C}}\\
	&<\infty
\end{align*}
Therefore, $\Psi_n$ is also bounded and by the monotone convergence theorem for sequences it converges to a constant limit.

\end{proof}

\begin{prop} \label{prop:Second:moment:truncated2}
Let $\tau$ be positive random variable such that $\Px[\tau>t]\sim \frac{\Theta}{t^s}$ with $0<s\leq1$ with $\Theta$ a constant. Define the sequence 
\[
X_n=\frac{n}{n^{1/s}\log(n)}\tau\indicator{\tau<n^{1/s}\log(n)}.
\]
Then, $\Phi_n:=\Ex[X_n]$ converges as $n\to\infty$.
\end{prop}

\begin{proof}

As in the previous result, let $F_\tau(t)$ denote the CDF of $\tau$. Then, again, since $\tau$ is a positive random variable,
\[
\Ex[g(\tau)]=\int_0^\infty g(x)dF_\tau(x).
\]
Substituting $g(x)=\frac{n}{n^{1/s}\log(n)}x\indicator{x<n^{1/s}\log(n)}$ in the above formula,
\begin{align*}
\Ex\left[\frac{n}{n^{1/s}\log(n)}\tau\indicator{\tau<n^{1/s}\log(n)}\right]&=\frac{n}{n^{1/s}\log(n)}\int_0^{n^{1/s}\log(n)}xdF_\tau(x)\\
		&=-\frac{n}{n^{1/s}\log(n)}\int_0^{n^{1/s}\log(n)}xd(1-F_\tau(x))\\
\end{align*}
Noting that the Left Hand Side is $\Phi_n$ and integrating by parts, 
\begin{align*}\label{bound:Phin}\nonumber
\Phi_n&=-\frac{n}{n^{1/s}\log(n)}\cdot x\Px[\tau>x]\Big|_{x=0}^{n^{1/s}\log(n)} +\frac{n}{n^{1/s}\log(n)}\int_{0}^{n^{1/s}\log(n)}\Px[\tau>t]dt\\ \nonumber
	&=-n\Px\left[\tau>n^{1/s}\log(n)\right] +\frac{n^{1-1/s}}{\log(n)}\int_{0}^{n^{1/s}\log(n)}\Px[\tau>t]dt\\ \nonumber
	&\sim -\left.n\frac{\Theta}{t^{s}}\right|_{t=n^{1/s}\log(n)}+\frac{n^{1-1/s}}{\log(n)}\int_{0}^{n^{1/s}\log(n)}\frac{\Theta}{t^s}dt\\ \nonumber
	&\sim -\frac{\Theta}{\log^s(n)}+\frac{n^{1-1/s}}{\log(n)}\int_{0}^{n^{1/s}\log(n)}\frac{\Theta}{t^s}dt\\
	&\sim -\frac{\Theta}{\log^s(n)}+\left.\frac{n^{1-1/s}}{\log(n)}\cdot\frac{\Theta t^{-s+1}}{-s+1}\right|_{t=0}^{n^{1/s}\log(n)}\\
	&\sim -\frac{\Theta}{\log^s(n)}+\frac{n^{1-1/s}}{\log(n)}\cdot\frac{\Theta n^{-1+1/s}\log^{-s+1}(n)}{-s+1}\\
	&\sim -\frac{\Theta}{\log^s(n)}+\frac{\Theta }{(-s+1)\log^{s}(n)}\\
	&\sim \frac{(2-s)\Theta}{(1-s)\log^s(n)}
\end{align*}
which shows that it is asymptotically decreasing to zero and thus $\Phi_n$ converges by using similar arguments as in the previous theorem.

\end{proof}

\end{document}